\documentclass[a4paper,10pt]{article}
\usepackage{amsthm,amsmath,amssymb}
\usepackage{subcaption,sidecap}
\usepackage{paralist}
\usepackage[usenames,dvipsnames]{color}
\usepackage[breaklinks,colorlinks,
    citecolor={blue}, linkcolor={blue},urlcolor=Maroon]{hyperref}  
\usepackage{tikz,pgfkeys}
\usepackage{tkz-graph}
\usepackage{graphicx}
\usetikzlibrary{arrows,shapes}
\usepackage{charter,eulervm}
\usepackage{tabularx}
\usepackage[final,expansion=alltext,protrusion=true]{microtype}
\newtheorem{theorem}{Theorem}[section]
\newtheorem{lemma}[theorem]{Lemma}
\newtheorem{corollary}[theorem]{Corollary}
\newtheorem{proposition}[theorem]{Proposition}

\newtheorem{claim}{Claim}
\newtheorem{definition}{Definition}
\def\boxit#1{\vbox{\hrule\hbox{\vrule\kern4pt
  \vbox{\kern1pt#1\kern1pt}
\kern2pt\vrule}\hrule}}

\newcommand{\lig}{locally interval graph}
\newcommand{\msm}{maximal strong module}
\newcommand{\lp}[1]{\ensuremath{{\mathtt{lp}(#1)}}}
\newcommand{\rp}[1]{\ensuremath{{\mathtt{rp}(#1)}}}
\newcommand{\lpp}[2]{\ensuremath{{\mathtt{lp{#1}}(#2)}}}
\newcommand{\rpp}[2]{\ensuremath{{\mathtt{rp{#1}}(#2)}}}
\newcommand{\lint}[1]{\ensuremath{{\mathtt{left}(#1)}}}
\newcommand{\rint}[1]{\ensuremath{{\mathtt{right}(#1)}}}
\newcommand{\head}[1]{\ensuremath{{\mathtt{last}(#1)}}}
\newcommand{\tail}[1]{\ensuremath{{\mathtt{first}(#1)}}}
\newcommand{\stpath}[2]{$#1$-$#2$ path}
\newcommand{\stsep}[2]{$#1$-$#2$ separator}

\newcommand{\cG}{\ensuremath{{\cal G}}}
\newcommand{\cF}{\ensuremath{{\cal F}}}
\newcommand{\oo}{\ensuremath{T}}
\newcommand{\ec}{\ensuremath{E_{\text{c}}}}
\newcommand{\ecc}{\ensuremath{E_{\text{cc}}}}
\newcommand{\oc}{\ensuremath{T_{\text{c}}}}
\newcommand{\occ}{\ensuremath{T_{\text{cc}}}}
\newcommand{\lc}{\ensuremath{L_{\text{c}}}}
\newcommand{\lcc}{\ensuremath{L_{\text{cc}}}}

\newcommand{\og}[1]{\ensuremath{\phi(#1)}}
\newcommand{\hv}[1]{\ensuremath{H\langle #1\rangle}}

\newcommand{\misp}{minimum interval supergraph}
\newcommand{\misb}{maximum spanning interval subgraph}
\newcommand{\miib}{maximum induced interval subgraph}

\newcommand{\comment}[1]{\hfill $\setminus\!\!\setminus$ {\em #1}}
\newcommand{\myref}[1]{\hyperref[can:#1]{\tt C{#1}}}
\title{Linear Recognition of Almost Interval Graphs\thanks{Dedicated
 to Jianer Chen on the occasion of his 60th birthday.}  }

\author{{\sc Yixin Cao}\thanks{Department of Computing, Hong Kong
    Polytechnic University, Hong Kong, China,
    \href{mailto:yixin.cao@polyu.edu.hk}{\tt yixin.cao@polyu.edu.hk}.
    Work done in part while the author was at Institute for Computer
    Science and Control, Hungarian Academy of Sciences, when he was
    supported by the European Research Council (ERC) under the grant
    280152 and the Hungarian Scientific Research Fund (OTKA) under the
    grant NK105645.  }}

\date{}

\begin{document}
\maketitle

\begin{abstract}
  Let $\mbox{interval} + k v$, $\mbox{interval} + k e$, and
  $\mbox{interval} - k e$ denote the classes of graphs that can be
  obtained from some interval graph by adding $k$ vertices, adding $k$
  edges, and deleting $k$ edges, respectively.  When $k$ is small,
  these graph classes are called almost interval graphs.  They are
  well motivated from computational biology, where the data ought to
  be represented by an interval graph while we can only expect an
  almost interval graph for the best.  For any fixed $k$, we give
  linear-time algorithms for recognizing all these classes, and in the
  case of membership, our algorithms provide also a specific interval
  graph as evidence.  When $k$ is part of the input, these problems
  are also known as graph modification problems, all NP-complete.  Our
  results imply that they are fixed-parameter tractable parameterized
  by $k$, thereby resolving the long-standing open problem on the
  parameterized complexity of recognizing $\mbox{interval}+ k e$,
  first asked by Bodlaender et al.\ [Bioinformatics, 11:49--57, 1995].
  Moreover, our algorithms for recognizing $\mbox{interval}+ k v$ and
  $\mbox{interval}- k e$ run in times $O(6^k \cdot (n + m))$ and
  $O(8^k \cdot (n + m))$, (where $n$ and $m$ stand for the numbers of
  vertices and edges respectively in the input graph,) significantly
  improving the $O(k^{2k}\cdot n^3m)$-time algorithm of Heggernes et
  al.~[STOC 2007] and the $O(10^k \cdot n^9)$-time algorithm of Cao
  and Marx [SODA 2014] respectively.
\end{abstract}

\thispagestyle{empty}
\setcounter{page}{0}
\newpage
 \section{Introduction}\label{sec:introduction}
A graph is an \emph{interval graph} if its vertices can be assigned to
intervals on the real line such that there is an edge between two
vertices if and only if their corresponding intervals intersect.  This
set of intervals is called an \emph{interval model} for the graph.
The study of interval graphs has been closely associated with
(computational) biology
\cite{benzer-59-topology-genetic-structure,waterman-86-Interval-graphs-and-DNA}.
For example, in {\em physical mapping} of DNA, which asks for
reconstructing the relative positions of clones along the target DNA
based on their pairwise overlap information
\cite{lander-88-genomic-mapping,arratia-91-genomic-mapping}, the input
data can be easily represented by a graph, where each clone is a
vertex, and two clones are adjacent if and only if they overlap
\cite{waterman-86-Interval-graphs-and-DNA,zhang-94-physical-mapping,karp-93-mapping-genomes},
hence an interval graph.  A wealth of literature has been devoted to
algorithms on interval graphs, which include a series of linear-time
recognition algorithms
\cite{booth-76-test-c1p,mohring-85-interval-algorithmic-aspects,korte-89-recognizing-interval-graphs,hsu-92-simple-test-interval-graphs,hsu-99-recognizing-interval-graphs,
  habib-00-LBFS-and-partition-refinement,
  corneil-09-lbfs-strucuture-and-interval-recognition}.  Ironically,
however, these recognition algorithms are never used as they are
intended to be.  Biologists never need to roll up their sleeves and
feed their data into any recognition algorithm before claiming the
answer is ``NO'' with full confidence, i.e., their data would not give
an interval graph though they ought to.  The reason is that biological
data, obtained by mainly experimental methods, are destined to be
flawed.

More often than not, biologists are also confident that their data,
though not perfect, are of reasonably good quality: there are only few
errors hidden in the data \cite{lander-88-genomic-mapping}.  This
leads us naturally to consider graphs that are not interval graphs,
but close to one in some sense.  We say that a graph is an
\emph{almost interval graph} if it can be obtained from an interval
graph by a small amount of modifications; it may or may not be an
interval graph itself.  Different applications are afflicted with
different types of errors, e.g., there might be outliers,
false-positive overlaps, and/or false-negative overlaps.  We can
accordingly define different measures for \emph{closeness}.  For any
given nonnegative integer $k$, we use $\mbox{interval}+ k v$,
$\mbox{interval}+ k e$, and $\mbox{interval} - k e$ to denote the
classes of graphs that can be obtained from some interval graph by
adding at most $k$ vertices, adding at most $k$ edges, and deleting at
most $k$ edges, respectively.\footnote{Here we use ``at most'' instead
  of ``precisely'' for both practical and theoretical reasons.
  Practically, this formulation is more natural for aforementioned
  applications, where less modifications are preferred.
  Theoretically, it allows all classes fully contain interval graphs
  itself; in particular, we allow interval$+ ke$ and interval$- ke$ to
  contain graphs with no edge and cliques, respectively.  As a matter
  of fact, one can show that except the trivial cases (i.e., the input
  graph has less than $k$ vertices, $k$ edges, or $k$ missing edges),
  if a graph can be made from an interval graph $G'$ by $k'$
  operations, where $k'<k$, then we can also obtain it from another
  interval graph $G$ by exactly $k$ operations of the same type.}  We
remark that this definition can be easily generalized to any
\emph{hereditary} graph class (i.e., closed under taking induced
subgraphs).  Interval graphs and all other graph classes to be
mentioned in this paper are hereditary
\cite{golumbic-2004-perfect-graphs,brandstadt-99-graph-classes,spinrad-03-efficient-graph-representations}.

The first task is of course to efficiently decide whether a given
graph is an almost interval graph or not, and more importantly,
identify an object interval graph if one exists.  Computationally,
finding an object interval graph is equivalent to pinpointing the few
but crucial errors in the data.  For any fixed $k$, this can be
trivially done in polynomial time: given a graph $G$ on $n$ vertices,
we can in $n^{O(k)}$ time try every subset of $k$ vertices, edges, or
missing edges of $G$.  Such an algorithm is nevertheless inefficient
even for very small $k$, as $n$ is usually large.  The main results of
this paper are linear-time recognition algorithms for all three
classes of almost interval graphs.
\begin{theorem}\label{thm:main}
  Let $k$ be any fixed nonnegative integer.  Given a graph $G$ on $n$
  vertices and $m$ edges, the membership of $G$ in each of
  $\mbox{interval} + k v$, $\mbox{interval} + k e$, and
  $\mbox{interval} - k e$ can be decided in $O(n+m)$ time.  Moreover,
  in case of affirmative, an object interval graph can be produced in
  the same time.
\end{theorem}
Thm.~\ref{thm:main} extends the line of linear-time algorithms for
recognizing interval graphs.  In the running times of all the three
algorithms, needless to say, the constants hidden by big-Oh rely on
$k$.  Since all the problems are NP-hard when $k$, instead of being
constant, is part of the input
\cite{lewis-80-node-deletion-np,kashiwabara-79-interval-completion,goldberg-95-interval-edge-deletion},
the dependence on $k$ is necessarily super-polynomial (assuming
P$\ne$NP).  Now that the linear dependence on the graph size is
already optimum, we would like to minimize the factor of $k$.  We are
thus brought into the framework of parameterized computation.  Recall
that a problem, associated with some parameter, is {\em
  fixed-parameter tractable (FPT)} if it admits a polynomial-time
algorithm where the exponent on the input size ($n + m$ in this paper)
is a global constant independent of the parameter
\cite{downey-fellows-99}.  From the lens of parameterized computation,
the recognition of almost interval graphs is conventionally defined as
graph modification problems, where the parameter is $k$, and the task
is to transform a graph to an interval graph by at most $k$
modifications \cite{cai-96-hereditary-graph-modification}.  For the
classes $\mbox{interval} + k v$, $\mbox{interval} + k e$, and
$\mbox{interval} - k e$, the modifications are vertex deletions, edge
deletions, and completions (i.e., edge additions) respectively, which
are the most commonly considered on hereditary graph classes.  The
parameterized problems are accordingly named {\sc interval vertex
  deletion}, {\sc interval edge deletion}, and {\sc interval
  completion}.  Our results can then be more specifically stated as:

\begin{theorem}\label{thm:main-2}
  Given a graph on $n$ vertices and $m$ edges and a nonnegative
  parameter $k$, the problems \textsc{interval vertex deletion},
  \textsc{ interval edge deletion}, and \textsc{interval completion}
  can be solved in time $O(8^k \cdot (n+m))$, $k^{O(k)} \cdot (n+m)$,
  and $O(6^k \cdot (n+m))$, respectively.
\end{theorem}
In particular, we show that \textsc{interval edge deletion} is FPT,
thereby resolving a long-standing open problem first asked by
Bodlaender et al.~\cite{bodlaender-95-fpt-computational-biology}.
Further, our algorithms for \textsc{interval vertex deletion} and
\textsc{interval completion} significantly improve the $O(k^{2k}\cdot
n^3m)$-time algorithm of Heggernes et
al.~\cite{villanger-09-interval-completion} and the $O(10^k \cdot
n^9)$-time algorithm of Cao and Marx \cite{cao-14-interval-deletion},
respectively.  We remark that it can also be derived an $O(m n +
n^2)$-time approximation algorithm of ratio 8 for the minimum interval
vertex deletion problem.

We feel obliged to point out that computational biologists cannot
claim all credit for the discovery and further study of interval
graphs.  Independent of \cite{benzer-59-topology-genetic-structure},
Haj{\'o}s \cite{hajos-57-interval-graphs} formulated the class of
interval graphs out of nothing but coffee.  Since its inception in
1950s, its natural structure earns itself a position in many other
applications, among which the most cited ones include jobs scheduling
in industrial engineering \cite{bar-noy-01-resource-allocation},
temporal reasoning \cite{golumbic-93-temporal-reasoning}, and
seriation in archeology \cite{kendall-69-seriation-1}.  All these
applications involve some temporal structure, which is understandable:
before the final invention of time traveling vehicles, a graph
representing relationship of temporal activities has to be an interval
graph.  With errors involved, almost interval graphs arise naturally.

\subsection{Notation}\label{sec:graph-classes}
All graphs discussed in this paper shall always be undirected and
simple.  The \emph{order} $|G|$ and \emph{size} $||G||$ of a graph $G$
are defined to be the cardinalities of its vertex set $V(G)$ and its
edge set $E(G)$ respectively.  We assume without loss of generality
that $G$ is connected and nontrivial (containing at least two
vertices); thus $|G| = O(||G||)$.  We sometimes use the customary
notation $v\in G$ to mean $v\in V(G)$, and $u\sim v$ to mean $uv\in
E(G)$.  The {degree} of a vertex $v$ is denoted by $d(v)$.  A vertex
$v$ is \emph{simplicial} if $N[v]$ induces a clique; let $SI(G)$
denote the set of simplicial vertices of $G$.  The length of a path or
a cycle is defined to be the number of edges in it.  Standard
graph-theoretical and algorithmic terminology can be found in
\cite{diestel-10,golumbic-2004-perfect-graphs}.

\begin{figure*}[t] %lb.tex
  \centering
  \begin{subfigure}[b]{0.18\textwidth}
    \centering    \includegraphics{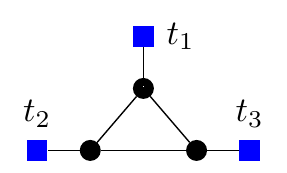} 
    \caption{net}
    \label{fig:net}
  \end{subfigure}%  
  \,
  \begin{subfigure}[b]{0.18\textwidth}
    \centering    \includegraphics{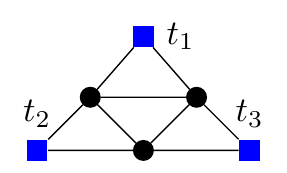} 
    \caption{sun}
    \label{fig:tent}
  \end{subfigure}%  
  \begin{subfigure}[b]{0.18\textwidth}
    \centering    \includegraphics{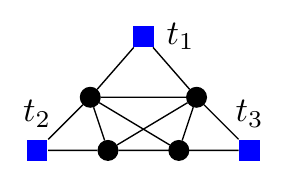} 
    \caption{rising sun}
    \label{fig:tent}
  \end{subfigure}%  
  \,
  \begin{subfigure}[b]{0.19\textwidth}
    \centering    \includegraphics{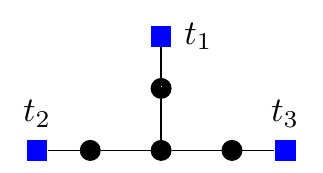} 
    \caption{long claw}
    \label{fig:long-claw}
  \end{subfigure}%  
  \,
  \begin{subfigure}[b]{0.18\textwidth}
    \centering    \includegraphics{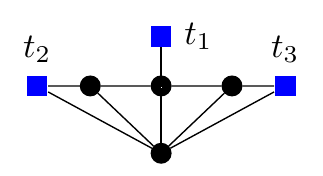} 
    \caption{whipping top}
    \label{fig:top}
  \end{subfigure}%  

  \begin{subfigure}[b]{0.4\textwidth}
    \centering    \includegraphics{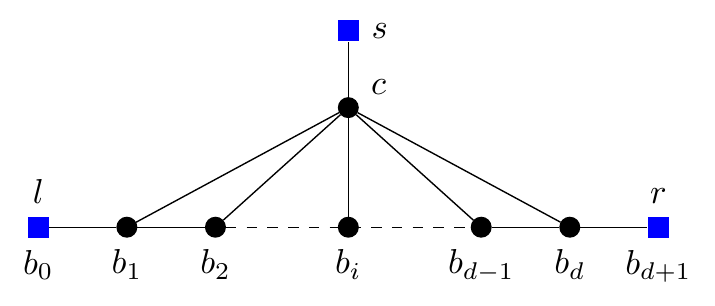} 
    \caption{$\dag_d$ ($s:c,c:l,B,r$)$\quad$ ($d = |B| \geq 3$)}
    \label{fig:dag}
  \end{subfigure}%  
  \qquad   \quad
  \begin{subfigure}[b]{0.4\textwidth}
    \centering    \includegraphics{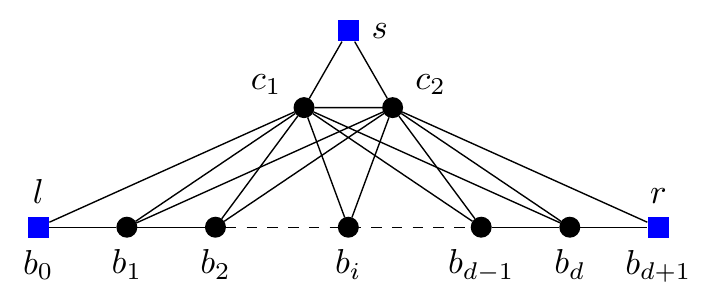} 
    \caption{$\ddag_d$ ($s:c_1,c_2:l,B,r$)$\quad$ ($d = |B| \geq 3$)}
    \label{fig:ddag}
  \end{subfigure}%  
  \caption{Minimal chordal asteroidal witnesses (squares vertices make
    the triple).}
  \label{fig:at}
\end{figure*}
A cycle induced by $d$ vertices, where $d\ge 4$, is called a
\emph{$d$-hole}, or simply a {\em hole} if $d$ is irrelevant.  In
other words, a hole is an induced cycle that is not a triangle.  A
graph is \emph{chordal} if it contains no holes.  Lekkerkerker and
Boland \cite{lekkerkerker-62-interval-graphs} showed that a graph is
an interval graph if and only if it is chordal and does not contain a
structure called \emph{asteroidal triple} (\emph{at} for short), i.e.,
three vertices such that each pair of them is connected by a path
avoiding neighbors of the third one.  They went further to list all
minimal chordal graphs that contain an at.  These graphs, reproduced
in Fig.~\ref{fig:at}, are called \emph{chordal asteroidal witnesses}
(caws for short).  

Let $\cF_I$ denote the set of minimal forbidden induced subgraphs of
interval graphs, i.e., all holes and caws.  Let $\cF_{LI}$ be the set
\{net, sun, rising sun, long claw, whipping top, $4$-hole, $5$-hole\}
(see the first row of Fig.~\ref{fig:at}).  An important ingredient of
our algorithms is a comprehensive study of the following graph class.
Clearly, $\cF_{LI}\subset \cF_I$, and thus all interval graphs satisfy
this definition. 
\begin{definition}
  \emph{Locally interval graphs} are defined by forbidding all
  subgraphs in $\cF_{LI}$.
\end{definition} 

An \emph{induced interval subgraph} of $G$ is an interval subgraph
induced by a set $U\subseteq V(G)$ of vertices.  An interval graph
$\underline G$ (resp., $\widehat G$) is called a \emph{spanning
  interval subgraph} (resp., an \emph{interval supergraph}) of $G$ if
it has the same vertex set as $G$ and $E(\underline G)\subseteq E(G)$
(resp., $E(G)\subseteq E(\widehat G)$).  An {induced interval subgraph
  $G[U]$} (resp., a spanning interval subgraph $\underline G$ or an
interval supergraph $\widehat G$) of $G$ is maximum (resp., maximum or
minimum) if $|U|$ (resp., $||\underline G||$ or $||\widehat G||$) is
maximum (resp., maximum or minimum) among all induced interval
subgraphs (resp., spanning interval subgraphs or interval supergraphs)
of $G$; in other words, the number of modifications $V_-:=
V(G)\setminus U$ (resp., $E_-:= E(G)\setminus E(\underline G)$ or
$E_+:= E(\widehat G)\setminus E(G)$) is minimum.

A subset $M$ of vertices forms a \emph{module} of $G$ if all vertices
in $M$ have the same neighborhood outside of $M$.  In other words, for
any pair of vertices $u,v \in M$, a vertex $x \not\in M$ is adjacent
to $u$ if and only if it is adjacent to $v$ as well.  The set $V(G)$
and all singleton vertex sets are modules, called \emph{trivial}.  A
graph on less than three vertices has only trivial modules, while a
graph on three vertices always has a nontrivial module.  A graph on at
least four vertices is \emph{prime} if it contains only trivial
modules, e.g., all holes of length at least five and all caws are
prime.  Two disjoint modules are either nonadjacent or completely
adjacent.  Given any partition $\{M_1,\dots,M_p\}$ of $V(G)$ such that
$M_i$ for every $1\le i\le p$ is a module of $G$, we can associate a
\emph{quotient graph} $Q$, where each vertex represents a module of
$G$, and for any pair of distinct $i, j$ with $1\le i,j\le p$, the
$i$th and $j$th vertices of $Q$ are adjacent if and only if $M_i$ and
$M_j$ are adjacent in $G$.  From $Q$ and $G[M_i]$ for all $1\le i\le
p$ (their total sizes are bounded by $O(||G||)$), the original graph
$G$ can be easily and efficiently retrieved.

\subsection{Our major results}\label{sec:major-results}
We state here the major results of this paper (besides
Thms.~\ref{thm:main} and \ref{thm:main-2}) that are of independent
interest.  Our first result is a straightforward observation on
modules of \lig s and interval graphs.
\begin{proposition}\label{lem:lig-and-modules}
  Let \cG\ be the class of interval graphs or the class of \lig s.  A
  graph $G$ is in \cG\ if and only if a quotient graph $Q$ of $G$ is
  in \cG\ and
  \begin{enumerate}[(1)]
  \item every non-simplicial vertex of $Q$ represents a clique
    module; and
  \item in any pair of adjacent vertices of $Q$, at least one
    represents a clique module.
  \end{enumerate}
\end{proposition} 
Our second major result comprises of a set of theorems.  They
characterize the minimum modifications with respect to modules of the
input graph.  Note that after replacing a module $M$ by another
subgraph, we add edges between every vertex in the new subgraph to
$N(M)$.
\begin{theorem}\label{thm:modules-induced-subgraph}
  Let $G[U]$ be a \miib\ of graph $G$.  For any module $M$ of $G$
  intersecting $U$, the set $M\cap U$ is a module of $G[U]$, and if
  $G$ is $4$-hole-free, then replacing $G[M\cap U]$ by any \miib\ of
  $G[M]$ in $G[U]$ gives a \miib\ of $G$.
\end{theorem}
\begin{theorem}\label{thm:modules-subgraph}
Let $G$ be a $4$-hole-free graph.  There is a \misb\ $\underline{G}$
of $G$ such that the following hold for every module $M$ of $G$:
  \begin{inparaenum}[\itshape i)]
  \item $M$ is a module of $\underline{G}$; and
  \item replacing $\underline{G}[M]$ by any \misb\ of $G[M]$ in
    $\underline{G}$ gives a \misb\ of $G$.
  \end{inparaenum}
\end{theorem}
\begin{theorem}\label{thm:modules-supergraph}
  For any graph $G$, there is a \misp\ $\widehat G$ of $G$ such that
  the following hold for every module $M$ of $G$:
  \begin{inparaenum}[\itshape i)]
  \item $M$ is a module of $\widehat G$; and
  \item if $\widehat G[M]$ is not a clique, then replacing
    $\widehat{G}[M]$ by any \misp\ of $G[M]$ in $\widehat{G}$ gives a
    \misp\ of $G$.
  \end{inparaenum}
\end{theorem}
These results hold regardless of $k$, and thus can be used for any
algorithmic approach, e.g., Thm.~\ref{thm:modules-supergraph} has
already been used in \cite{bliznets-14-interval-completion}.  We
remark that there has been a long relationship between modules and
interval graphs.  Indeed, the algorithm of
\cite{hsu-99-recognizing-interval-graphs}, based on a characterization
of prime interval graphs by Hsu~\cite{hsu-95-recognition-cag}, is
arguably the simplest among all known recognition algorithms for
interval graphs.

Let ${\cal K}$ be a connected graph whose vertices, called
\emph{bags}, are the set of all maximal cliques of $G$.  We say that
$\cal K$ is a \emph{clique decomposition} of $G$ if for any $v \in G$,
the set of bags containing $v$ induces a connected subgraph of ${\cal
  K}$.  A \emph{caterpillar} is a tree that consists of a main path
and all other vertices are leaves connected to it.  An \emph{olive
  ring} is a uni-cyclic graph that consists of a hole (called the main
cycle) and all other vertices are pendant (having degree $1$) and
connected to this hole.  The deletion of any edge from the main cycle
of an {olive ring} results in a caterpillar.  Our third result is on
the clique decomposition of prime \lig s.
\begin{theorem}\label{thm:characterization-lig}
  A prime \lig\ $G$ has a clique decomposition that is either a
  caterpillar when it is chordal; or an olive ring otherwise.  This
  decomposition can be constructed in $O(||G||)$ time.
\end{theorem}

Indeed, given a prime graph $G$ that does not have such a
decomposition, our algorithm is able to identify a subgraph of $G$ in
$\cF_{LI}$.  The following statement is stronger than
Thm.~\ref{thm:characterization-lig} and implies it.
\begin{theorem}\label{thm:decompose-lig}
  Given a prime graph $G$, we can in $O(||G||)$ time either build an
  olive-ring/caterpillar decomposition for $G$ or find a subgraph of
  $G$ in $\cF_{LI}$.  
\end{theorem}

In addition to the above listed concrete results, our algorithms also
suggest a meta approach for designing fixed-parameter algorithms for
vertex deletion problems (where modules are trivially preserved): {\em
  If the object graph class can be characterized by a set of forbidden
  induced subgraphs of which only a finite number are not prime, then
  we may break them first and then use divide-and-conquer, i.e., solve
  the quotient graph and subgraphs induced by modules
  individually.\footnote{ This approach has been widely used in graph
    algorithms.  Based on the quotient graph is taken care of first or
    last, it is known as the top-down way or bottom-up way along the
    modular decomposition tree of the graph \cite{habib-10-survey-md}.
    A notational remark: In this paper we will refrain from using the
    name ``modular decomposition tree'' to avoid confusion with clique
    decomposition.}  }This extends the result of
Cai~\cite{cai-96-hereditary-graph-modification}, and might also be
applicable to some edge modification problems, on which, however, the
preservation of modules needs to be checked case by case.  The main
advantage of this approach is that it enables us to concentrate on
prime graphs and use their structural properties.

\subsection{Motivation and background}
The aforementioned physical mapping of DNA is a central problem in
computational biology
\cite{lander-88-genomic-mapping,arratia-91-genomic-mapping}.  In a
utopia where experimental data were \emph{perfect}, they should define
an interval graph.  Then the problem is equivalent to constructing an
interval model for the graph, which can be done in linear time.  In
the real world we live, however, data are always inconsistent and
contaminated by a few but crucial errors, which have to be detected
and fixed.  In particular, on the detection of false-positive errors
that correspond to fake edges, Goldberg et
al.~\cite{goldberg-95-interval-edge-deletion} formulated the minimum
interval edge deletion problem and showed its NP-hardness.  Likewise,
the deletion of vertices can be used to formulate the detection of
outliers (i.e., elements participating in many false overlaps, both
positive and negative), and the minimum interval vertex deletion
problem is long known to be NP-hard
\cite{krishnamoorthy-79-node-deletion,lewis-80-node-deletion-np}.

Solving the minimum interval vertex deletion problem and the minimum
interval edge deletion problem is equivalent to finding the maximum
induced interval subgraph
\cite{erdos-89-chordal-interval-subgraphs,bliznets-13-max-chordal-interval-subgraphs}
and the maximum spanning interval subgraph
\cite{domotor-13-maximum-interval-subgraph} respectively.  In light of
the importance of interval graphs, it is not surprising that some
natural combinatorial problems can be formulated as, or
computationally reduced to the {interval deletion} problems.  For
instance, Narayanaswamy and Subashini \cite{narayanaswamy-13-d-cos-r}
recently solved the maximum {consecutive ones sub-matrix} problem and
the minimum {convex bipartite deletion} problem by a reduction to
minimum {interval vertex deletion}.  Oum et
al.~\cite{oum-14-vertex-partition-cliquewidth} showed that an induced
interval subgraph can be used to find a special branch decomposition,
which can be in turn used to devise FPT algorithms for a large number
of problems, namely, locally checkable vertex subset and vertex
partitioning problems.  They both used our previous algorithm
\cite{cao-14-interval-deletion} as a subroutine, and thus will benefit
from an improved algorithm directly.

The minimum {interval completion} problem is also a classic NP-hard
problem
\cite{kashiwabara-79-interval-completion,yannakakis-81-minimum-fill-in}.
Besides computational biology, its most important application should
be sparse matrix computations
\cite{tarjan-75-graph-theory-and-gaussian-elimination}.  The {profile
  method} is an extension of the bandwidth method
\cite{rose-72-sparse-matrix,papadimitriou-76-bandwidth}, and their
purpose is to minimize the storage used during Gaussian elimination
for a symmetric sparse matrix.  Both methods attempt to reorder the
rows and columns of the input matrix such that all elimination are
limited within a band or an envelope around the main diagonal, while
all entries outside are always zeroes during the whole computation.
Therefore, we only need to store the elements in the band or envelop,
whose sizes are accordingly called the \emph{bandwidth} and
\emph{profile} \cite{george-81-sparse-positive-definite}.  Rose
\cite{rose-72-sparse-matrix} correlated bandwidth with graphs.  Tarjan
\cite{tarjan-75-graph-theory-and-gaussian-elimination} showed that a
symmetric matrix has a reordering such that its profile coincides with
non-zero entries if and only if it defines an interval graph (there is
an edge between vertices $i$ and $j$ if and only if the $i,j$-element
is non-zero), and finding the minimum profile is equivalent to solving
the minimum interval completion problem.

A very similar problem is the minimum {pathwidth} problem, which also
asks for an interval supergraph $\widehat G$ of $G$ but the objective
is to minimize the size of the maximum clique in $\widehat G$.  This
problem was also known to be NP-hard \cite{kashiwabara-79-pathwidth}.
In light of the hardness of both problems, people turned to finding
minimal interval completions, which can be viewed as a relaxation of
both of them.  Ohtsuki et
al.~\cite{ohtsuki-81-minimal-interval-completion} designed an
algorithm that finds a minimal interval completion in $O(|G|\cdot
||G||)$ time.  Very recently, Crespelle and Todinca
\cite{crespelle-13-minimal-interval-completion} proposed an improved
algorithm that runs in $O(|G|^2)$ time.  This is the best known, and
it remains open to develop a linear-time algorithm for finding a
minimal interval completion.  See also Heggernes et al.\
\cite{heggernes-07-characterize-minimal-interval-completion} for a
characterization of minimal interval completions.  

M{\"o}hring \cite{mohring-96-triangulating-at-free} showed that if a
graph is free of ats, then any minimal chordal supergraph of it is an
interval graph.  The converse was later shown to be true as well
\cite{corneil-97-at-free}.  Since the minimum chordal completion
problem (also known as \emph{minimum fill-in}) is known to be NP-hard
on at-free graphs \cite{arnborg-87-complexity-fill-in}, the minimum
interval completion problem remains NP-hard on at-free graphs.  Other
graph classes on which the minimum interval completion problem remains
NP-hard include chordal graphs
\cite{peng-06-interval-completion-on-choral-graphs}, permutation
graphs \cite{bodlaender-95-widths-on-permutation}, and cocomparability
graphs \cite{habib-94-widths-on-cocomparability}.  On the positive
side, see \cite{kloks-97-at-free-to-interval} for some polynomial
solvable special cases.

\subsection{Graph modification problems and their fixed-parameter tractability}\label{sec:gmp-fpt}
Many classical graph-theoretic problems can be formulated as graph
modification problems to specific graph classes.  For example, Garey
and Johnson \cite[section A1.2]{GJ79} listed 18 NP-complete graph
modification problems (two of which are indeed large collections of
problems; see also
\cite{lewis-80-node-deletion-np,yannakakis-79-connected-maximum-subgraph}).
Graph modification problems are also among the earliest problems whose
parameterized complexity were considered, e.g., Kaplan et
al.~\cite{kaplan-99-chordal-completion} and
Cai~\cite{cai-96-hereditary-graph-modification} devised FPT algorithms
for completion problems to chordal graphs and related graphs.  Indeed,
since the graph modification problems are a natural computational
method for detecting few errors in experimental data, they were an
important motivation behind parameterized computation.  In the special
case when the desired graph class \cG\ can be characterized by a
finite number of forbidden (induced) subgraphs, their fixed-parameter
tractability follows from a basic bounded search tree algorithm
\cite{cai-96-hereditary-graph-modification}.  However, many important
graph classes, e.g., forests, bipartite graphs, and chordal graphs,
have minimal obstructions of arbitrarily large size (cycles, odd
cycles, and holes, respectively).  It is much more challenging to
obtain fixed-parameter tractability results for such classes.

Besides holes, $\cF_I$ has another infinite set of obstructions
(caws), which is far less understood
\cite{corneil-97-at-free,kratsch-06-find-at}.  Since adding or
deleting a single edge is sufficient to fix an arbitrarily large caw,
the modification problems to interval graphs are more complicated than
chordal graphs.  Their fixed-parameterized tractability were
frequently posed as important open problems
\cite{kaplan-99-chordal-completion,downey-fellows-99,bodlaender-95-fpt-computational-biology}.
Only after about two decades were \textsc{interval completion} and
\textsc{interval vertex deletion} shown to be FPT
\cite{villanger-09-interval-completion,cao-14-interval-deletion}.
Both algorithms use a {\em two-phase approach}, where the first phase
breaks all (problem-specifically) {\em small} forbidden induced
subgraphs and the second one takes care of the remaining ones with the
help of combinatorial properties that hold only in graphs without
those small subgraphs.  Nevertheless, neither approach of
\cite{villanger-09-interval-completion,cao-14-interval-deletion}
generalizes to \textsc{interval edge deletion} in a natural way, whose
parameterized complexity remained open to date.  Moreover, both
algorithms of
\cite{villanger-09-interval-completion,cao-14-interval-deletion}
suffer from high time complexity.

In passing let us point out that the vertex deletion version can be
considered as the most robust variant, as it encompasses both edge
modifications in the following sense: if a graph $G$ can be made an
interval graph by $k_-$ edge deletions and $k_+$ edge additions, then
it can also be made an interval graph by at most $k_-+k_+$ vertex
deletions (e.g., one vertex from each added/deleted edge).  In other
words, the graph class $\mbox{interval} + k v$ contains both classes
$\mbox{interval} + k e$ and $\mbox{interval} - k e$.  The similar fact
holds for all hereditary graph classes.  On the other hand,
$\mbox{interval} + k e$ and $\mbox{interval} - k e$ are incomparable
in general, e.g., a $6$-hole is in interval$+ 1 e$ and a $K_{2,3}$ is
in interval$- 1 e$ but not the other way.

\subsection{Efficient detection of (small) forbidden induced
  subgraphs}
\label{sec:hardness-find-fis}
As said, if the object graph class has only a finite number of
forbidden induced subgraphs, then the modification problem is
trivially FPT.  This observation can be extended to a family of
forbidden induced subgraphs that, though infinite, can be detected in
polynomial time and destroyed by a bounded number of ways; the most
remarkable example is chordal completion
\cite{kaplan-99-chordal-completion,cai-96-hereditary-graph-modification}.
For the purpose of contrast, let us call this {\em one-phase
  approach}.  In carrying out the aforementioned two-phase approach,
one usually focuses on the second phase, on the ground that the first
phase {\em seems to be the same} as the one-phase approach.  This
ground is, nevertheless, shaky: more often than not, algorithms based
on the one-phase approach run in linear time, but all previous
algorithms
\cite{villanger-09-interval-completion,villanger-13-pivd,cao-14-interval-deletion}
based on this two-phase approach have high polynomial factors in their
running times, which are mainly determined by the time required to
detect \emph{small} forbidden induced subgraphs in the first phase.
As we will see, the detection of a \emph{small} forbidden induced
subgraph is usually far more demanding than an arbitrary one.

Kratsch et al.~\cite{kratsch-06-certifying-interval-and-permutation}
presented a linear-time algorithm for detecting a hole or an at from a
non-interval graph.  It first calls the hole-detection algorithm of
Tarjan and Yannakakis \cite{tarjan-85-certifyig-chordal-recognition},
which either returns a hole, or reduces to finding an at in a chordal
graph.  The additional chordal condition for the detection of an at is
crucial: we do not know how to find an at in a general graph in linear
time.  The best known recognition algorithm for at-free graphs takes
$O(|G|^{2.82})$ time \cite{kohler-04-recognize-at-free}, and Kratsch
and Spinrad \cite{kratsch-06-find-at} showed that this algorithm can
be used to find an at in the same time if the graph contains one.  A
more important result of \cite{kratsch-06-find-at} is that recognizing
at-free graphs is at least as difficulty as finding a triangle.  The
detection of an at cannot be easier than the recognition of at-free
graphs, and hence a linear-time algorithm for it is very unlikely to
exist.  (See also
\cite{spinrad-03-efficient-graph-representations}.\footnote{The
  detection of triangles is a fundamental computational problem and
  has been extensively studied.  However, the best algorithms are the
  trivial ones, using either enumeration or fast matrix
  multiplication.  Recall that the current fastest algorithm for
  matrix multiplication takes $O(|G|^{\omega})$ time, where
  $\omega<2.3727$ \cite{williams-12-matrix-multiplication}.
  Spinrad listed an $o(|G|^3)$-time combinatorial algorithm for
  detecting triangles as an open problem \cite[Open problem 8.1, page
  101]{spinrad-03-efficient-graph-representations}.  In the same work
  he also conjectured that it is computationally equivalent to
  (0,1)-matrix multiplication verification problem.  Recall that in
  matrix multiplication verification problem, we are given three
  matrices $A$, $B$, and $C$, and asked whether $A\times B = C$ or
  not.  See also \cite{williams-10-subcubic-equivalence}.})  When an
at is detected, the algorithm of Kratsch et
al.~\cite{kratsch-06-certifying-interval-and-permutation} also
provides in the same time a witness for it.  This witness, although
unnecessarily minimal itself, can be used to easily retrieve a minimal
one, i.e., a caw (see also \cite{lindzey-13-find-forbidden-subgraphs}
for another approach).

Obviously, for any hereditary graph class, the detection of a
forbidden induced subgraph is never easier than the recognition of
this graph class.  On the other hand, we have seen that the detection
of a hole, an at with witness, and a subgraph in $\cF_I$ can be done
in the same asymptotic time as the recognition of chordal graphs,
at-free graphs, and interval graphs, respectively.  From these
examples one may surmise that the requirement of explicit evidence
does not seem to pose an extra burden to the recognition algorithms.
This is known to be true for almost all polynomial-recognizable graph
classes with known characterization by forbidden induced subgraphs.

However, it changes drastically when the evidence is further required
to have a small or minimum number of vertices.  The most famous
example should be the detection of cycles: while an arbitrary cycle
can be trivially found in linear time, the detection of a shortest
cycle, which includes the triangle-detection as a special case, is
very unlikely to be done in linear time.  Even finding a short cycle
in linear time seems to be out of the question (see, e.g.,
\cite{itai-78-minimum-circuit}).  Assuming that triangles cannot be
detected in linear time, we can also rule out the possibility of
linear-time detection of a minimum subgraph in $\cF_I$ or a shortest
hole.  Let $G'$ be the graph obtained by subdividing a graph $G$
(i.e., for each edge $uv\in E(G)$, adding a new vertex $x$, connecting
it to both $u$ and $v$, and deleting $u v$), then $G$ contains a
triangle if and only if the minimum subgraph of $G'$ in $\cF_I$ is a
$6$-hole.  Since $G'$ has $|G| + ||G||$ vertices and $2||G||$ edges,
an linear-time algorithm for finding a minimum subgraph in $\cF_I$ can
be used to detect a triangle in linear time.  With a similar
reduction, we can show that a linear-time algorithm for detecting
subgraphs in $\cF_{LI}$---recall that they are small graphs in
$\cF_I$---is unlikely to exist, as it can be used to detect a claw in
linear time, and further to detect a triangle in $O(|G|^2)$ time,
which would have groundbreaking consequence (see \cite[Open problem
8.3, page 103]{spinrad-03-efficient-graph-representations}).  Similar
phenomenon has been observed in detecting minimum Tucker submatrices,
i.e., a minimal matrix that does not have consecutive-ones property
\cite{blin-12-minimum-tucker-submatrices} and shortest even holes
\cite{chudnovsky-05-detect-even-holes}. \footnote{In the published
  version of the paper \cite{chudnovsky-05-detect-even-holes}, the
  algorithm is stated as detecting an arbitrary even hole and it was
  asked as an open problem for an algorithm that finds a shortest one.
  But according to Seymour (private communication), the authors later
  observed that the return of their algorithm has to be the shortest.}

Another crucial step of our algorithm is to find all simplicial
vertices of a graph.  Again, it is unlikely to be done in linear time:
Kratsch and Spinrad \cite{kratsch-06-find-at} showed that counting the
number of simplicial vertices is already at least as hard as detecting
a triangle.  Indeed, there is even no known algorithm that can detect
a single simplicial vertex in linear time.  The only known way of
finding a simplicial vertex is either enumerating all vertices or
using fast matrix multiplication.  Kloks et al.\
\cite{kloks-00-finding-simplicial-vertices} showed that in the same
time one can actually list all simplicial vertices.  This is the best
known in general graphs.  See also \cite[open problems 4.3 and
4.4]{woeginger-08-open-problems}.

\subsection{Main challenges and our techniques}
\label{sec:techniques}
We describe here the main challenges and intuitions behind the
techniques that we use to address them.  They can be roughly put into
two categories: for the linear dependence on the graph size and for
the smaller exponential dependence on the parameter.  Also sketched
here is why known techniques from previous work will not suffice.  We
basically take the two-phase approach, subgraphs in $\cF_{LI}$ first
and then the rest (large ones).  We say that caws and holes in
$\cF_{LI}$ are \emph{small} and \emph{short} respectively; other caws,
namely, $\dag$s and $\ddag$s, are \emph{large}, and holes of length
six or more are \emph{long}.  It is worth noting that the thresholds
are chosen by structural properties instead of sizes.

\paragraph{Linear dependence on the graph size.}  The biggest
challenge is surely the efficient detection of a subgraph in
$\cF_{LI}$, or more specifically, the detection of a short hole or
small caw.  As explained above, we do not expect a linear-time
algorithm for this task.  Instead, we relax it to the following:
either find a subgraph in $\cF_{LI}$ or build a structural
decomposition (Thm.~\ref{thm:decompose-lig}) that is sufficient for
the second phase.  For the disposal of large forbidden induced
subgraphs in the second phase, the algorithm of
\cite{cao-14-interval-deletion} breaks long holes first, and then
large caws in a chordal graph.  There is no clear way to implement
this tactic in linear time: the disposal of holes introduces a factor
$|G|$, while finding a caw gives another factor $||G||$.  Neither of
them seems to be improvable to $o(|G|)$.  We are thus forced to
consider an alternative approach, i.e., we may have to deal with large
caws in a non-chordal graph.  Hence completely new techniques are
required.  Overcoming these two difficulties enables us to deliver
linear-time algorithms.

\paragraph{Exponential dependence on the parameter.}  To claim the
fixed-parameter tractability of \textsc{interval edge deletion} and
better dependence on $k$ for \textsc{interval completion}, we still
have some major concerns to address.  Since fixing holes by edge
additions is well understood, the algorithm of Heggernes et
al.~\cite{villanger-09-interval-completion} for \textsc{interval
  completion} assumes the input graph to be chordal, and focuses on
the disposal of caws.  However, holes pose a nontrivial challenge to
us in the deletion problems, and thus the techniques of
\cite{villanger-09-interval-completion} do not apply.  On the other
hand, the algorithm in \cite{cao-14-interval-deletion} heavily relies
on the fact that the deletion of vertices leaves an induced subgraph.
Essentially, it looks for a minimum set of vertices intersecting all
subgraphs in $\cF_I$, so called \emph{hitting set}.  Deleting any
vertex from a subgraph in $\cF_I$ breaks this subgraph once and for
all, but adding/deleting an edge to break an erstwhile subgraph in
$\cF_I$ might introduce new one(s).  As a result, the ``hitting set''
observation does not apply to edge modifications problems.
\begin{itemize}
\item The first difficulty that presented itself at this point is on
  the preservation of modules, which is trivial for vertex deletions,
  but not true for edge modifications in general.  Simple examples
  tell us that not all \misb s and \misp s preserve all modules.  What
  we do here is to identify appropriate technical conditions, under
  which there exists some \misb\ or \misp\ that preserves all modules,
  and make them satisfied at the onset of the second phase.
\item The other difficulty is why it suffices to consider a bounded
  number of modifications to fix a special caw, for which we need to
  argue that most possible modifications are local to it and can be
  decided locally.  In \cite{cao-14-interval-deletion}, we studied in
  a chordal graph with no small caws, how a caw interacts with others;
  similar arguments are obviously inapplicable to edge variations.
  Even for vertex deletions, as we had make a compromise to work on
  non-chordal graph, we need a new argument that does not assume the
  chordality.
\end{itemize}

 \section{Outline}\label{sec:outline}
The purpose of this section is to describe the main steps of our
algorithm at a high level.  A quotient graph $Q$ is isomorphic to an
induced subgraph of $G$, e.g., we can pick an arbitrary vertex from
each module $M$ of the module partition and take the induced subgraph.
Therefore, whenever a forbidden induced subgraph of $Q$ is detected,
it can be translated into a forbidden induced subgraph of $G$
directly.

\subsection{Maximal strong modules}
Behind Prop.~\ref{lem:lig-and-modules} and
Thms.~\ref{thm:modules-induced-subgraph}-\ref{thm:modules-supergraph}
is a very simple observation: $4$-holes are the only non-prime graph
in $\cF_{I}$ and $\cF_{LI}$.  Note that for any induced subgraph $X$
intersecting a module $M$, their intersection $V(X)\cap M$ is a
(possibly trivial) module of $X$.  Therefore, if $X$ is prime and
$V(X)\not\subseteq M$, then it intersects $M$ by at most one vertex.
Fix any module partition and accordingly a quotient graph $Q$.  If
$X$ is in $\cF_{I}$ or $\cF_{LI}$ but not a $4$-hole, then $X$ either
contains at most one vertex from each module, thus isomorphic to an
induced subgraph of $Q$, or is fully contained in some module from the
given partition.  On the other hand, a $4$-hole may contain precisely
two vertices of a module $M$, and then the other two vertices must be
neighbors of this module.  We have two cases: the other two vertices
belong to the same module $M'$ that is adjacent to $M$, or they belong
to two different (nonadjacent) modules.  In other words, either two
non-clique modules are adjacent, or a non-clique module is not
simplicial in $Q$.  This concludes Prop.~\ref{lem:lig-and-modules}.

However, Prop.~\ref{lem:lig-and-modules} has no direct algorithmic
use: a graph might have an exponential number of modules and quotient
graphs.  A module $M$ is \emph{strong} if for every other module $M'$
that intersects $M$, one of $M$ and $M'$ is a proper subset of the
other.  All trivial modules are strong.  We say that a strong module
$M$, different from $V(G)$, is \emph{maximal} if the only strong
module properly containing $M$ is $V(G)$.  Using definition it is easy
to verify that maximal strong modules of $G$ are disjoint and every
vertex $v$ of $G$ appears in one of them.  Therefore, they partition
$V(G)$, and define a special {quotient graph} $Q$.  If $G$ is not
connected, then each maximal strong module is a component of it, and
$Q$ has no edge.  Recall that the complement graph of $G$ is defined
on the same vertex set $V(G)$, where a pair of vertices $u$ and $v$ is
adjacent if and only if $u\not\sim v$ in $G$.  Thus, the complement of
$G$ has the same set of modules as $G$; in particular, if it is not
connected, then its components are the maximal strong modules of $G$,
and hence $Q$ is complete.  If both the graph $G$ and its complement
are connected, then $Q$ must be prime
\cite{gallai-67-transitive-orientation}.  Note that this is the only
case that a quotient graph can be prime; in other words, a prime
quotient graph must be defined by maximal strong modules.  

Hereafter, the quotient graph $Q$ is always decided by maximal strong
modules of $G$; when $G$ itself is prime, they are isomorphic.  There
are at most $|G|$ maximal strong modules, which can be found in linear
time \cite{habib-10-survey-md}.  Therefore, the following corollary of
Prop.~\ref{lem:lig-and-modules} will be more useful for algorithmic
purpose.  Recall that a vertex $v$ is \emph{universal} in $G$ if $N[v]
= V(G)$.  It is easy to verify that a prime graph is necessarily
connected, and its simplicial vertices are pairwise nonadjacent.
\begin{corollary}\label{lem:lig-and-max-modules}
  Let \cG\ be the class of interval graphs or the class of \lig s.  A
  graph $G$ having no universal vertices is in \cG\ if and only if
  \begin{enumerate}[(1)]
  \item the quotient graph $Q$ decided by maximal strong modules of
    $G$ is in \cG\ but not a clique;
  \item $G[M]\in \cG$ for every module $M$ represented by a simplicial
    vertex of $Q$; and
  \item $G[M]$ is a clique for every module $M$ represented by a
    non-simplicial vertex of $Q$.
  \end{enumerate}
\end{corollary} 

Every parameterized modification problem has an equivalent
optimization version, which asks for a minimum set of modifications;
the resulting interval graph is called an optimum solution to this
problem.  Clearly, a graph $G$ is in the class interval$+ k v$,
interval$+ k e$, or interval$- k e$ if and only if the minimum number
of vertex deletions, edge deletions, or edge additions respectively
that transform $G$ into an interval graph is no more than $k$.
Although the recognition/modification problems we are working on do
not explicitly ask for an optimum solution, an optimum one will serve
our purpose.  We have stated in
Thms.~\ref{thm:modules-induced-subgraph}-\ref{thm:modules-supergraph}
that there are always optimum solutions well aligned with modules of
the input graph.  Again, for algorithmic purpose, the following
variations formulated on maximal strong modules are more convenient
for our divide-and-conquer approach.  As we will see shortly, they are
indeed equivalent to
Thms.~\ref{thm:modules-induced-subgraph}-\ref{thm:modules-supergraph}
respectively.
\begin{theorem}\label{thm:modules-induced-subgraph+}
  Let $G$ be a graph of which every $4$-hole is contained in some
  maximal strong module, and let $G[U]$ be a \miib\ of $G$.  For any
  maximal strong module $M$ of $G$ intersecting $U$, the set $M\cap U$
  is a module of $G[U]$, and replacing $G[M\cap U]$ by any \miib\ of
  $G[M]$ in $G[U]$ gives a \miib\ of $G$.
\end{theorem}
\begin{theorem}\label{thm:modules-subgraph+} 
  Let $G$ be a graph of which every $4$-hole is contained in some
  maximal strong module.  There exists a \misb\ $\underline G$ of $G$
  such that every maximal strong module $M$ of $G$ is a module of
  $\underline G$, and replacing $\underline G[M]$ by any \misb\ of
  $G[M]$ in $\underline G$ gives a \misb\ of $G$.
\end{theorem}
  We may assume without loss of
generality that the input graph contains no universal vertices.
According to Cor.~\ref{lem:lig-and-max-modules}, the condition of
Thms.~\ref{thm:modules-induced-subgraph+} and
\ref{thm:modules-subgraph+} is satisfied if
\begin{inparaenum}[(\itshape i\upshape)]
\item $Q$ is not a clique,
\item $Q$ contains no $4$-hole, and
\item every non-simplicial vertex of $Q$ represents a clique module of
  $G$.
\end{inparaenum}
In this paper cliques are required to be nonempty.  It is easy to
verify that the \miib\ or \misb\ of a graph is clique if and only if
it is a clique; thus, under the condition of
Thms.~\ref{thm:modules-induced-subgraph+} and
\ref{thm:modules-subgraph+}, a maximal strong module $M$ is a clique
of the object interval graph if and only if it is a clique of $G$.

\begin{theorem}\label{thm:modules-supergraph+}
  There is a \misp\ $\widehat G$ of $G$ such that every maximal strong
  module $M$ of $G$ is a module of $\widehat G$, and if $\widehat
  G[M]$ is not a clique, then replacing $\widehat G[M]$ by any \misp\
  of $G[M]$ in $\widehat G$ gives a \misb\ of $G$.
\end{theorem}

\subsection{Characterization and decomposition of \lig s}
Prop.~\ref{lem:lig-and-modules} reduces the main task of the first
phase, the detection of a subgraph of $G$ in $\cF_{LI}$, to two
simpler tasks, namely, \emph{finding a subgraph of $Q$ in $\cF_{LI}$}
and \emph{finding all simplicial vertices of $Q$ when it is a \lig}.
Both tasks are trivial when $Q$ is an interval graph (including
cliques and edgeless graphs), and hence we concentrate on prime
non-interval graphs.  If such a graph contains no subgraph in
$\cF_{LI}$, i.e., being a \lig, then it must contain some large caw or
some long hole.  Therefore, we start from characterizing large caws
and long holes in prime locally interval graphs.  A glance at
Fig.~\ref{fig:at} tells us that each caw contains precisely three
simplicial vertices, which form the unique at of this caw; they are
called the \emph{terminals} of this caw.\footnote{One may have noticed
  that a hole of six or more vertices also witnesses an at (e.g., any
  three pairwise nonadjacent vertices from it) and is minimal.  It,
  however, behaves quite differently from a caw, e.g., none of its
  vertices is simplicial, and it has more than one at---indeed, every
  vertex is in some at.}  Each large caw (the second row of
Fig.~\ref{fig:at}) contains a unique terminal $s$, called the
\emph{shallow terminal}, such that the deletion of $N[s]$ from this
caw leaves an induced path.

\begin{theorem}\label{thm:1}
  Let $W$ be a large caw of a prime graph $G$.  We can in $O(||G||)$
  time find a subgraph of $G$ in $\cF_{LI}$ if the shallow terminal of
  $W$ is non-simplicial in $G$.
\end{theorem}
If a prime \lig\ $G$ is chordal, then by Thm.~\ref{thm:1}, every caw
contains a simplicial vertex (its shallow terminal), and thus $G -
SI(G)$ must be an interval graph.  In a chordal graph, $SI(G)$ can be
easily found, and then a caterpillar decomposition for $G$ can be
obtained by adding $SI(G)$ to a clique path decomposition for $G -
SI(G)$ (Section~\ref{sec:olive-ring-decomposition}).  This settles the
chordal case of Thm.~\ref{thm:decompose-lig}; we may hence assume that
$G$ is not chordal and has a long hole $H$.
\begin{theorem}\label{thm:2}
  Let $H$ be a hole of a prime graph $G$.  We can in $O(||G||)$ time
  find a subgraph of $G$ in $\cF_{LI}$ if there exists a vertex $v$
  satisfying one of the following:
  \begin{inparaenum}[(1)]
  \item the neighbors of $v$ in $H$ are not consecutive;
  \item $v$ is adjacent to $|H| - 2$ or more vertices in $H$; and
  \item $v$ is non-simplicial and nonadjacent to $H$.
  \end{inparaenum}
\end{theorem}

If $G$ is a prime \lig, then for any vertex $h$ of the hole $H$, the
subgraph $G - N[h]$ must be chordal; otherwise, $h$ and any hole of $G
- N[h]$ will satisfy Thm.~\ref{thm:2}(3).  Therefore, combining
Thms.~\ref{thm:1} and \ref{thm:2}, we conclude that
$G - SI(G) - N[h]$ must be an interval subgraph, and has a linear
structure.  These observations inspire the definition of the auxiliary
graph $\mho(G)$ (with respect to $H$), which is the main technical
tool for analyzing prime non-chordal graphs.  Here we need a special
vertex of $H$ satisfying some \emph{local properties}, which can be
found in linear time (Section~\ref{sec:mho}).  We number vertices in
$H$ such that $h_0$ is this special vertex and define $\oo := N[h_0]$.
We designate the ordering $h_0, h_1, h_2, \cdots$ of traversing $H$ as
\emph{clockwise}, and the other \emph{counterclockwise}.  The local
properties enable us to assign a direction to each edge between $\oo$
and $\overline \oo$, i.e., $V(G)\setminus \oo$, in accordance with the
direction of $H$ itself.  We use \ec\ and \ecc\ to denote the set of
clockwise and counterclockwise edges from $\oo$, respectively;
$\{\ec,\ecc\}$ partitions $\oo\times \overline{\oo}$.

\begin{definition}\label{def:mho}
  The vertex set of $\mho(G)$ consists of $\overline{\oo}\cup L\cup
  R\cup \{w\}$, where $L$ and $R$ are distinct copies of $\oo$, i.e.,
  for each $v\in \oo$, there are a vertex $v^l$ in $L$ and another
  vertex $v^r$ in $R$, and $w$ is a new vertex distinct from $V(G)$.
  For each edge $u v\in E(G)$, we add to the edge set of $\mho(G)$
\begin{itemize}
\item an edge $u v$ if neither $u$ nor $v$ is in $\oo$;
\item two edges $u^l v^l$ and $u^r v^r$ if both $u$ and $v$ are in
  $\oo$; or
\item an edge $u v^l$ or $u v^r$ if $v\in \oo$ and $uv\in \ec$ or $uv\in
  \ecc$ respectively.
\end{itemize}
Finally, we add an edge $w v^l$ for every $\{v\in \oo: u v \in \ecc\}$.
\end{definition} 

It is easy to see that the order and size of $\mho(G)$ are upper
bounded by $2|G|$ and $2||G||$ respectively.  We will show in
Section~\ref{sec:mho} that an adjacency list representation of
$\mho(G)$ can be constructed in linear time.  The auxiliary graph
carries all structural information of $G$ useful for us and is easy to
manipulate; in particular, the new vertex $w$ is introduced to
memorize the connection between $L$ and the right end of $G - T$.
{The shape of symbol $\mho$ is a good hint for understanding the
  structure of the auxiliary graph.}  Suppose $G$ has an olive-ring
structure, then $\mho(G)$ has a caterpillar structure, which is
obtained by unfolding the olive ring as follows.  The subgraph $G -
\oo$ has a caterpillar structure, to the ends of which we append two
copies of $\oo$.  The two copies of $\oo$, namely, $L$ and $R$, are
identical, and every edge between $\oo$ and $\overline \oo$ is carried
by only one copy of it, based on it is in \ecc\ or \ec.  Furthermore,
properties stated in the following theorem allow us to fold (the
reverse of the ``unfolding'' operation) the caterpillar structure of
$\mho(G)$ back to produce the olive-ring decomposition for $G$.  Note
that $\mho(G - SI(G))$ is different from $\mho(G) - SI(\mho(G))$.
\begin{theorem}\label{thm:3}
  A vertex different from $\{w\}\cup R$ is simplicial in $\mho(G)$ if
  and only if it is derived from some simplicial vertex of $G$.
  Moreover, we can in $O(||G||)$ time find a subgraph of $G$ in
  $\cF_{LI}$ if
  \begin{inparaenum}[\itshape 1)]
  \item $\mho(G)$ is not chordal; or
  \item $\mho(G - SI(G))$ is not an interval graph.
  \end{inparaenum}
\end{theorem}

We may assume that the graph $\mho(G)$ is chordal, whose simplicial
vertices can be identified easily.  As a result of Thm.~\ref{thm:3},
we can retrieve $SI(G)$ and obtain the graph $\mho(G - SI(G))$.  If it
is not an interval graph, then we are done with
Thm.~\ref{thm:decompose-lig}.  Otherwise, we apply the following
operation to sequentially build a hole decomposition for $G-SI(G)$ and
an olive-ring decomposition for $G$.  Noting that all holes of $G$ are
also in $G - SI(G)$, once the decomposition for $G - SI(G)$ is
produced, we can use it to find a shortest hole of $G$.  We proceed
only when this hole is long.
\begin{lemma}\label{thm:build-hole-decomposition}
  Given a clique path decomposition for $\mho(G - SI(G))$, we can in
  $O(||G||)$ time build a clique decomposition for $G - SI(G)$ that is
  a hole.  Moreover, we can find in $O(||G||)$ time a shortest hole of
  $G$.
\end{lemma}

\begin{theorem}\label{lem:extend-decomposition}
  Given a clique hole decomposition for $G-SI(G)$, we can in
  $O(||G||)$ time construct a clique decomposition for $G$ that is an
  olive ring.
\end{theorem}

Putting together these steps, we get the decomposition algorithm in
Fig.~\ref{fig:recognize-lig}, from which Thm.~\ref{thm:decompose-lig}
follows.  This concludes the proof of the characterization and
decomposition of prime \lig s.
\begin{figure}[h]
\setbox4=\vbox{\hsize28pc \noindent\strut
\begin{quote}
  \vspace*{-5mm} \small 

  {\bf Algorithm} {\bf decompose}($G$)
  \\
  {\sc input}: a prime graph $G$.
  \\
  {\sc output}: a caterpillar/olive-ring decomposition for $G$ or a
  subgraph of $G$ in $\cF_{LI}$.

  1 \hspace*{2ex} {\bf if} $G$ is chordal {\bf then}
  \\
  \hspace*{7ex} {\bf if} $G$ is an interval graph {\bf then return} a
  clique path decomposition for $G$;
  \\
  \hspace*{7ex} {\bf if} $G - SI(G)$ is an interval graph {\bf then
    return} a caterpillar decomposition for $G$;
  \\
  \hspace*{7ex} find a caw $W$; {\bf if} $W$ is small {\bf then
    return} $W$, {\bf else call} Thm.~\ref{thm:1};
  \\
  2 \hspace*{2ex} find a hole $H$ of $G$; build $\mho(G)$;
  \\
  3 \hspace*{2ex} {\bf if} $\mho(G)$ is not chordal {\bf then call}
  Thm.~\ref{thm:3}(1);
  \\
  4 \hspace*{2ex} find $SI(\mho(G))$ and $SI(G)$; construct $\mho(G) -
  SI(G)$;
  \\
  5 \hspace*{2ex} {\bf if} $\mho(G) - SI(G)$ is not an interval graph
  {\bf then call} Thm.~\ref{thm:3}(2);
  \\
  6 \hspace*{2ex} {\bf call} Lem.~\ref{thm:build-hole-decomposition}
  and Thm.~\ref{lem:extend-decomposition} to build an olive-ring
  decomposition for $G$.

\end{quote} \vspace*{-6mm} \strut} $$\boxit{\box4}$$
\vspace*{-9mm}
\caption{The decomposition algorithm for
  Thm.~\ref{thm:decompose-lig}.}
\label{fig:recognize-lig}
\end{figure}

\subsection{Recognition of almost interval graphs}
In lieu of general solutions, we may consider only those optimum
solutions satisfying
Thms.~\ref{thm:modules-induced-subgraph+}-\ref{thm:modules-supergraph+},
which focus us on the quotient graph $Q$ defined by maximal strong
modules of $G$.  If $Q$ is a clique, then we have either a $4$-hole or
a smaller instance (by removing all universal vertices).  Otherwise
$Q$ is prime and we call Thm.~\ref{thm:decompose-lig} with it, which
has two possible outcomes; there are only a constant number of
modifications applicable to a small caw, and thus we may assume that
the outcome is an olive-ring decomposition $\cal K$.  For the
completion problem, as holes can be easily filled, we can always
assume that the graph is chordal and $\cal K$ is a caterpillar.  With
decomposition $\cal K$, whether the input instance satisfies the
conditions of Thms.~\ref{thm:modules-induced-subgraph+} and
\ref{thm:modules-subgraph+} can be easily checked.  If some
non-simplicial vertex in $Q$ represents a non-clique module, then we
have a $4$-hole.  Otherwise, we work on all \msm s and find each of
them an optimum solution, for which it suffices to consider those
represented by simplicial vertices in $Q$.  Using definition it is
easy to verify that the resulting graph has the same set of \msm s as
$G$, and hence $Q$ remains the prime quotient of it.  With inductive
reasoning, we may assume that every simplicial vertex in $Q$
represents now an interval subgraph.  In summary, the only condition
of Cor.~\ref{lem:lig-and-max-modules} that might remain unsatisfied is
whether $Q$ itself is an interval graph.  Therefore, this section is
devoted to the disposal of $Q$, which is prime and has a
caterpillar/olive-ring decomposition $\cal K$.

Allow us to use some informality in explaining the intuition behind
the our algorithms for deletion problems.  Recall that clique path
decompositions are characteristic of interval graphs
\cite{fulkerson-65-interval-graphs}.  With a bird's-eye view, what we
have is an olive ring, while what we want is a path; it may help to
mention that the maximal cliques of the graph may change and the bags
of the latter is not necessarily a subset of the former.  Toward this
end, we need to \emph{cut the main cycle} and \emph{strip off its
  leaves} of the olive ring, and there are immediately two options
based on which action is taken first.  Interestingly, they correspond
to the disposal of holes and caws, respectively.  From $\cal K$ we can
observe that every hole $H$ of $G$ is global in the sense that it
dominates all holes.  In contrast, every caw is local, and with
diameter at most four, so it sees only a part of the main cycle.  The
structural difference of holes and caws suggests that different
techniques are required to handle them.  As explained in
Section~\ref{sec:techniques}, we strip the leaves off the olive ring
first to make it a hole.

Let ($s: c_1,c_2:l,B,r$) be a large caw in $Q$, possibly $c_1 = c_2$
(see the second row of Fig.~\ref{fig:at}).  We consider its terminals
as well as their neighbors, i.e., $\{s, c_1,c_2, l,b_1, b_d, r\}$.  It
is observed that if all of them are retained and their
adjacencies---except of $\{l, c_1\}$ and $\{c_2, r\}$, which are
adjacent in a $\ddag$ but not a $\dag$---are not changed, then in an
interval model of the object interval graph, they must be arranged in
the way depicted in Fig.~\ref{fig:fixed-frame}.  As indicated by the
dashed extensions, the interval for $c_1$ (resp., $c_2$) might or
might not extend to the left (resp., right) to intersect the interval
for $l$ (resp., $r$).  Our main observation is on the position of the
interval for $s$: it has to lie between $b_1$ and $b_d$, which are
nonadjacent---this explains why we single out net and (rising) sun
from $\dag$ and $\ddag$ respectively.  Recall that $s$ is originally
adjacent to no vertex in the \stpath{b_1}{b_d} $B$.  Therefore, we
need to delete some vertex or edge to break $B$, or add an edge to
connect $s$ to some inner vertex of $B$.

\begin{figure*}[h!]
  \centering \includegraphics{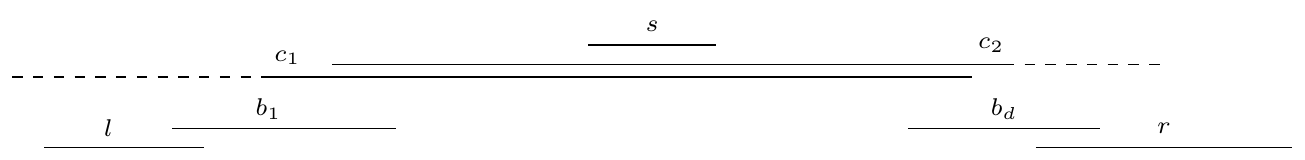}
  \caption{An interval model of an unchanged frame.}
  \label{fig:fixed-frame}
\end{figure*}

In the discussion above, what matters is only the terminals and their
neighbors, while the particular \stpath{b_1}{b_d} $B$ becomes
irrelevant.  Indeed, any induced \stpath{b_1}{b_d} in $(N(c_1)\cap
N(c_2))\setminus N(s)$ can be used in place of $B$ to give a caw of
the same type (though not necessarily the same size), which has the
same set of terminals.  A similar operation is thus needed for all of
them, and the particular base is immaterial, inspiring us to consider
the following two sets of vertices.  Of a large caw ($s:
c_1,c_2:l,B,r$), the \emph{frame} is denoted by ($s: c_1,c_2:l,b_1;
b_d,r$), and the set of \emph{inner vertices} is composed of all
vertices that can be used to make a caw with frame ($s: c_1,c_2:l,b_1;
b_d,r$); they are denoted by $F$ and $IN(F)$ respectively.  Without a
specific path $B$ in sight, it would be more convenient to use ($s:
c_1,c_2:l,l_b; r_b,r$) to denote a frame.

We can find a frame $F$ that is minimal in a sense.  Its definition,
give in Section~\ref{sec:minimal-frames}, is essentially the same as
what is used in our previous work \cite{cao-14-interval-deletion}.
The major concern here is how to find a minimal frame in linear time;
considering that the graph might still contain holes and small caws,
it is far more complicated than \cite{cao-14-interval-deletion}.  This
is achieved using the olive-ring decomposition
(Section~\ref{sec:minimal-frames}).  The rest is then devoted to the
disposal of (caws with) this minimal frame.

Consider first vertex deletions.  We show that any optimum solution
deletes either some vertex of $F$ or a minimum \stsep{l_b}{r_b} in the
subgraph induced by $IN(F)\cup\{l_b, r_b\}$.  The second case is our
main concern, for which we manage to show that any minimum
\stsep{l_b}{r_b} will suffice; it can be found in linear time.  This
case can be informally explained as follows.  All vertices in $IN(F)$
and $\{l_b, r_b, c_1, c_2\}$ must reside in a consecutive part of the
main cycle of the olive-ring decomposition, and we need to find some
``place'' in between to accommodate $s$.  We show that it suffices to
``cut any thinnest place'' between $l_b$ and $r_b$, and use this space
for $s$.  Recalling that $F$ has at most seven vertices, we have then
an $8$-way branching for disposing of this frame.

The basic idea for edge deletions is similar as vertex deletions,
i.e., we delete either one of a bounded number of edges or a minimum
edge \stsep{l_b}{r_b}, but we are now confronted with more complex
situations.  First, the assumption that no edge in $F$ is deleted does
not suffice, so instead we find a shortest \stpath{l_b}{r_b} $B$ with
all inner vertices from $IN(F)$.  If $B$ has a bounded length, then we
branch on deleting every edge in it.  Otherwise, we argue that either
one of the first or last $O(k)$ edges of $B$ is deleted, or it
suffices to find a minimum set of edges whose deletion separates $l_b$
and $r_b$ in $IN(F)$, which can also be viewed as the ``thinnest
place'' (in another sense) between $l_b$ and $r_b$.  This gives an
$O(k)$-way branching.

After all caws are destroyed as above, if $Q$ is chordal, then
problems are solved; otherwise, it has a clique hole
decomposition.\footnote{As we will see, it is a normal Helly
  circular-arc graph with no short holes.}  Since every simplicial
vertex of $Q$ represents an interval subgraph, this decomposition can
be extended to a clique hole decomposition for $G$.  As a result, all
its holes can be broken at a fell swoop in linear time, which solves
the problems.
\begin{lemma}\label{lem:hole-cover}
  Given a clique hole decomposition for graph $G$, the problems {\sc
    interval vertex deletion} and {\sc interval edge deletion} can be
  solved in time $O(||G||)$ and $k^{O(k)}\cdot ||G||$ respectively.
\end{lemma}

The ``thinnest places'' are also crucial for completions, though the
argument becomes even more delicate.  Our focus is on a \misp\
$\widehat G$ that contains no edge in $\{l c_2, c_1 r, l_b r_b, s l_b,
s r_b\}$; in particular, $l_b$ and $r_b$ remain nonadjacent in
$\widehat G$.  As said, we attend to caws only when the graph is
already chordal, which means that the clique decomposition $\cal K$ is
a caterpillar.  Therefore, $l_b$ and $r_b$ can be used to decide a
left-right relation for both the caterpillar decomposition of $G$ and
an interval model of $\widehat G$.  After adding edges, an interval
for a vertex that is to the right of $r_b$ in $G$ might intersect part
or all intervals between $l_b$ and $r_b$.  We argue that such an
interval either reaches $l_b$, or is to the right of some position
(informally speaking, the ``rightmost thinnest place'').  A symmetric
argument works for a vertex to the left of $l_b$.  As a result, we
have two points such that all structures between them is totally
decided by $F$ and $IN(F)$; in particular, it suffices to put $s$ in
any ``thinnest'' place in between.  This gives a $6$-way branching.

Putting together these steps, a high-level outline of our algorithms
is given in Fig.~\ref{fig:alg-interval-deletion}.  This concludes
Thms.~\ref{thm:main} and \ref{thm:main-2}.

\begin{figure*}[ht]
\setbox4=\vbox{\hsize28pc \noindent\strut
\begin{quote}
  \vspace*{-5mm} \small
  {\sc input}: a graph $G$ and a nonnegative integer $k$.
  \\
  {\sc output}: a set of at most $k$ modifications that transforms $G$
  into an interval graph; or ``NO.''

  0 \hspace*{2ex} {\bf if} $k<0$ {\bf then return} ``NO''; {\bf if}
  $G$ is an interval graph {\bf then return} $\emptyset$;
  \\
  1 \hspace*{2ex} [{\em only} for {\sc interval completion}] fill all
  holes of $G$;
  \\
  2 \hspace*{2ex} {\bf if} the quotient graph $Q$ defined by maximal
  strong modules of $G$ is edgeless {\bf then}
  \\
  \hspace*{7ex} solve each component individually;
  \\
  3 \hspace*{2ex} {\bf if} $Q$ is a clique {\bf then}
  \\
  \hspace*{7ex} {\bf if} there are two non-clique modules {\bf then}
  \\
  \hspace*{10ex} find a $4$-hole and {\bf branch} on disposing of it;
  \\
  \hspace*{7ex} {\bf else} solve the subgraph induced by the only
  non-clique module; \comment{$G$ is not an interval graph.}
  \\
  4 \hspace*{2ex} {\bf call decompose}($Q$);
  \\
  5 \hspace*{2ex} {\bf if } a small caw or short hole is found {\bf
    then branch} on disposing of it;
  \\
  $\setminus\!\!\setminus$ We have hereafter a caterpillar/olive-ring
  decomposition.
  \\
  6 \hspace*{2ex} {\bf if} a non-simplicial vertex of $Q$ represents a
  non-clique module {\bf then}
  \\
  \hspace*{7ex} find a $4$-hole and {\bf branch} on disposing of it;
  \\
  7 \hspace*{2ex} {\bf for each} module $M$ represented by a
  simplicial vertex of $Q$ {\bf do}
  \\
  \hspace*{7ex} solve the subgraph $G[M]$;
  \\
  8 \hspace*{2ex} {\bf if } the clique decomposition is not a hole
  {\bf then}
  \\
  \hspace*{7ex} find a minimal frame and {\bf branch} on disposing of
  it;
  \\
  9 \hspace*{2ex} [{\em not} for {\sc interval completion}] {\bf call}
  Prop.~\ref{lem:hole-cover}.

\end{quote} \vspace*{-6mm} \strut} $$\boxit{\box4}$$
\vspace*{-9mm}
\caption{Outline of our algorithms.}
\label{fig:alg-interval-deletion}
\end{figure*}

\paragraph{Organization.} The rest of the paper is organized as
follows.  Section~\ref{sec:module-preservation} relates modules to
optimum solutions of all the three problems, and proves
Thms.~\ref{thm:modules-induced-subgraph+}-\ref{thm:modules-supergraph+}
as well as
Thms.~\ref{thm:modules-induced-subgraph}-\ref{thm:modules-supergraph}.
Section~\ref{sec:forbidden-subgraphs} gives the characterization of
large caws and long holes in prime locally interval graphs, and proves
Thms.\ref{thm:1} and \ref{thm:2}.  Section~\ref{sec:olive-ring}
presents the details of decomposing prime graphs and proves
Thms.~\ref{thm:3}-\ref{lem:extend-decomposition}.
Section~\ref{sec:caw} presents the details on the disposal of large
caws.  Section~\ref{sec:algs} use all these results to complete the
algorithms.  Section~\ref{sec:remark} closes this paper by describing
some follow-up work and discussing some possible improvement and new
directions.

 \section{Modules}\label{sec:module-preservation}
This section is devoted to the proof of
Thms.~\ref{thm:modules-induced-subgraph}-\ref{thm:modules-supergraph}
and
Thms.~\ref{thm:modules-induced-subgraph+}-\ref{thm:modules-supergraph+}.
Each of these theorems comprises two assertions on modules of $G$.
The \emph{module preservation} asserts that a (maximal strong) module
$M$ (or its remnant after partial deletion) remains a module of the
optimum solution, and \emph{local optimum} asserts that the optimum
solution restricted to $G[M]$ is an optimum solution of itself, and
can be replaced by any optimum solution of $G[M]$.
Thms.~\ref{thm:modules-induced-subgraph} and
\ref{thm:modules-induced-subgraph+} (vertex deletions) turn out to be
quite straightforward.  As a matter of fact, a weaker version of them
has been proved and used in \cite{cao-14-interval-deletion}, and a
similar argument, which is based on the characterization of forbidden
induced subgraphs and the hereditary property, also works here.  On
the other hand, this approach does not seem to be adaptable to the
edge modifications problems.

Simple examples tell us that not all \misb s or \misp s preserve all
modules.  For example, consider the graph in
Fig.~\ref{fig:misp-modules} with only solid edges, which is obtained
from a $\dag$ as follows: the center $c$ is replaced by a clique of
$5$ vertices, and the shallow terminal is replaced by two nonadjacent
vertices $s_1$ and $s_2$.  A minimum completion to this graph must be
adding for each of $s_1$ and $s_2$ an edge to connect it to some
vertex at the bottom.  These two vertices do not need to be the same,
e.g., the dashed edges in Fig.~\ref{fig:misp-modules}; however,
$\{s_1, s_2\}$ is not a module of the resulting \misp.  On the one
hand, the subgraph induced by the \msm\ $\{s_1, s_2\}$ is not
connected; on the other hand, we may alternatively connect both $s_1$
and $s_2$ to the same vertex so that we obtain another \misp\ that
preserves $\{s_1, s_2\}$ as a module.  These two observations turn out
to be general: 1) if a module $M$ of a graph $G$ is not a module of
some \misp\ $\widehat G$ of $G$, then $G[M]$ \emph{must} be
disconnected, and 2) we \emph{can always} modify $\widehat G$ to
another \misp\ $\widehat G'$ of $G$ such that $M$ is a module of
$\widehat G'$.
\begin{SCfigure}[][h]
  \centering  \includegraphics{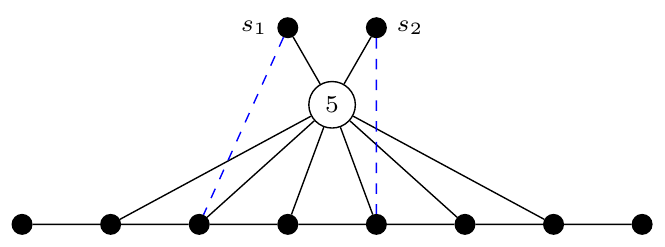} 
  \caption{The module $\{s_1, s_2\}$ is not preserved by a \misp\
    (dashed edges are added; number $5$ in a circle means a clique of
    $5$ vertices).}
  \label{fig:misp-modules}
\end{SCfigure}

To make it worse, a \misb\ may have to break some \msm s.  The
simplest example is a $4$-hole graph, which has two nontrivial
modules, but its \misb\ must be a simple path, which is prime.  Even
the connectedness does not help here.  For example, consider the graph
in Fig.~\ref{fig:misb-modules-1} (all edges, both solid and dashed),
which is obtained by completely connecting two induced paths $v_1 v_2
v_3 v_4$ and $u_1 u_2 u_3 u_4$.  A \misb\ of it has to be isomorphic
to Fig.~\ref{fig:misb-modules-1} after dashed edges deleted, which is
again prime.  Both examples contain some $4$-hole, which urges us to
study $4$-hole-free graphs.  We show that any $4$-hole-free graph has
a \misb\ that preserves all its modules.  It is worth stressing that
not all \misb s of a $4$-hole-free graph preserve all its modules,
e.g., the graph (with both solid and dashed edges) and its \misb\
(after dashed edges deleted) in Fig.~\ref{fig:misb-modules-2}.

\begin{figure*}[h]
  \centering
  \begin{subfigure}[b]{0.47\textwidth}
    \centering \includegraphics{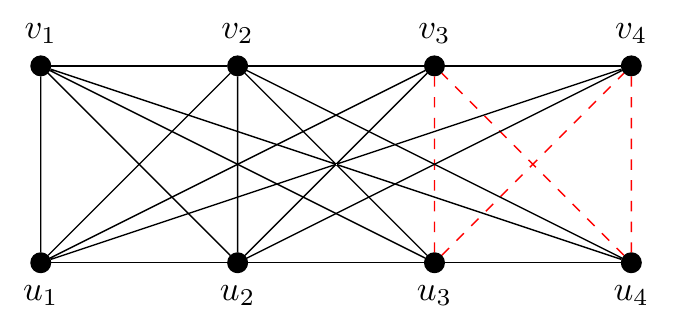} 
    \caption{The only two nontrivial modules $\{v_1,v_2,v_3,v_4\}$ and
      $\{u_1,u_2,u_3,u_4\}$, both connected, are not preserved by
      \emph{any} \misb.}
    \label{fig:misb-modules-1}
  \end{subfigure}
  \qquad
  \begin{subfigure}[b]{0.47\textwidth}
    \centering \includegraphics{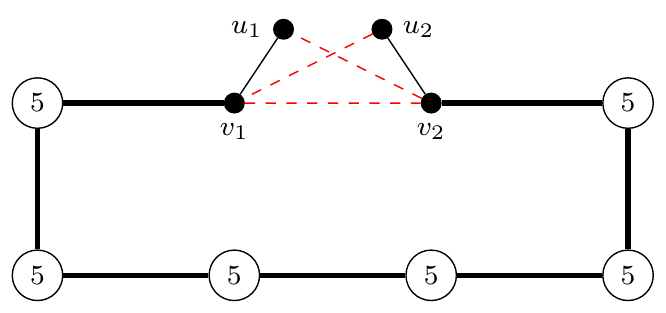} 
    \caption{The module $\{u_1,u_2\}$ is not preserved by \emph{some}
      \misb\ of a $4$-hole-free graph (number $5$ in a circle means
      a clique of $5$ vertices).}
    \label{fig:misb-modules-2}
  \end{subfigure}%  
  \caption{Modules not preserved by \misb s (dashed edges are deleted).}
  \label{fig:misb-modules}
\end{figure*}

The way we prove Thms.~\ref{thm:modules-subgraph+} and
\ref{thm:modules-supergraph+} is using interval models: we construct
an interval graph satisfying the claimed conditions by explicitly
giving an interval model for it.  For this purpose we need more
notation on interval models.  In an interval model, each vertex $v$
corresponds to a closed interval $I_v = [\lp{v}, \rp{v}]$, where
\lp{v} and \rp{v} are the left and right \emph{endpoints} of $I_v$,
respectively, and $\lp{v}< \rp{v}$.  An interval model is called
\emph{normalized} if no pair of distinct intervals in it shares an
endpoint; every interval graph has a normalized interval model.  All
interval models in this section are normalized.  For a subset $U$ of
vertices, we define $\lp{U} := \min_{v\in U} \lp{v}$ and $\rp{U} :=
\max_{v\in U} \rp{v}$.  Observe that if $U$ induces a connected
subgraph, then the interval $[\lp{U}, \rp{U}]$ is exactly the union of
$\{I_v: v\in U\}$.  Let $P$ be a set of points that are in an interval
$[\alpha,\beta]$.  By \emph{project}ing $P$ from $[\alpha,\beta]$ to
another interval $[\alpha',\beta']$ we mean the following operation:
\[
\rho\to \frac{\beta' - \alpha'}{\beta - \alpha}(\rho - \alpha) + \alpha'
\quad\text{for each } \rho\in P.
\]
In other words, each point in $[\alpha,\beta]$ is proportionally
shifted to a point in $[\alpha',\beta']$.  It is easy to verify that
all new points are in $[\alpha',\beta']$ and this operation retains
relations between every pair of points.  In particular, if we project
the endpoints of all intervals for $V(G)$, the set of new intervals
defines the same interval graph.

The following simple observation will be crucial for our arguments.
For any point $\rho$, we can find a positive value $\epsilon$ such
that the only possible endpoint of $\cal I$ in $[\rho-\epsilon,
\rho+\epsilon]$ is $\rho$.  Here the value of $\epsilon$ should be
understood as a function---depending on the interval model as well as
the point $\rho$---instead of a constant.

\subsection{Modules in \miib s}\label{sec:modules-induced-subgraph}
For any module $M$ and vertex set $U$ of $G$, the set $U\cap M$, if
not empty, is a module of the subgraph $G[U]$.  This property implies
the preservation of modules in {\em all} \miib s.  Therefore, for
Thms.~\ref{thm:modules-induced-subgraph} and
\ref{thm:modules-induced-subgraph+}, it suffices to prove their second
assertions, which follow from the following statement.  Recall that if
a $4$-hole contains precisely two vertices from some module $M$, then
neither $M$ nor $N(M)$ induces a clique.
\begin{lemma}
  Let $G[U]$ be a \miib\ of a graph $G$.  Let $M$ be a module of $G$
  such that at least one of $M$ and $N(M)$ induces a clique.  If
  $M\cap U\ne \emptyset$, then replacing $G[M\cap U]$ by any \miib\ of
  $G[M]$ in $G[U]$ gives a \miib\ of $G$.
\end{lemma}
\begin{proof}
  Suppose, for contradiction, that the new graph $G[U']$ is not an
  interval graph.  From $G[U']$ we can find a subgraph $X$ in $\cF_I$,
  which must intersect both $M$ and $V(G)\setminus M$.  Since at least
  one of $M\cap U'$ and $N(M)\cap U'$ induces a clique, $X$ contains
  exactly one vertex of $M$; let it be $x$.  By assumption, there
  exists a vertex $x'\in M\cap U$ (possibly $x'=x$); let $X' =
  X\setminus \{x\}\cup \{x'\}$.  Clearly, $X'\subseteq U$, but $G[X']$
  is isomorphic to $G[X]$, hence in $\cF_I$, contradicting that $G[U]$
  is an interval graph.
\end{proof}

\subsection{Modules in \misb s}
Before the proof of Thm.~\ref{thm:modules-subgraph+}, we show a
stronger result on clique modules.
\begin{lemma}\label{thm:clique-modules-subgraph+}
  A clique module $M$ of a graph $G$ is also a clique module of any
  \misb\ $\underline G$ of $G$.
\end{lemma}
\begin{proof}
  Let $C$ be the component of $\underline G[M]$ such that
  $|N_{\underline G}(C)|$ attains the maximum value among all
  components of $\underline G[M]$.  We modify a given normalized
  interval model ${\cal I} = \{I_v:v\in V(G)\}$ for ${\underline G}$
  as follows.  Let $\alpha = \lp{C}$ and $\beta = \rp{C}$.  For each
  $v\in M$, we set $\lpp{'}{v}$ to a distinct value in
  $(\alpha-\epsilon, \alpha)$, and set $\rpp{'}{v}$ to a distinct
  value in $(\beta, \beta+\epsilon)$.  For each $u\in V(G)\setminus
  M$, we set $I'_u = I_u$.  Let $\underline G'$ be the interval graph
  defined by ${\cal I}'$.  By construction, a vertex $u\in
  V(G)\setminus M$ is adjacent to $C$ in $\underline G'$ if and only
  if it is in $N_{\underline G}(C)$.  Since $N_{\underline
    G}(C)\subseteq N_G(M)$, we have $\underline G'\subseteq G$.  For
  each $v\in M$, it holds that
  \[
    |N_{\underline G}(v)\setminus M|\le |N_{\underline G}(C)\setminus M|
    = |N_{\underline G}(C)| = |N_{\underline G'}(v)\setminus M|.
  \]
  On the other hand, $M$ induces a clique in $\underline G'$.  They
  together imply $|\underline G|\le |\underline G'|$, while the
  equality is only attainable when $\underline G[M]$ is a clique,
  hence $C = M$, and $M$ is completely connected to $N_{\underline
    G}(C) = N_{\underline G}(M)$.  Therefore, $\underline G' =
  \underline G$, and this verifies the lemma.
  \end{proof}

\paragraph{Theorem~\ref{thm:modules-subgraph+} (restated).}
Let $G$ be a graph of which every $4$-hole is contained in some
maximal strong module.  There exists a \misb\ $\underline G$ of $G$
such that every maximal strong module $M$ of $G$ is a module of
$\underline G$, and replacing $\underline G[M]$ by any \misb\ of
$G[M]$ in $\underline G$ gives a \misb\ of $G$.
\begin{proof}
  As a consequence of Lem.~\ref{thm:clique-modules-subgraph+}, it
  suffices to consider the case when $G[M]$ is not a clique, and then
  $N_G(M)$ must induce a clique of $G$.  Let ${\cal I} = \{I_v:v\in
  V(G)\}$ be a normalized interval model for ${\underline G}$.
  \begin{claim}
    For any component $C$ of $\underline G[M]$, the set $N_{\underline
      G}(C)$ induces a clique of $\underline G$.
  \end{claim}
  \begin{proof}
    Supposing the contrary, we construct an interval graph $\underline
    G'$ with $\underline G\subset \underline G'\subseteq G$ as
    follows.  Let $x,y$ be a pair of vertices in $N_{\underline G}(C)$
    such that $I_x\cap I_y=\emptyset$, i.e., $x\not\sim y$ in
    ${\underline G}$.  Without loss of generality, assume that $I_x$
    is to the left of $I_y$, and let $\rho$ be an arbitrary point in
    between, i .e., $\rp{x}<\rho< \lp{y}$.  For every $v\in
    N_{\underline G}(C)$, we extend the interval $I_v$ to include
    $\rho$: if $I_v$ is to the left of $\rho$, we set $\rpp{'}{v}$ to
    be a distinct point in ($\rho, \rho+\epsilon$); if $I_v$ is to the
    right of $\rho$, we set $\lpp{'}{v}$ to be a distinct point in
    ($\rho-\epsilon, \rho$).  We use the graph defined by the set of
    new intervals as $\underline G'$.  To see $\underline G\subseteq
    \underline G'$, note that all intervals are extended only; to see
    $\underline G\ne \underline G'$, note that $xy$ is an edge in
    $\underline G'$ but not in $\underline G$.  Every edge in
    $E(\underline G')\setminus E(\underline G)$ is always incident to
    $N_{\underline G}(C)$, which is a subset of $N_G(M)$, and hence
    exists also in $E(G)$.  Therefore, $\underline G'$ is an interval
    subgraph of $G$ with strictly more edges than $\underline G$,
    which is impossible.  
    \renewcommand{\qedsymbol}{$\lrcorner$}
  \end{proof}

  Let $\underline {G_M}$ be any \misb\ of $G[M]$.  We modify
  $\underline G$ first to make $M$ satisfy the claimed condition.  Let
  $C$ be the component of $\underline G[M]$ such that $|N_{\underline
    G}(C)|$ attains the maximum value among all components of
  $\underline G[M]$.  We have seen that $N_{\underline G}(C)$ induces
  a clique in $\underline G$.  The intersection of all intervals
  $\{I_v: v\in N_{\underline G}(C)\}$ is thus nonempty; let it be
  $[\alpha,\beta]$.  Since $C$ is connected, $\bigcup_{v\in C} I_v =
  [\lp{C}, \rp{C}]$.  The interval $[\lp{C}, \rp{C}]$ intersects
  $[\alpha,\beta]$, and we can choose a common point $\xi$ in them.
  We construct another graph ${\underline G}'$ by projecting an
  interval model for $\underline {G_M}$ into $[\xi-\epsilon,
  \xi+\epsilon]$.  The graph represented by the set of new intervals
  will be the sought-after interval graph $\underline G'$.

  We now verify $\underline G'\subseteq G$ and $||\underline
  G'||\ge||\underline G||$.  On the one hand, $V(G)\setminus M$
  induces the same subgraph in $\underline G'$ and $\underline G$.  On
  the other hand, by assumption, $\underline G'[M] = \underline {G_M}$
  is an interval subgraph of $G[M]$ and has no more edges than
  $\underline G[M]$.  Therefore, it suffices to consider edges between
  $M$ and $V(G)\setminus M$.  For $\underline G'$, there edges are
  $M\times N_{\underline G}(C)$.  Since $N_{\underline G}(C)$ is a
  subset of $N_G(M)$, it holds that $\underline G'\subseteq G$.  By
  selection of $C$ (i.e., $N_{\underline G}(C)$ has the largest size),
  $||\underline G'||\ge||\underline G||$.  

  We have now constructed a \misb\ of $G$ where $M$ satisfies the
  claimed conditions.  Only intervals for vertices in $M$ are changed,
  and thus this operation can be successively applied on the maximal
  strong modules of $G$ one by one.  If a module already satisfies the
  conditions, then it remains true after modifying other modules.
  Therefore, repeating this process will derive a claimed \misb\ of
  $G$.
\end{proof}

As explained below, this settles Thm.~\ref{thm:modules-subgraph} as
well.  Thm.~\ref{thm:modules-subgraph} will not be directly used in
this paper, and thus the reader may safely skip the following proof
without losing track of the development of our algorithms.
\begin{lemma}\label{lem:modules-subgraph-equal}
  Thms.~\ref{thm:modules-subgraph} and \ref{thm:modules-subgraph+} are
  equivalent.
\end{lemma}
\begin{proof}
  We show first that Thm.~\ref{thm:modules-subgraph} implies
  Thm.~\ref{thm:modules-subgraph+}.  Let $G$ be a graph of which every
  $4$-hole is contained in some maximal strong module.  Let $\{M_1,
  \cdots, M_p\}$ be the set of maximal strong modules of $G$, and let
  $Q$ be the quotient graph defined by them.  Let $\underline G$ be a
  \misb\ of $G$; for each $1\le i\le p$, we replace $G[M_i]$ by
  $\underline G[M_i]$, which is an interval subgraph; let $G'$ denote
  the obtained graph.  Clearly, $G'$ is a subgraph of $G$, and
  $\underline G$ is a \misb\ of $G'$.  Moreover, every $M_i$ remains a
  module of $G'$ and thus $Q$ is a quotient graph of $G$.  By
  Lem.~\ref{thm:clique-modules-subgraph+}, $G'[M_i]$ is not a clique
  if and only if $G[M_i]$ is not a clique.  Thus, $G'$ is
  $4$-hole-free, and by Thm.~\ref{thm:modules-subgraph}, there is a
  \misb\ $\underline G'$ of $G'$ such that $\underline G'[M_i] =
  G'[M_i] = \underline G[M_i]$ for each $1\le i \le p$.  This implies
  that $\underline G'[M_i]$ is a \misb\ of $G[M_i]$, and the
  substitutability follows from Cor.~\ref{lem:lig-and-max-modules}.
  Moreover, since $\underline G'$ and $\underline G$ are both
  \misb s of $G'$, they have the same size, which implies that
  $\underline G'$ is a \misb\ of $G$ as well.  This verifies that
  $\underline G'$ satisfies the claimed conditions of
  Thm.~\ref{thm:modules-subgraph+}, and concludes this direction.

  We now verify the other direction.  Let $G$ be a $4$-hole-free
  graph.  Note that every strong module different from $V(G)$ is a
  subset of some maximal strong module $M$, and a strong module in
  $G[M]$ \cite{habib-10-survey-md}.  We first use inductive reasoning
  to show that the assertions i) and ii) of
  Thm.~\ref{thm:modules-subgraph} hold for every strong module of $G$.
  The base case is trivial: the largest strong module is $V(G)$.  The
  inductive steps follow from Thm.~\ref{thm:modules-subgraph+}: since
  every strong module $M$ induces a $4$-hole-free subgraph $G[M]$, its
  quotient graph trivially satisfies the condition of
  Thm.~\ref{thm:modules-subgraph+}.  This settles all strong modules,
  and then we consider modules that are not strong.  Such a module $M$
  is composed of more than one strong modules, and they are either
  pairwise adjacent or pairwise nonadjacent.  In the first case,
  (noting that graph contains no $4$-hole,) at most one of these
  strong modules is nontrivial.  In the second case, they are
  different components of $G[M]$.  Both cases are straightforward.
\end{proof}

\subsection{Modules in \misp s}
Before the proof of Thm.~\ref{thm:modules-supergraph+}, we show a
stronger result on {\em connected modules}, i.e., modules inducing
connected subgraphs, of a graph with respect to its \misp s.  For a
subset $U$ of vertices, we denote by $\mathring N_G(U)$ the set of
{common neighbors} of $U$, i.e., $\mathring N_G(U) := \bigcap_{v\in U}
N(v)$.  Note that $\mathring N_G(U) \subseteq N_G(U)$, and the
equality is attained if and only if $U$ it is a module of $G$.
\begin{theorem}\label{thm:preserving-modules}
  Let $\widehat G$ be a \misp\ of a graph $G$.  Every connected module
  $M$ of $G$ is a module of $\widehat G$, and if $\widehat G[M]$ is
  not a clique, then replacing $\widehat{G}[M]$ by any \misp\ of
  $G[M]$ in $\widehat{G}$ gives a \misp\ of $G$.
\end{theorem}
\begin{proof}
  The statement holds vacuously if $M$ consists of a single vertex or
  a component; hence we may assume $|M| > 1$ and $N_G(M)\ne\emptyset$.
  Let ${\cal I} = \{I_v:v\in V(G)\}$ be a normalized interval model
  for ${\widehat G}$.  We define $\alpha = \min_{v\in M}\rp{v}$ and
  $\beta = \max_{v\in M}\lp{v}$.  Let $x$ and $y$ be the vertices such
  that $\rp{x} = \alpha$ and $\lp{y} = \beta$; possibly $x = y$, which
  is irrelevant in the following argument.  By assumption, $\mathring
  N_G(M) = N_G(M) \subseteq \mathring N_{\widehat G}(M) \subseteq
  N_{\widehat G}(M)$, and the first assertion is equivalent to
  $N_{\widehat G}(M) = \mathring N_{\widehat G}(M)$.  Suppose, for
  contradiction, that there exists $z\in N_{\widehat G}(M)\setminus
  \mathring N_{\widehat G}(M)$, then we modify $\cal I$ into another
  set of intervals ${\cal I}' = \{I'_v: v\in V(G)\}$.  We argue that
  the interval graph $\widehat G'$ defined by $\cal I'$ is a
  supergraph of $G$ and has strictly smaller size than $\widehat G$.
  This contradicts the fact that $\widehat G$ is a minimum interval
  supergraph of $G$, and thus the assertion must be true.
  Since $\cal I$ is normalized, $\alpha\ne \beta$.

  {\em Case 1, $\alpha > \beta$.}  Then $M$ induces a clique of
  ${\widehat G}$.  We have $[\beta,\alpha]\subseteq I_v$ for every
  $v\in M$, and $I_u\cap [\beta,\alpha]\ne\emptyset$ for every $u\in
  \mathring N_{\widehat G}(M)$.  We construct ${\cal I'}$ as follows.
  We keep $\rpp{'}{x} = \alpha$ and $\lpp{'}{y} = \beta$; for $v\in
  M\setminus \{y\}$, we set $\lpp{'}{v}$ to a distinct value in
  $(\beta-\epsilon, \beta)$; and for $v\in M\setminus \{x\}$, we set
  $\rpp{'}{v}$ to a distinct value in $(\alpha, \alpha+\epsilon)$.
  For each $u\in V(G)\setminus M$, we set $I'_u = I_u$.  In the graph
  $\widehat G'$ represented by $\cal I'$, the subgraph induced by $M$
  is a clique; the subgraph induced by $V(G)\setminus M$ is the same
  as $\widehat G - M$; and $M$ is completely connected to
  $N_G(M)\subseteq \mathring N_{\widehat G}(M)$.  This verifies
  $G\subseteq \widehat G'$.  On the other hand, by the construction of
  the new intervals, $I'_v\subseteq I_v$ holds for every vertex $v$,
  and thus $N_{\widehat G'}(v) \subseteq N_{\widehat G}(v)$; it
  follows that $\widehat G'\subseteq \widehat G$.  By assumption that
  $z\not\in \mathring N_{\widehat G}(M)$, the interval $I'_z$($ =
  I_z$) is either to the left of $\beta$ or to the right of $\alpha$,
  and then from the choice of $\epsilon$ we can conclude that it is
  either to the left of $\beta-\epsilon$ or to the right of
  $\alpha+\epsilon$.  As a result, $z\not\sim M$ in $\widehat G'$ and
  thus $\widehat G'\ne \widehat G$; in other words, $\widehat G'$ is a
  proper subgraph of $\widehat G$.  This contradiction verifies the
  first assertion for the case $\alpha>\beta$.

\begin{figure*}[h!]
  \centering  \includegraphics{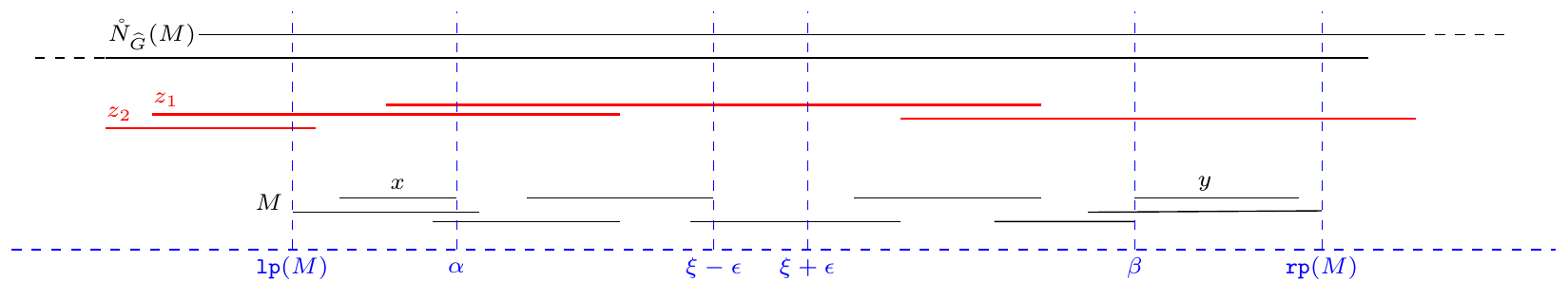} 
  \caption{Case 2 of proof of Thm.~\ref{thm:preserving-modules}.\protect\\
    A thick (red) interval breaks the modularity of $M$.}
  \label{fig:interval-representation}
\end{figure*} {\em Case 2, $\alpha < \beta$.}  Then $x\ne y$ and
$x\not\sim y$.  We have $I_v\cap [\alpha,\beta]\ne\emptyset$ for every
$v\in M$, and $[\alpha,\beta]\subset I_u$ for every $u\in \mathring
N_{\widehat G}(M)$.  (See Figure~\ref{fig:interval-representation}.)
We construct ${\cal I'}$ as follows.  For each point $\rho\in
[\alpha,\beta]$, its thickness $\theta_\rho$ is defined to be the
number of vertices of $V(G)\setminus M$ whose intervals contain this
point, i.e., $\theta_\rho = |\{v\in V(G)\setminus M:\rho\in I_v\}|$.
Let $\xi$ be a point in $[\alpha,\beta]$ that attains the minimum
thickness; without loss of generality, we may assume $\xi$ is
different from any endpoint of intervals in $\cal I$.  For each $v\in
M$, we set $I'_v$ by projecting $I_v$ from $[\alpha,\beta]$ to
$[\xi-\epsilon, \xi+\epsilon]$.  For each $u\in V(G)\setminus M$, we
set $I'_u = I_u$.  Let $\widehat G'$ be represented by $\cal I'$.
The subgraphs induced by $M$ and $V(G)\setminus M$ are the same as
$\widehat G[M]$ and $\widehat G - M$ respectively.  Hence, we only
need to consider edges between $M$ and $V(G)\setminus M$.  For each
$u\in N_G(M)$, the interval $I'_u$($=I_u$) contains $[\alpha,\beta]$
which contains $[\xi-\epsilon, \xi+\epsilon]$ in turn.  Therefore,
$N_G(M) \subseteq \mathring N_{\widehat G'}(M)$, and it follows that
$G\subseteq \widehat G'$.

It remains to verify $||\widehat G'||< ||\widehat G||$, which is
equivalent to 
\begin{equation}
  \label{eq:module}
  |E(\widehat G')\cap (M\times V(G)\setminus M)|< |E(\widehat G)\cap
  (M\times V(G)\setminus M)|.
\end{equation}
By the selection of $\xi$ and $\epsilon$, no interval in $\cal I$ has
an endpoint in $[\xi-\epsilon, \xi+\epsilon]$.  Thus, for each $u\in
V(G)\setminus M$, the interval $I'_u$ ($= I_u$) contains $\xi$ if and
only if $[\xi-\epsilon, \xi+\epsilon] \subset I'_u$.  As a result,
$N_{\widehat G'}(M) = \mathring N_{\widehat G'}(M) = \{u\in
V(G)\setminus M:\xi\in I_u\}$, and the left-hand side of
\eqref{eq:module} is equal to $|M|\cdot \theta_\xi$.  We now consider
the right-hand side of \eqref{eq:module}, i.e., the number of edges
between $M$ and $V(G)\setminus M$ in ${\widehat G}$.  For every $v\in
M$, the interval $I_v$ contains some point $\rho \in [\alpha,\beta]$,
which means $v$ is adjacent to all vertices in $\{u\in V(G)\setminus
M:\rho\in I_u\}$.  By the selection of $\xi$, it holds that
$|N_{\widehat G}(v)\setminus M|\ge \theta_\xi = |N_{\widehat
  G'}(v)\setminus M|$.  For the correctness of \eqref{eq:module}, it
suffices to show that this is strict for at least one vertex in $M$.

As $z\not\in \mathring N_{\widehat G}(M)\cup M$, the interval $I'_z$
($= I_z$) does not contain $[\alpha,\beta]$ (see the thick/red edges
in Figure~\ref{fig:interval-representation}).  If $\alpha < \rp{z} <
\beta$ (see $z_1$ in Figure~\ref{fig:interval-representation}), then
$\theta_{\rp{z}}> \theta_{\rp{z} + \epsilon}\ge \theta_{\xi}$.  As
$\widehat G[M]$ is connected, there exists a vertex $v\in M$ such that
${\rp{z}}\in I_v$, and we are done.  A symmetric argument applies when
$\alpha < \lp{z} < \beta$.  Hence we may assume there exists no vertex
$u\in V(G)\setminus M$ such that $I'_u$ ($=I_u$) has an endpoint in
$[\alpha,\beta]$.  Suppose now $\rp{z} < \alpha$ (see $z_2$ in
Figure~\ref{fig:interval-representation}), and let $v\in N_{\widehat
  G}(z)\cap M$.  By the selection of $\alpha$, it holds that
$\alpha\in I_v$ and $v$ is adjacent to all vertices in $\{u\in
V(G)\setminus M: \alpha\in I_u\}$, which does not contain $z$.
Therefore, $|N_{\widehat G}(v)\setminus M| \ge 1 + \theta_\alpha >
\theta_\xi$.  A symmetric argument applies if $\lp{z} > \beta$.  This
verifies \eqref{eq:module} and finishes the proof of case 2 and the
first assertion.

For the second assertion, we may assume that $\widehat G[M]$ is not a
clique.  Then $N_{\widehat G}(M)$ must be a clique, and thus we can
take the nonempty intersection of all intervals for $N_{\widehat
  G}(M)$; let it be $[\alpha,\beta]$.  Then replacing intervals
$\{I_v: v\in M\}$ by an interval model of any minimum interval
supergraph of $G[M]$ projected to $[\alpha,\beta]$ will makes another
interval model.  Since $N_{G}(M)\subseteq N_{\widehat G}(M)$, which
remains common neighbors of $M$, the graph defined by this new model
is clearly an interval supergraph of $G$.  It is easy to verify that
its size is no larger than $\widehat G$.  This completes the proof.
\end{proof} 

{Thm.}~\ref{thm:preserving-modules} will ensure preservation of any
connected module in perpetuity.  Observing that a graph $\widehat G$
is a minimum interval supergraph of $G$ if and only if it is a minimum
interval supergraph of any graph $G'$ satisfying $G \subseteq G'
\subseteq \widehat G$, it can be further strengthened to:
\begin{corollary}\label{lem:preserving-modules}
  Let $\widehat G$ be a minimum interval supergraph of graph $G$.  A
  connected module $M$ of any graph $G'$ satisfying $G \subseteq G'
  \subseteq \widehat G$ is a module of $\widehat G$.
\end{corollary}

We are now ready to prove Thm.~\ref{thm:modules-supergraph+}, which is
restated below.
\paragraph{Theorem~\ref{thm:modules-supergraph+} (restated).}
  There is a \misp\ $\widehat G$ of $G$ such that every maximal strong
  module $M$ of $G$ is a module of $\widehat G$, and if $\widehat
  G[M]$ is not a clique, then replacing $\widehat G[M]$ by any \misp\
  of $G[M]$ in $\widehat G$ gives a \misb\ of $G$.
\begin{proof}
  It suffices to consider maximal strong modules that are not
  connected in $G$.  Let $M$ be such a module and let $C$ be a
  component of $G[M]$.  By definition, $N_G(C) = N_G(M) = \mathring
  N_G(M)$, and $C$ is a connected module of $G$.  Thus, by
  Thm.~\ref{thm:preserving-modules}, $C$ remains a module of $\widehat
  G$.  Let $\{I_v: v\in V(G)\}$ be a normalized interval model for
  $\widehat G$.

  Assume first that $\widehat G$ contains no edge between different
  components of $G[M]$.  Let $x,y$ be the vertices in $N_{\widehat
    G}(C)$ such that $\rp{x}$ and $\lp{y}$ are the smallest and
  largest, respectively.  Suppose $x\not\in N_G(C)$, then we can
  delete edges between $x$ and $C$ to obtain an interval supergraph of
  $G$, which is strictly less edges than $\widehat G$: an interval
  model for the new graph can be obtained by setting $\lp{v} = \rp{x}
  +\epsilon$ or $\rp{v} = \lp{y} - \epsilon$ for every $v\in C$; here
  we are using the fact that the model is normalized.  A symmetrical
  argument applies to $y$, and thus both of $x$ and $y$ must be in
  $N_G(C)$.  We now argue that $N_{\widehat G}(C)$ must induce a
  clique in $\widehat G$.  Suppose, for contradiction, that
  $N_{\widehat G}(C)$ does not induces a clique, then $x\not\sim y$ in
  ${\widehat G}$.  Since $N_G(C) = N_G(M) = \mathring N_G(M)$, for
  every $v\in M\setminus C$, the interval $I_v$ fully contain
  $[\rp{x}, \lp{y}]$.  However, $v$ is then adjacent to $C$ in
  $\widehat G$, which contradicts the assumption.  Therefore,
  $N_{\widehat G}(C)$ induces a clique of $\widehat G$.  We construct
  another \misp\ $\widehat G'$ of $G$ satisfying the specified
  conditions as follows.  Let $C_0$ be the component of of $G[M]$ such
  that $N_{\widehat G}(C_0)$ has the minimum size.  We choose a point
  $\xi$ contained in all intervals for $N_{\widehat G}(C_0)$, and
  project an interval model of any \misp\ of $G[M]$ to $[\xi
  -\epsilon, \xi +\epsilon]$.  Let $\widehat G'$ be the graph defined
  by this new model.  By construction, $N_{\widehat G'}(M) =
  N_{\widehat G}(C_0)$, which fully contains $N_G(C_0)$, and hence
  $G\subseteq \widehat G'$.  The selection of $C_0$ implies that
  $||\widehat G'|| \le ||\widehat G||$.

  Assume now that $E(\widehat G)$ contains edges between different
  components of $G[M]$.  We may add these edges first, and then
  consider the resulted graph $G'$.  According to
  Cor.~\ref{lem:preserving-modules}, $\widehat G$ is a \misp\ of $G'$,
  and this reduces to the previous case.

  Therefore, if $\widehat G[M]$ is not a clique, then it can be
  replaced by any \misp\ of $G[M]$.  This operation only change
  intervals for vertices in $M$, it can be successively applied on the
  maximal strong modules of $G$ one by one.  After that, the condition
  holds for every of them.
\end{proof}
With a similar argument as Lem.~\ref{lem:modules-subgraph-equal}, we
can derive the equivalence between Thm~\ref{thm:modules-supergraph}
and Thm~\ref{thm:modules-supergraph+}.
\begin{lemma}
  Thms.~\ref{thm:modules-supergraph} and \ref{thm:modules-supergraph+}
  are equivalent.
\end{lemma}

 \section{Large caws and long holes in prime locally interval graphs}
\label{sec:forbidden-subgraphs}
This section is devoted to the study of large caws and long holes in
\lig s; they are both minimal witnesses of asteroidal triples.
Thms.~\ref{thm:1} and \ref{thm:2}, as well as some of their
implications, will be derived as a result of this study.

We give a symbol for each vertex in a large caw, i.e., a $\dag$ or
$\ddag$ (see Fig.~\ref{fig:at}).  Recall that $s$ is the {shallow
  terminal}, and the removal of $N[s]$ from this caw leaves an induced
path.  This path connects the other two terminals $l,r$, called
\emph{base terminals}.  The neighbor(s) $c_1,c_2$ of $s$ are the
\emph{center(s)}, and all other vertices, $\{b_1,\dots, b_d\}$, are
called \emph{base vertices}.  The \stpath{b_1}{b_d} through all base
vertices is called the \emph{base}, denoted by $B$.  The center(s) and
base vertices are called non-terminal vertices.  Note that a $\dag$
has only one center, and both $c_1$ and $c_2$ refer to it; while in a
$\ddag$, centers $c_1$ and $c_2$ are decided in the way that
$c_1\not\sim r$ and $c_2\not\sim l$.  In summary, we use
$(s:c_1,c_2:l,B,r)$ to denote both kinds of large caws.  This uniform
notation will greatly simplify our presentation, and another, possibly
more important, benefit is the structural information it reveals.  For
the sake of notational convenience, we will also use $b_0$ and
$b_{d+1}$ to refer to the base terminals $l$ and $r$, respectively,
even though they are not part of the base $B$.

All indices of vertices in a hole $H$ should be understood as modulo
$|H|$, e.g., $h_{-1} = h_{|H|-1}$.  We define the ordering $h_{-1},
h_0, h_1, \dots$ of traversing $H$ to be \emph{clockwise}, and the
other to be \emph{counterclockwise}.  In other words, vertices $h_1$
and $h_{-1}$ are the clockwise and counterclockwise successors,
respectively, of $h_0$; and edges $h_0 h_{1}$ and $h_0 h_{-1}$ are
clockwise and counterclockwise, respectively, from $h_0$.

In this section, we will use $\hv{v}$ as a shorthand for $N[v] \cap
V(H)$, i.e., the closed neighborhood of $v$ in $H$.  The notation
$B^+\langle v\rangle$, where $B^+$ is the \stpath{l}{r} $l b_1 \cdots
b_d r$ in a large caw, is defined analogously.  Note that $v\in
\hv{v}$ if and only if $v\in H$, and then $|\hv{v}| = 3$.  The
following subroutine will be used to find $\hv{v}$ and similar
indices.

\begin{proposition}\label{lem:find-indices}
  Let $U$ be a subset of ordered vertices.  For any vertex $v$, we can
  in $O(d(v))$ time compute an ordered list that contains $i$ if and
  only if $u_i\in N[v]$.
\end{proposition}
\begin{proof}
  Let $U = \{u_0, u_2, \ldots, u_{|U|-1}\}$.  We pre-allocate a list
  {\sf IND} of $d(v)$ slots, initially all empty.  For each neighbor
  of $v$, if it is $u_i$, then add $i$ into the next empty slot of
  {\sf IND}.  After all neighbors of $v$ have been checked, we shorten
  {\sf IND} by removing empty slots from the end, which leaves
  $|N[v]\cap U|$ slots.  We radix sort these indices and return the
  sorted list.
\end{proof}

Another common task of our algorithm is to explore a connected graph
to find a vertex with some specific property; we use breadth-first
search (BFS), which outputs a rooted tree.  Recall that except the
starting vertex $o$, i.e., the root, every other vertex $v$ is
explored after a specific neighbor of $v$ (i.e., its earliest explored
neighbor), which is denoted by {\sf prev}($v$).  Between $v$ and $o$
there is a unique path in the BFS tree, which is a shortest
\stpath{v}{o} and can be retrieved using the {\sf prev} function.  BFS
can be used to find some vertex that has some property and has the
shortest distance to the starting vertex.  As long as the total time
for the property test for all vertices can be done in linear time, the
whole process remains linear-time.  In particular, this is the case
when the property for vertex $v$ can be tested in $O(d(v))$ time.
Another variation of BFS we need is that, instead of starting from a
single vertex, we may start from a subset $X$ of vertices; it can be
viewed as applying BFS on the extended graph with a new vertex
adjacent to all vertices in $X$.  We can then use the {\sf prev}
function to retrieve a shortest path from any vertex to $X$.

\subsection{Shallow terminals of chordal asteroidal witnesses}
\label{sec:shall-terminals}
The most important vertex of a large caw is its shallow terminal $s$,
which is the focus of Theorem~\ref{thm:1}.
\paragraph{Theorem~\ref{thm:1} (restated).}
  Let $W$ be a large caw of a prime graph $G$.  We can in $O(||G||)$
  time find a subgraph of $G$ in $\cF_{LI}$ if the shallow terminal of
  $W$ is non-simplicial in $G$.
\begin{proof}
  Let ($s:c_1,c_2:l,B,r$) be the caw.  Before presenting the main
  procedure of this proof, we introduce two subroutines.  They apply
  to two special structures that arise frequently in the main
  procedure.  The outcomes of both subroutines are always subgraphs in
  $\cF_{LI}$; hence the detection of either of these two structures
  will suffice to terminate the main procedure.  Both structures
  involve some vertex $x\in N(s)\setminus \{c_1,c_2\}$.  Note that
  $c_1,c_2\in N(s)$ and $x\not\in B$.

  In the first structure $x$ is nonadjacent to one or both of $c_1$
  and $c_2$.  We apply subroutine A (Fig.~\ref{fig:thm:1:sa}) when
  $x\not\sim c_2$, and the case $x\not\sim c_1$ can be handled in a
  symmetric way.
  \begin{figure}[h!]  \centering\small
  \setbox4=\vbox{\hsize.9\linewidth \noindent\strut
    \hspace*{0.03\linewidth}\parbox[t]{0.87\linewidth}{ 

      1 \hspace*{2ex} {\bf if} $x\sim b_1$ {\bf then return} $4$-hole
      $x s c_2 b_1 x$;
      \\
      2 \hspace*{2ex} {\bf if} $x\sim b_d$ {\bf then return} $4$-hole
      $x s c_2 b_d x$;
      \\
      2 \hspace*{2ex} {\bf if} $x\sim l$ {\bf then return} $5$-hole
      $x s c_2 b_1 l x$;
      \\
      4 \hspace*{2ex} {\bf if} $x\sim r$ {\bf then return} $5$-hole
      $x s c_2 b_d r x$ or $4$-hole $x s c_2 r x$; \comment{$\dag$
        or $\ddag$ respectively.}
      \\
      5 \hspace*{2ex} {\bf if} $x \sim c_1$ {\bf then return} whipping
      top $\{r,c_2,s,x,c_1,l,b_1\}$; \hfill$\setminus\!\!\setminus$
      {\em only $\ddag$.}
      \\
      6 \hspace*{2ex} {\bf return} long claw $\{x,s,c,b_1,l,b_d,r\}$
      or net $\{x,s,l,c_1,r,c_2\}$; \comment{$\dag$ or $\ddag$
        respectively.}  

    }\strut} $$\boxit{\box4}$$ \vspace*{-6mm}
  \caption{Subroutine A for the proof of Thm.~\ref{thm:1}.}
  \label{fig:thm:1:sa}
  \end{figure}

  The second structure has a maximal sub-path ($b_p\cdots b_q$) of
  $B^+\langle x\rangle$ that contains at most two vertices of $B$,
  i.e., $q-p\le 2$ and the equality can only be attained when $p=0$ or
  $q=d+1$.  (Recall that $b_0$ and $b_{d+1}$ do not belong to $B$.)
  The maximality implies that if $b_p$ (resp., $b_q$) is not an end of
  $B^+$, i.e., $p\ge 1$ (resp., $q\le d$), then $x\not\sim b_{p-1}$
  (resp., $x\not\sim b_{q+1}$).  We apply subroutine B
  (Fig.~\ref{fig:thm:1:sb}).

  \begin{figure}[h!]
  {\centering\small
  \setbox4=\vbox{\hsize.9\linewidth \noindent\strut
    \hspace*{0.03\linewidth}\parbox[t]{0.87\linewidth}{ 

      0 \hspace*{2ex} {\bf if} $x$ is nonadjacent to $c_1$ or $c_2$
      {\bf then call} subroutine A;
      \\
      1 \hspace*{2ex} {\bf if} $p=0$ {\bf then}
      \\
      \hspace*{7ex} {\bf if} $q=0$ {\bf then return} $4$-hole ($x l
      b_1 c_2 x$);
      \\
      \hspace*{7ex} {\bf return} sun $\{s, x, c_2, l, b_1, b_2\}$ or
      rising sun $\{s, x, c_2, l, b_1, b_2, b_3\}$;
      \comment{$q = 1,2$.}
      \\
      2 \hspace*{2ex} {\bf if} $p=1$ {\bf then}
      \\
      \hspace*{7ex} {\bf if} $x\sim b_{q+2}$ {\bf then return}
      $4$-hole ($x b_q b_{q+1} b_{q+2} x$);
      \\
      \hspace*{7ex} {\bf if} $q=1$ {\bf then return} whipping top
      $\{l, b_1, x, s, c_2, b_3, b_2\}$;
      \\
      \hspace*{7ex} {\bf return} $\{s, x, l, b_1\ b_{2}, b_3\}$;
      \comment{$q = 2$.}
      \\
      3 \hspace*{2ex} {\bf if} $q=d$ or $d+1$ {\bf then} {\sf
        symmetric as steps 1,2};
      \\
      4 \hspace*{2ex} {\bf else}  \comment{$1<p\le q< d$.}
      \\
      \makebox[7ex][l]{} {\bf if} $x\sim b_{q+2}$ {\bf then return}
      $4$-hole ($x b_q b_{q+1} b_{q+2} x$);
      \\
      \hspace*{7ex} {\bf if} $q=p$ {\bf then return} long claw
      $\{b_{p-2}, b_{p-1}, b_p, s, x, b_{p+2}, b_{p+1}\}$;
      \\
      \makebox[7ex][l]{} {\bf return} net $\{s,x,b_{p-1},b_p,
      b_{q},b_{q+1}\}$.

    }\strut} $$\boxit{\box4}$$\vspace*{-6mm}
}
\caption{Subroutine B for the proof of Thm.~\ref{thm:1}.}
\label{fig:thm:1:sb}
\end{figure}

It is easy to verify that both subroutines correctly return in time
$O(||G||)$ a subgraph in $\cF_{LI}$.  Now we are ready to present the
main procedure (Fig.~\ref{fig:non-simplicial-shallow-terminal}).

  \begin{figure}[h!]  \centering \small
  \setbox4=\vbox{\hsize.9\linewidth \noindent\strut
    \hspace*{0.03\linewidth}\parbox[t]{0.87\linewidth}{ 
      \vspace*{1mm}
      1 \hspace*{2ex} {\bf for each} $x\in N(s)\setminus \{c_1,c_2\}$
      {\bf do} \comment{initially $C = \{c_1,c_2\}$.}
      \\
      \makebox[7ex][l]{1.1} {\bf if} $B^+\langle x\rangle = \emptyset$ 
      {\bf then goto} 1;
      \\
      \makebox[7ex][l]{1.2} {\bf if} $x\not\sim c_1$ or $c_2$ {\bf
        then call} subroutine A with $W$ and $x$;
      \\
      \makebox[7ex][l]{1.3} find the maximal sub-path ($b_i\cdots
      b_j$) of $B^+\langle x\rangle$ such that $i$ is minimum;
      \\
      \makebox[7ex][l]{1.4} {\bf if} $j-i\le 1$ or $j=2$ or $i=d-1$
      {\bf then} call subroutine~B with $W$ and $x$;
      \\
      \makebox[7ex][l]{1.5} \parbox[t]{.8\textwidth}{ {\bf if} $i\ge
        1$ {\bf then} $c_2 =x;\;p=i$; {\bf else} $p = 1$;
        \\
        {\bf if} $j\le d$ {\bf then} $c_1 = x;\;q=j$; {\bf else} $q =
        d$;}
      \\
      \makebox[7ex][l]{1.6} $C = C \cup \{x\}$; $l = b_{p-1}$; 
      $B = b_{p}\cdots b_q$; $r = b_{q+1}$;
      \\
      2 \hspace*{2ex} {\bf if} there are $x,y\in 
      C$ such that $x\not\sim y$ {\bf then return} $4$-hole $s x b_1 y s$;
    \\
    3 \hspace*{2ex} from $s$, apply BFS in $G - C$ to find the
    first vertex $x$ such that $x\sim B^+$ or $x\not\sim y\in C$
    \\
    \hspace*{4ex} $v={\sf prev}(x)$; {\bf if} $v\ne s$ {\bf
      then} $u={\sf prev}(v)$;
    \\
    4 \hspace*{2ex} {\bf if} $x\sim B^+$ {\bf then}
    \\
    \makebox[7ex][l]{4.1} find the maximal sub-path $b_i\cdots
      b_j$ of $B^+\langle x\rangle$ such that $i$ is minimum;
    \\
    \makebox[7ex][l]{4.2} {\bf if} $j-i\le 1$ or $j=2$ or $i=d-1$
      {\bf then} call subroutine~B with ($v:c_1,c_2:l,B,r$) and $x$;
    \\
    \makebox[7ex][l]{4.3} {\bf if} $i = 0$ and $j = d+1$ {\bf then}
    \\
    \hspace*{10ex} {\bf return} whipping top $\{s,c_1,l,b_1,x,b_d,r\}$
    or sun $\{s,c_1,l,x,r,c_2\}$; \comment{$\dag$ or $\ddag$ respectively.}
    \\[-1ex]
    \makebox[7ex][l]{4.4} {\bf call} subroutine A with $W'$ and $u$,
    where $W'=$\parbox{2in}{ \[ \begin{cases} (v: x,c_2: l, b_1\cdots
        b_j, b_{j+1}) & \text{if } i = 0,
        \\
        (v: c_1,x:b_{i-1}, b_i\cdots b_d, r) & \text{if } j = d+1,
        \\
        (v: x,x: b_{i-1}, b_i\cdots b_j, b_{j+1}) &
        \text{otherwise}; \end{cases} \] }
  \\[-1ex]
  5 \hspace*{2ex} {\bf else} \comment{$x\not\sim B^+$.}
  \\
  \makebox[7ex][l]{5.1} {\bf if} $x$ is nonadjacent to $c_1$ or $c_2$
 {\bf then} call subroutine~A with ($v:c_1,c_2: l,B,r$) and $x$;
  \\
  \makebox[7ex][l]{5.2} {\bf if} $y$ is a common neighbor of $B^+$
 {\bf then}
  \\
  \hspace*{10ex} {\bf return} whipping top $\{x,c_1,l,b_1,y,b_d,r\}$ or
  sun $\{x,c_1,l,y,r,c_2\}$; \comment{$\dag$ or $\ddag$ respectively.}
  \\[-1ex]
  \makebox[7ex][l]{5.3} \mbox{call} subroutine~A with $W'$ and $x$,
  where $W'=$\parbox{2in}{
    \[
    \begin{cases}
      (v: y,c_2: l, b_1\cdots b_j, b_{j+1}) & \text{if }l\sim y,\\
      (v: c_1,y: b_{i-1}, b_i\cdots b_d, r) & \text{if }r\sim y,\\
      (v: y,y: b_{i-1}, b_i\cdots b_j, b_{j+1}) & \text{otherwise}.
    \end{cases}
    \]
  }

}\strut} $$\boxit{\box4}$$ \vspace*{-6mm}
\caption{Main procedure for the proof of Thm.~\ref{thm:1}.}
\label{fig:non-simplicial-shallow-terminal}
  \end{figure}

  Let us verify the correctness of the main procedure.  Step~1
  searches for a special large caw in a local and greedy way.  Some of
  its iterations might update the caw $W$, and when it ends, the
  following conditions are satisfied by $W$:
  \begin{inparaenum}[(\itshape 1\upshape)]
  \item its shallow terminal is still $s$;
  \item its base is a subset of the base of the original caw given in
    the input; and
  \item if a vertex $x\in N(s)$ is adjacent to the base, then it is a
    common neighbor of it.
  \end{inparaenum}
  In the progress of this step, (1) and (2) will always be satisfied
  by the current caw, while (3) is satisfied by the set $C$ of
  vertices in $N(s)$ that have been explored: in the progress, $C$ is
  the set of common neighbors of $s$ and the current base $B$.  It is
  clear that this holds true initially, when only the center(s) of $W$
  are explored, both in $C$.  Each iteration of the for-loop explores
  a new vertex $x$ in $N(s)$.  A vertex nonadjacent to $B^+$ satisfies
  all three conditions vacuously; hence omitted (step 1.1).  If one of
  two aforementioned structures is found, the procedure calls either
  subroutine A (step 1.2) or subroutine B (step 1.4).  Otherwise it
  updates $W$ accordingly (steps 1.5 and 1.6).  It is easy to verify
  that $W$ is a valid caw.  Moreover, after each update of $W$, the
  new base is a subset of the previous one; hence (3) remains true for
  all explored vertices with respect to the new caw.  After step 1,
  every neighbor of $s$ is either in $C$, which are common neighbors
  of $\{s\} \cap B$, or nonadjacent to $B$.  Step~2 runs a check to
  ensure that $C$ induces a clique, which is straightforward.

  Step 3 then applies BFS to find vertex $x$ that is either adjacent
  to $B^+$ or nonadjacent to some vertex $y\in C$.  The existence of a
  vertex nonadjacent to some vertex in $C$ can be argued by
  contradiction.  Let $M$ be the component of $G - C$ that contains
  $s$.  Suppose, for contradiction, that $M$ is completely adjacent to
  $C$, then $M$ is a module of $G$.  Since $G$ is prime, we must have
  $M =\{s\}$, and then $s$ is simplicial in $G$, contradicting the
  assumption.  This ensures that the desired vertex can be found, but
  if a neighbor of $B^+$ is met first, then it will be $x$; it might
  be adjacent to all vertices in $C$.  Clearly, $x$ cannot be $s$;
  hence ${\sf prev}(x)$ is well-defined.  This verifies step 3.  Based
  on which condition $x$ satisfies, the procedure enters one of steps
  4 and 5.  Note that $x$ is not in $C$; hence if it is adjacent to
  $B^+$, then $v$ cannot be $s$.  In other words, $u$ is defined in
  step 4.  By assumption, $v$ is adjacent to every vertex in $C$ but
  nonadjacent to $B^+$, which means that ($v: c_1, c_2: l,B,r$) is a
  caw isomorphic to $W$.  Steps 4.1-4.3 and 5.1-5.3 are
  straightforward.  For step 4.4, note that $u\not\sim x$.

  We now analyze the running time of the main procedure.  Note that
  subroutines $A$ and $B$ can be called at most once, which terminate
  the procedure.  The dominating step in the for-loop of step 1 is
  finding the sub-path (step 1.3), which takes $O(d(x))$ time for each
  $x$ (Prop.~\ref{lem:find-indices}).  In total, step 1 takes
  $O(||G||)$ time.  The condition of step 3, i.e., whether $x$ is
  adjacent to $B^+$ and a common neighbor of $C$, can be checked in
  $O(d(x))$ time; hence step 3 can be done in $O(||G||)$ time.  Steps
  2, 4, and 5 are straightforward and all can be done in $O(||G||)$
  time.  This completes the time analysis and the proof.
\end{proof}

As a result, if a prime \lig\ $G$ is chordal, then $G - SI(G)$ is an
interval graph.

\subsection{Holes}\label{sec:holes}
For holes in a prime graph, we focus on their relation with other
vertices.  Let $H$ be a given hole of a prime graph $G$; we may assume
$|H|\ge 6$.  We start from characterizing \hv{v} for every $v\in G$:
we specify some forbidden structures not allowed to appear in a prime
\lig, and more importantly, we show how to find a subgraph of $G$ in
$\cF_{LI}$ if one of these structures exists.
\begin{lemma}\label{lem:non-consecutive}
  For every vertex $v$, we can in $O(d(v))$ time decide whether or
  not $\hv{v}$ induces a (possibly empty) sub-path of $H$, and if yes,
  in the same time find the ends of the path.  Otherwise, we can in
  $O(||G||)$ time find a subgraph of $G$ in $\cF_{LI}$.
\end{lemma}
\begin{proof}
  We call Prop.~\ref{lem:find-indices} to fetch the ordered list {\sf
    IND} of indices of \hv{v} in $H$.  If {\sf IND} is empty,
  i.e. $v\not\sim V(H)$, then we return a empty path with no vertex.
  Hereafter {\sf IND} is assumed to be nonempty.  We consider first
  the case where {\sf IND} contains at most $|H| - 1$ elements.  Let
  $p$ and $q$ be the first and last elements, respectively, of {\sf
    IND}.  Starting from the first element $p$, we traverse {\sf IND}
  to the end for the first $i$ such that ${\sf IND}[i+1] > {\sf
    IND}[i] + 1$.  If no such $i$ exists, then we return ($h_p\cdots
  h_q$) as the path $P$.  In the remaining cases, we may assume that
  we have found the $i$; let $p_1:={\sf IND}[i]$ and $p_2 := {\sf
    IND}[i+1]$.  We continue to traverse from $i+1$ to the end of {\sf
    IND} for the first $j$ such that ${\sf IND}[j+1] > {\sf IND}[j] +
  1$.  This step has three possible outcomes:
  \begin{inparaenum}[(1)]
  \item if $j$ is found, then $p_3:={\sf IND}[j]$ and $p_4:={\sf
      IND}[j+1]$;
  \item if no such $j$ is found, and at least one of $q<|H|-1$ and
    $p>0$ holds, then $p_3:=q$ and $p_4:= p + |H|$; and
  \item otherwise ($p=0$, $q=|H|-1$, and $j$ is not found).
  \end{inparaenum}
  In the third case, we return ($h_{p_2}\cdots h_{|H| - 1} h_0 \cdots
  h_{p_1}$) as the path induces by \hv{v}.  In the first two cases,
  $p_3$ and $p_4$ are defined, and $p_4 > p_3 +1$.  In other words, we
  have two nontrivial sub-paths, $h_{p_1} h_{p_1+1}\dots h_{p_2}$ and
  $h_{p_3} h_{p_3+1}\dots h_{p_4}$, of $H$ such that $v$ is adjacent
  to their ends but none of their inner vertices.  We then call
  subroutine A (Fig.~\ref{fig:non-consecutive:sa}), whose correctness
  is straightforward.

\begin{figure}[h!]
\setbox4=\vbox{\hsize28pc \noindent\strut
\begin{quote}
  \vspace*{-5mm} \small
  1 \hspace*{2ex} {\bf if} $p_2 < p_1 + 3$ {\bf then return} $4$- or
  $5$-hole $v h_{p_1} h_{p_1+1} \cdots h_{p_2} v$;
  \\
  2 \hspace*{2ex} {\bf if} $p_4 < p_3 + 3$ {\bf then return} $4$- or
  $5$-hole $v h_{p_3} h_{p_3 + 1} \cdots h_{p_4} v$;
  \\
  3 \hspace*{2ex} {\bf if} $p_3 = p_2$ {\bf then return} long claw
  $\{h_{p_1},v,h_{p_2 - 2}, h_{p_2 - 1}, h_{p_2}, h_{p_2 + 1}, h_{p_2
    + 2}\}$;
  \\
  4 \hspace*{2ex} {\bf if} $p_3 = p_2 + 1$ {\bf then return} net
  $\{h_{p_1},v,h_{p_2 - 1},h_{p_2}, h_{p_3}, h_{p_3 + 1}\}$;
  \\
  5 \hspace*{2ex} {\bf return} long claw $\{h_{p_1 + 1} ,h_{p_1}, v,
  h_{p_2-1}, h_{p_2 - 1}, h_{p_3}, h_{p_3 + 1}\}$. \comment{$p_3\ge
    p_2 + 2$.}
\end{quote} \vspace*{-6mm} \strut} $$\boxit{\box4}$$
\vspace*{-9mm}
\caption{Subroutine A for the proof of
  Lem.~\ref{lem:non-consecutive}.}
\label{fig:non-consecutive:sa}
\end{figure}

Assume now that {\sf IND} contains all $|H|$ elements
$\{0,1,\cdots,|H|-1\}$.  Let $C$ be the set of common neighbors of
$V(H)$, which is nonempty.  We find $C$ by using
Prop.~\ref{lem:find-indices} to check each vertex for its neighbors in
$H$.  Starting from $H$, we apply BFS in $G - C$ to find the first
vertex $u$ such that $u\not\sim x$ for some vertex $x\in C$.  The
existence of such a pair of vertices can be argued by contradiction.
Let $M$ be the component of $G - C$ that contains $H$.  Suppose, for
contradiction, that $M$ is completely adjacent to $C$, then $M$ is a
nontrivial module of $G$, which is impossible.

Let $u_0\cdots u_q$ be the searching path that leads from $u_0 \in H$
to $u_q = u$; that is, $u_i = {\sf prev}(u_{i+1})$ for $0\le i< q$.
Note that $u_1\not\in C$.  We find the path induced by $\hv{u_1}$,
which is nonempty and proper; otherwise we can use the previous case.
Based on the value of $q$, we proceed as follows
(Fig.~\ref{fig:non-consecutive:sb}).
  \begin{figure}[h!]
  \centering\small
  \setbox4=\vbox{\hsize.9\linewidth \noindent\strut
    \hspace*{0.03\linewidth}\parbox[t]{0.87\linewidth}{ %

  1 \hspace*{2ex} {\bf if} $q = 1$ {\bf then} 
  \\
  \hspace*{7ex} {\bf if} $\hv{u_1} = \{h_i\}$ {\bf then return} whipping top $\{x,
  h_{i-2}, h_{i-1}, h_{i}, h_{i+1}, h_{i+2}, u_1\}$;
  \\
  \hspace*{7ex} {\bf if} $\hv{u_1} = \{h_i, h_{i+1}\}$ {\bf then return}
  sun $\{x, h_{i-1}, h_{i}, h_{i+1}, h_{i+2}, u_1\}$;
  \\
  \hspace*{7ex} {\bf return} $4$-hole  $x h_{\tail{u_1}} u h_{\head{u_1}} x$;
  \\
  2 \hspace*{2ex} {\bf else} \comment{$q\ge 2$ and $u_1\sim x$.}
  \\
  \hspace*{7ex} {\bf if} $\hv{u_1} = \{h_i\}$ {\bf then return} long claw 
  $\{u_2, u_1, h_{i-2}, h_{i-1}, h_{i}, h_{i+1}, h_{i+2}\}$;
  \\
  \hspace*{7ex} {\bf if} $\hv{u_1} = \{h_i, h_{i+1}\}$ {\bf then return}
  net $\{u_2, u_1, h_{i-1}, h_{i}, h_{i+1}, h_{i+2}\}$;
  \\
  \hspace*{7ex} let $\hv{u_1} = \{h_i, h_{i+1}, \ldots, h_j\}$, {\bf return}
  whipping top $\{x, h_{i-1}, h_{i}, u_1, h_{j}, h_{j+1}, u_2\}$.
  
}\strut} $$\boxit{\box4}$$ \vspace*{-6mm}
\caption{Subroutine B for the proof of  Lem.~\ref{lem:non-consecutive}.}
\label{fig:non-consecutive:sb}
\end{figure}

The list {\sf IND} can be constructed in $O(d(v))$ time using
Prop.~\ref{lem:find-indices}.  We can traverse it and find its ends in
the same time if it induces a path.  We now consider the other
situations, and analyze the running time of finding subgraphs in
$\cF_{LI}$.  When $|\hv{v}|<|H|$ the detection of a subgraph in
$\cF_{LI}$ can also be done in $O(d(v))$ time: the main step is to
traverse {\sf IND} to obtain the indices $p_1,p_2,p_3,p_4$, which can
be done in $O(d(v))$ time, while the rest uses constant time.  The
dominating step of the last case is the construction of $C$, which
takes $O(||G||)$ time: the test of $|\hv{u}| = |H|$ for each vertex
(again using Prop.~\ref{lem:find-indices}) takes $O(d(u))$ time.  All
other steps use constant time.  This concludes the time analysis and
completes the proof.
\end{proof}

Now let $v$ be a vertex such that \hv{v} induces a path $P$.  We can
assign a direction to $P$ in accordance to the direction of $H$, and
then we have clockwise and counterclockwise ends of $P$.  For
technical reasons, we assign canonical indices to the ends of the path
$P$ as follows.

\begin{definition}
  For each vertex $v$ with nonempty $\hv{v}$, we denote by \tail{v}
  and \head{v} the indices of the counterclockwise and clockwise,
  respectively, ends of the path induced by \hv{v} in $H$ satisfying
  \begin{itemize}
  \item $- |H| < \tail{v} \le 0\le \head{v} < |H|$ {if } $h_0\in \hv{v}$; or
  \item $0< \tail{v}\le \head{v}< |H|$, {otherwise}.
  \end{itemize}
\end{definition}

It is possible that $\head{v} = \tail{v}$, when $|\hv{v}| = 1$.  In
general, $\head{v} - \tail{v} = |\hv{v}| - 1$ and $v=h_i$ or $v\sim
h_i$ for each $i$ with $\tail{v}\le i\le \head{v}$.  The indices
$\tail{v}$ and $\head{v}$ can be easily retrieved from
Lem.~\ref{lem:non-consecutive}, and with them we can check the
adjacency between $v$ and any vertex $h_i\in H$ in constant time, even
when $v\not\sim V(H)$.  (For example, with the definition of \tail{v}
and \head{v}, we may represent the fact $v\not\sim V(H)$ by
$\tail{v}>\head{v}$.)  One should be warned that ${\head{h_i}}$ may or
may not be ${i+1}$.

If $v$ is adjacent to $|H|-2$ or $|H|-1$ vertices in $H$, then it is
trivial to find a short hole in $V(H)\cup\{v\}$.  In the rest of this
paper, whenever we meet a hole and a vertex such that that $|\hv{v}|
\ge |H| - 2$ or $\hv{v}$ is nonconsecutive, we either return a short
hole or call Lem.~\ref{lem:non-consecutive}.  To avoid making the
paper unnecessarily ponderous, we will tacitly assume otherwise.  We
now turn to the vertices that are nonadjacent to $V(H)$. The following
lemma, together with Lem.~\ref{lem:non-consecutive} and the discussion
above, concludes the proof of Thm.~\ref{thm:2}.
\begin{lemma}\label{lem:c-star}
  Given a non-simplicial vertex $v$ that is nonadjacent to $H$, we can
  in $O(||G||)$ time find a subgraph of $G$ in $\cF_{LI}$.
\end{lemma}
\begin{proof}
  We may assume without loss of generality that some neighbor $u$ of
  $v$ is adjacent to $V(H)$: otherwise we can find (by BFS) a shortest
  path from $v$ to $H$ and take the last two inner vertices from this
  path as $v$ and $u$, respectively; in particular, as an inner vertex
  of a chordless path, the new vertex $v$ is necessarily
  non-simplicial.  We return long claw $\{v,u, h_{\tail{u}-2}, \ldots,
  h_{\tail{u}+2}\}$ if $|\hv{u}| = 1$; or net $\{v, u, h_{\tail{u}-1},
  h_{\tail{u}}, h_{\head{u}}, h_{\head{u}+1}\}$ if $|\hv{u}| = 2$.
  Otherwise, $|\hv{u}|\ge 3$, and we can call Thm.~\ref{thm:1} with
  large caw ($v: u,u: h_{\tail{u}-1}, h_{\tail{u}} \cdots
  h_{\head{u}}, h_{\head{u}+1}$).  Here we are using the assumption
  that $\hv{u}$ induces a path of at most $|H| - 3$ vertices and the
  fact that $v$ is not simplicial in $G$.  The dominating step is
  finding the appropriate vertices $v,u$, which takes $O(||G||)$ time.
\end{proof}

Now consider the neighbors of more than one vertices in $H$.  Here
$\hv{U} := \bigcup_{v\in U}\hv{v} = (\bigcup_{v\in U}N[v])\cap V(H)$.
\begin{lemma}\label{lem:non-consecutive-2}
  Given a set $U$ of vertices such that $G[U]$ is connected and
  $\hv{U}$ is not consecutive in $H$, we can in $O(||G||)$ time find a
  subgraph of $G$ in $\cF_{LI}$.
\end{lemma}
\begin{proof}
  For two disjoint sub-paths of $\hv{U}$, we can find a pair of
  vertices $u_1,u_2\in U$ such that they are adjacent to the paths
  respectively.  Traversing an induced \stpath{u_1}{u_2} in $G[U]$, we
  will obtain either a vertex nonadjacent to $H$, or a pair of
  adjacent vertices $u,v$ such that $\hv{u}$ and \hv{v} are both
  nonempty and disjoint.  In the first case we call
  Lem.~\ref{lem:c-star}.  We consider then the second case.  Clearly,
  neither of $u$ and $v$ can be in $H$.  We may renumber vertices of
  $H$ such that $\tail{u} = 0$, and then $\head{u} < \tail{v} <
  \head{v} < |H|$.  The neighborhood of $v$ in the hole $h_0 u
  h_{\head{u}} \cdots h_{|H| -1} h_0$ is non-consecutive, and thus we
  can call Lem.~\ref{lem:non-consecutive}.
\end{proof}

In particular, for any pair of adjacent vertices $u,v$, if neither
\hv{u} nor \hv{v} is a subset of the other, then at least one of
$h_{\head{v}}$ and $h_{\tail{v}}$ needs to be in $\hv{u}$.  The
following lemma asserts that they cannot be both.
\begin{lemma}\label{lem:non-helly}
  Given a set $U$ of two or there pairwise adjacent vertices such that
  $\hv{U} = V(H)$, we can in $O(||G||)$ time find a subgraph of $G$ in
  $\cF_{LI}$.
\end{lemma}
  \begin{figure}[h!]
  {\centering\small
  \setbox4=\vbox{\hsize.9\linewidth \noindent\strut
    \hspace*{0.03\linewidth}\parbox[t]{0.87\linewidth}{ %

      1 \hspace*{2ex} {\bf if} $u_2\not\sim h_{\head{u_1}}$ {\bf then
        return} a $4$-hole;
      \\
      2 \hspace*{2ex} {\bf if} $u_2\sim h_{\tail{u_1}}$ {\bf then}
      \comment{must be the case when $|U| = 2$.}
      \\
      \hspace*{7ex} {\bf call} Lem.~\ref{lem:non-consecutive} with
      $u_1 h_{\head{u_1}} h_{\head{u_1} + 1} h_{\tail{u_1}}
      u_1$ and $u_2$;
      \\
      3 \hspace*{2ex} {\bf if} $u_2\sim \{h_{\tail{u_1}-1},
      h_{\tail{u_1}-2}\}$ {\bf then return} a short hole;
      \\
      4 \hspace*{2ex} {\bf if} $u_3\not\sim h_{\tail{u_1}}$or
      $u_3\not\sim h_{\head{u_2}}\}$ {\bf then return} a short hole;
      \\
      5 \hspace*{2ex} {\bf call} Lem.~\ref{lem:non-consecutive} with
      $u_1 u_2 h_{\head{u_2}} h_{\head{u_2} + 1} h_{\tail{u_1}}$ and 
      $u_3$.

    }\strut} $$\boxit{\box4}$$\vspace*{-6mm}
}
\caption{Procedure for the proof of Lem.~\ref{lem:non-helly}.}
\label{fig:lem:non-helly}
\end{figure}
\begin{proof}
  We start from an arbitrary vertex $u_1$ of $U$.  Without loss of
  generality, we may assume $u_1\sim H$, and in particular,
  Lem.~\ref{lem:non-consecutive} returns a proper sub-path; otherwise
  we are done.  There must be another vertex in $U$ that is adjacent
  to $h_{\head{u_1}+1}$; let it be $u_2$.  We proceed as
  Fig.~\ref{fig:lem:non-helly}.  The correctness and running time of
  this procedure are straightforward.
\end{proof}

 \section{Characterization and decomposition of prime \lig s}
\label{sec:olive-ring}
This section presents the details for algorithm decompose
(Fig.~\ref{fig:recognize-lig}).  If $G$ is chordal, then it suffices
to check whether $G - SI(G)$ is an interval graph or not: if yes, a
caterpillar decomposition for $G$ can be easily built; otherwise, we
can use Thm.~\ref{thm:1} to find a subgraph in $\cF_{LI}$.  Therefore,
in this section we are mainly concerned with non-chordal graphs, where
a hole can be found in linear time
\cite{tarjan-85-certifyig-chordal-recognition}.  Section~\ref{sec:mho}
gives the details on the construction of the auxiliary graph
$\mho(G)$, which is used in Section~\ref{sec:negative} to prove
Thm.~\ref{thm:3}, and in Section~\ref{sec:olive-ring-decomposition} to
build the olive-ring decomposition for $G$.

With the customary abuse of notation, the same symbol $K$ is used for
a maximal clique of $G$ and its corresponding bag in a clique
decomposition ${\cal K}$ for $G$.  A complete graph on all maximal
cliques of a graph gives a trivial {clique decomposition} for the
graph, which is uninteresting.  We are only interested in clique
decompositions that can be stored and manipulated in linear time.
Every clique decomposition $\cal K$ in this paper will satisfy
({\itshape 1}) $||{\cal K}||\le |{\cal K}|\le |G|$; and ({\itshape 2})
each vertex $v\in G$ appears in at most $d(v)$ bags.  Since $\cal K$
is connected, it either is a tree or has a unique cycle.

For example, a chordal graph $G$ has at most $|G|$ maximal cliques
\cite{dirac-61-chordal-graphs}, which can be arranged as a tree such
that for every $v\in G$, the set of maximal cliques containing $v$
induces a subtree \cite{buneman-1974-rigid-circuit-graphs}.  Interval
graphs are chordal, and thus admit clique tree decompositions as well.
Fulkerson and Gross \cite{fulkerson-65-interval-graphs} showed that an
interval graph always has a clique path decomposition.  Also of
interest in this paper are clique decompositions that are holes,
caterpillars, and olive rings, which are called, for the sake of
brevity, hole decompositions, caterpillar decompositions, and
olive-ring decompositions, respectively.  A path of at least four bags
can be made an hole by adding an edge connecting its end bags: adding
an extra edge to the path does not break any condition in the
definition of clique decomposition.  For the same reason, a
caterpillar decomposition whose central path has at least four bags
can be viewed as an olive-ring decomposition.

We point out that clique tree decompositions for chordal graphs have
different formulations, all of which, as shown by Blair and Peyton
\cite{blair-91-chordal-graphs-clique-trees}, are equivalent.  The
definition we use here, not relying on the fact that any pair of bags
is connected by a unique path, is easier to be generalized.  In this
paper, no vertex will be allowed to occupy every bag of the cycle of
$\cal K$ (when it has a cycle), and thus a graph $G$ that has an
olive-ring decomposition can also be viewed as the intersection graph
of subtrees of an olive ring.

For any simplicial vertex $v$, the clique induced by $N[v]$ must be
maximal.  This observation allows us to find all simplicial vertices
of a graph by traversing its clique decomposition: we count the
occurrences of all vertices, and return those vertices with number
$1$.  This approach runs in time $\sum_{K\in \cal K} |K| = O(||G||)$,
and works for all aforementioned classes.  In particular, a clique
tree for a chordal graph can be built in linear time.
\begin{lemma}\label{lem:find-simplicial-vertices}
  We can in $O(||G||)$ time find all simplicial vertices of a chordal
  graph $G$.
\end{lemma}

\subsection{The auxiliary graph $\mho(G)$}
\label{sec:mho}
If a vertex $v$ is adjacent to four or more consecutive vertices in a
hole $H$, i.e., $\head{v}-\tail{v} > 2$, then $v\not\in H$.  We can
use $h_{\tail{v}} v h_{\head{v}}$ as a short cut for the sub-path
induced by the neighbors of $v$ in $H$, thereby yielding a strictly
shorter hole.  To simplify the later presentation, we would like that
$h_0$ cannot be bypassed as such.  The following lemma formally states
this condition and gives a procedure for finding a hole satisfying it.

\begin{lemma}\label{lem:hole-conditions}
  In $O(||G||)$ time, we can find either a subgraph of $G$ in
  $\cF_{LI}$, or a hole $H$ such that $N[v]\subseteq N[h_0]$ holds for
  every vertex $v\in N(h_{-1})\cap N(h_{1})$.
\end{lemma}
\begin{figure}[h]
\setbox4=\vbox{\hsize28pc \noindent\strut
\begin{quote}
  \vspace*{-5mm} \small

  $[\star]$ \emph{If a subgraph in $\cF_{LI}$ is found in any step
    {then return} it.}
  \\[1ex]
  0 \hspace*{2ex} $a = -1$; $b = 1$; $C=\{h_0\}$;
  \\
  1 \hspace*{2ex} {\bf for each} $v\in V(G)\setminus V(H)$ {\bf do}
  \\
  1.1 \hspace*{4ex} compute \tail{v} and \head{v} in $H$;
  \\
  1.2 \hspace*{4ex} {\bf if} \big($\tail{v}< a$ and $\head{v} \ge
  b$\big) or \big($\tail{v}= a$ and $\head{v} >b$\big) {\bf then}
  \\
  \hspace*{10ex} $a = \tail{v}$; $b = \head{v}$; $C =\{v\}$;
  \comment{$\tail{v}< 0 < \head{v}$.}
  \\
  1.3 \hspace*{4ex} {\bf else if} $\tail{v}= a$ and $\head{v} = b$
  {\bf then} $C = C\cup \{v\}$;
  \\
  2 \hspace*{2ex} {\bf if} $C$ has a pair of nonadjacent vertices
  $v_1,v_2$ {\bf then return} $4$-hole $v_1 h_a v_2 h_b v_1$;
  \\
  3 \hspace*{2ex} let $h$ be the vertex in $C$ that has the maximum
  degree;
  \\
  4 \hspace*{2ex} {\bf if} there exists $v$ such that $\tail{v} = b$
  and $\head{v} = |H| + a$ {\bf then}
  \\
  4.1 \hspace*{4ex} {\bf if} $v\not\sim h$ {\bf then return} $4$-hole
  $v h_a h h_b v$;
  \\
  4.2 \hspace*{4ex} {\bf else call} Lem.~\ref{lem:non-helly} with
  $\{v,h\}$;
  \\
  5 \hspace*{2ex} {\bf if} $a + |H| - b\le 3$ {\bf then return} $h h_b
  h_{b+1} \cdots h_a h$ as a short hole;
  \\
  6 \hspace*{2ex} {\bf for each} $v\in C\setminus \{h\}$ {\bf do}
  \\
  6.1 \hspace*{4ex} {\bf if} there exists $x\in N[v]\setminus N[h]$
  {\bf then}
  \\
  \hspace*{10ex} find $y\in N[h]\setminus N[v]$; {\bf return} net
  $\{h_{a-1}, h_a, h, v, x, y\}$;
  \\
  7 \hspace*{2ex} {\bf return} $h h_{b} h_{b + 1} \cdots h_{a} h$
  where $h$ is the new $h_0$.  
\end{quote} \vspace*{-3mm} \strut} $$\boxit{\box4}$$
\vspace*{-7mm}
\caption{Procedure for finding the hole for Lem.~\ref{lem:hole-conditions}.}
\label{fig:compress-hole}
\end{figure}
\begin{proof}
  We apply the procedure given in Fig.~\ref{fig:compress-hole}.  As
  noted by remark [$\star$], several steps of the procedure might end with a
  subgraph in $\cF_{LI}$, and in this case, we terminate it by
  returning this subgraph.  The set $C$ stores all explored vertices
  $v$ satisfying $\tail{v} = a$ and $\head{v} = b$: initially, $a=-1$
  and $b=1$, and vertices in $H$ are considered explored, hence $C =
  \{h_0\}$.  Step~1 greedily extends $[a,b]$ in either or both
  directions such that $\{h_a, h_{a+1},\dots, h_b\}$ is the maximal
  neighborhood in $H$ among all explored vertices.  Note that $a<0<b$
  always holds as index $a$ is non-increasing while index $b$ is
  nondecreasing.  Each iteration of step~1 checks a vertex $v$ that
  has not been explored.  If either condition of step 1.2 is
  satisfied, then $N[v]\cap V(H)$ properly contains
  $\{h_a,h_{a+1},\dots,h_b\}$, and $a$ and $b$ are updated to be
  \tail{v} and \head{v} respectively.  After this update, no vertex in
  $C$ is adjacent to all of $\{h_a,\dots,h_b\}$, and hence they are
  purged from $C$.  No previously explored vertex can be adjacent to
  all of $\{h_a,h_{a+1},\dots,h_b\}$ either.  Therefore, now $C$
  consists of only $v$.  Step 1.3 puts into $C$ those vertices whose
  closed neighborhood in $H$ is precisely $\{h_a,h_{a+1},\dots,h_b\}$.

  The correctness of steps 2-4 is clear; after they have been passed,
  $C$ must induce a clique, and $h h_{b} h_{b + 1} \cdots h_{a} h$ is
  a hole.  If one or both of $a$ and $b$ have been updated in step 1,
  then this hole is strictly shorter than $H$; we only proceed when
  its length is at least $6$ (step~5).  For the correctness of step 6,
  notice that the existence of $y$ is ensured by the selection of $h$
  (having the maximum degree in $C$ and hence no less than $d(v)$) and
  the condition on $x$ (adjacent to $v$ but not $h$).  Step 7 returns
  the new hole, and number it in a way that the vertices $h$ and $h_b$
  are the new $h_0$ and $h_1$ respectively.  Implicitly from step 1.1,
  the neighborhood of every vertex in $H$ is consecutive.  Moreover,
  as step~4 has been passed, $v\in N(h_a)\cap N(h_b)$ if and only if
  $v\in C$; such a vertex satisfying $N[v]\subseteq N[h_0]$.

  Let us now analyze the running time.  What dominates step 1 is finding
  \tail{v} and \head{v} for all vertices (step~1.1), which takes
  $O(d(v))$ time for each vertex $v$ and $O(||G||)$ time in total.
  Steps 2 and 3 take $O(||G||)$ time.  If Lem.~\ref{lem:non-helly} is
  called in step 4.2, then it takes $O(||G||)$ time and the procedure
  is finished; otherwise it passes step 4 in $O(|G|)$ time.  Step 5
  takes constant time.  Step 6 takes $O(d(v))$ time for each vertex
  $v\in C$ and $O(||G||)$ time in total.  Therefore, the whole
  procedure can be implemented in $O(||G||)$ time.
\end{proof}

Recall that $\oo := N[h_0]$ and $\overline{\oo} := V(G)\setminus \oo$.
By Lem.~\ref{lem:hole-conditions}, every common neighbor of $h_1$ and
$h_{-1}$ is in $\oo$, and is adjacent to neither $h_{-2}$ nor $h_2$.
Let $v\in \oo$ and $u\in \overline \oo$ be a pair of adjacent
vertices; note that $u$ is adjacent to at most one of $h_1$ and
$h_{-1}$.  We may assume that neither of
Lems.~\ref{lem:non-consecutive} and \ref{lem:non-consecutive-2}
applies to $v$ and $u$, as otherwise we have already found a subgraph
of $G$ in $\cF_{LI}$.  In particular, we can exclude the possibility
that $v$ is adjacent to neither $h_{-1}$ nor $h_1$:
Lem.~\ref{lem:non-consecutive-2} applies when $u\sim V(H)$ , while we
can return long claw $\{v,u,h_{-2},h_{-1},h_{0},h_{1},h_{2}\}$
otherwise.  We are therefore left with three cases based on the
adjacency between $v$ and $h_{-1},h_1$, and the direction of edge $u
v$ is accordingly determined as follows.
\begin{itemize}
\item If $\tail{v} < \head{v} = 0$ (i.e., $N[v]$ contains $h_{-1}$ but
  not $h_1$), then $v u$ is \emph{counterclockwise} from $\oo$;
\item if $\tail{v} = 0 < \head{v}$ (i.e., $N[v]$ contains $h_{1}$ but
  not $h_{-1}$), then $v u$ is \emph{clockwise} from $\oo$; or
\item otherwise, $\tail{v} = -1, \head{v} = 1$ (i.e., $v$ is adjacent
  to both $h_{-1}$ and $h_1$), then $u$ is adjacent to either $h_{-1}$
  or $h_1$ (Lems.~\ref{lem:hole-conditions} and
  \ref{lem:non-consecutive-2}), and the edge $v u$ is
  \emph{counterclockwise} or \emph{clockwise} from $\oo$ respectively.
\end{itemize}
Put it simply, the direction of $v u$ is decided by $v$ when ``$v$
lies at one side of $h_0$,'' and by $u$ otherwise.  In particular,
$h_{-1} h_{-2}\in \ecc$ and $h_1 h_2\in\ec$.  It is easy to verify
that $h_{-1}\in N[v]$ (resp., $h_{1}\in N[v]$) always holds true when
$v u\in \ecc$ (resp., $v u\in \ec$).  We now have all the details for
the definition of the auxiliary graph $\mho(G)$, which is restated
below.
\paragraph{Definition~\ref{def:mho} (restated).}
  The vertex set of $\mho(G)$ consists of $\overline{\oo}\cup L\cup
  R\cup \{w\}$, where $L$ and $R$ are distinct copies of $\oo$, i.e.,
  for each $v\in \oo$, there are a vertex $v^l$ in $L$ and another
  vertex $v^r$ in $R$, and $w$ is a new vertex distinct from $V(G)$.
  For each edge $u v\in E(G)$, we add to the edge set of $\mho(G)$
\begin{itemize}
\item an edge $u v$ if neither $u$ nor $v$ is in $\oo$;
\item two edges $u^l v^l$ and $u^r v^r$ if both $u$ and $v$ are in
  $\oo$; or
\item an edge $u v^l$ or $u v^r$ if $uv\in \ec$ or $uv\in \ecc$
  respectively ($v\in \oo$ and $u\in \overline\oo$).
\end{itemize}
Finally, we add an edge $w v^l$ for every $\{v\in \oo: u v \in \ecc\}$.

\medskip

By definition, $L$ and $R$ are nonadjacent, and for any $v\in \oo$, the
two vertices $v^l$ and $v^r$ derived from $v$ have no common
neighbors.  The following lemma takes care of other situations when 
the distance between $L$ and $R$ is small.
\begin{lemma}\label{lem:one-side}
  We can in $O(||G||)$ time find a subgraph of $G$ in $\cF_{LI}$ if
  given
  \begin{enumerate}[(1)]
  \item a vertex $u\in\overline{\oo}$ as well as edges $u v_1\in \ecc$
    and $u v_2\in \ec$; or
  \item a pair of adjacent vertices $u_1,u_2\in \overline \oo$ as well
    as edges $u_1 v_1\in \ecc$ and $u_2 v_2\in \ec$.
  \end{enumerate}
\end{lemma}
\begin{proof}
  (1) Consider first that $u\sim V(H)$.  Then $u v_1\in \ecc$ and $u
  v_2\in \ec$ imply that either $u$ is adjacent to all of
  $\{h_{\head{v_2}},\dots,h_{\tail{v_1}}\}$, then we can apply
  Lem.~\ref{lem:non-helly}; or the neighbors of $u$ in $H$ is not
  consecutive, then we can apply Lem.~\ref{lem:non-consecutive}.
  Assume now $u\not\sim V(H)$, then according to
  Lem.~\ref{lem:hole-conditions}, $v_1\not\sim h_1$ and $v_2\not\sim
  h_{-1}$.  We can return net $\{u,v_1,v_2,h_{-1},h_0,h_1\}$ if
  $v_1\sim v_2$ or $4$-hole $u v_1 h_0 v_2 u$ otherwise.

  (2) Consider first that $v_1 = v_2$.  Then by the definition of
  \ecc\ and \ec, $v_1$ is adjacent to both $h_{-1}$ and $h_1$, and it
  follows that $u_1\sim h_{-1}$ and $u_2\sim h_1$.  Observing that
  $u_1\sim u_2$ but neither of $u_1$ and $u_2$ is adjacent to $h_0$,
  at least one of Lems.~\ref{lem:non-consecutive} and
  \ref{lem:non-helly} applies.  Assume now that $v_1\ne v_2$, and
  without loss of generality, $u_1\not\sim v_2$ and $u_2\not\sim v_1$
  (otherwise, we are already in case (1)).  Then we can return $u_1
  v_1 v_2 u_2 u_1$ or $u_1 v_1 h_0 v_2 u_2 u_1$ as a short hole.
\end{proof}

Let \occ\ (resp., \oc) denote the subset of vertices of $\oo$ that are
incident to edges in \ecc\ (resp., \ec).  We have mentioned that
$\{\ecc,\ec\}$ partitions edges between $\oo$ and $\overline \oo$, but
a vertex in $\oo$ might belong to both \occ\ and \oc, or neither of
them.  Clearly, $\occ\subset N[h_{-1}]$ and $\oc\subset N[h_1]$, on
which we have the following result.
\begin{lemma}\label{lem:O}
  Given a pair of nonadjacent vertices $u,x\in \occ$ (or \oc), we can
  in $O(||G||)$ time find a subgraph of $G$ in $\cF_{LI}$.
\end{lemma}
\begin{proof}
  In this proof we consider the set \occ, and a symmetrical argument
  applies to \oc.  By definition, we can find edges $u v,x y\in \ecc$,
  where $v,y\in \overline\oo$.  We have three (possibly intersecting)
  chordless paths $h_0 h_1 h_2$, $h_0 u v$, and $h_0 x y$.  If both
  $u$ and $x$ are adjacent to $h_1$, then we can return $4$-hole $u
  h_{-1} x h_1 u$.  Hence we may assume, without loss of generality,
  $x\not\sim h_1$.

  Assume first that $u\sim h_1$.  We consider the subgraph induced by
  $X_1 := \{h_0, h_1, h_2, u, v, x\}$, which is clearly distinct.
  Here the only uncertain adjacencies are between $v, x$, and $h_2$:
  by assumption, $h_0, h_1$, and $u$ are pairwise adjacent; $x$ is
  adjacent to neither $u$ nor $h_1$; $h_2$ is adjacent to neither
  $h_0$ nor $u$ (because Lem.~\ref{lem:hole-conditions}: by
  assumption, $N[u]\cap V(H) = \{h_{-1}, h_0, h_1\}$ and $v\sim
  h_{-1}$); and $v$ is adjacent to neither $h_0$ nor $h_1$.  If $v,
  x$, and $h_2$ are pairwise nonadjacent, then we return $G[X_1]$ as a
  net.  Otherwise, there is at least one edge among $v, x$, and $h_2$,
  then we return a $4$-hole, e.g., $v x h_0 u v$ when the edge is
  $vx$.

  In the remaining cases, $u$, $x$, and $h_1$ are pairwise
  nonadjacent.  If any two of $\{v, y, h_2\}$ are identical or
  adjacent, then we return a $4$- or $5$-hole, e.g., $h_0 u v x h_0$
  or $h_0 u v y x h_0$ when $v=y$ or $v\sim y$ respectively.
  Otherwise, $v, y$, and $h_2$ are distinct and pairwise nonadjacent,
  and we return long claw $\{h_0, h_1, h_2, u, v, x, y\}$.

  Edges $u v$ and $x y$ can be found in $O(||G||)$ time, and only a
  small constant number of adjacencies are checked in this procedure;
  it thus takes $O(||G||)$ time in total.  
\end{proof}

\begin{lemma}\label{lem:construct-mho}
  The order and size of $\mho(G)$ are upper bounded by $2|G|$ and
  $2||G||$ respectively.  Moreover, an adjacency list representation
  of $\mho(G)$ can be constructed in $O(||G||)$ time.
\end{lemma}
\begin{figure}[h!]
  \vspace*{-5mm}
  \setbox4=\vbox{\hsize28pc \noindent\strut
  \begin{quote}
  \vspace*{-5mm} \small

  {\sc input}: a prime graph $G$ and a hole $H$ satisfying conditions
  of Lem.~\ref{lem:hole-conditions}.
  \\
  {\sc output}: the auxiliary graph $\mho(G)$ or a subgraph of $G$ in
  $\cF_{LI}$.$^{[\star]}$
  \\[1ex]
  0 \hspace*{2ex} {\bf for each} $v\in V(G)$ {\bf do} compute \tail{v}
  and \head{v};
  \\
  1 \hspace*{2ex} {\bf for each} $v\in \oo$ {\bf do}
  \\
  1.1 \hspace*{3ex} add vertices $v^l$ and $v^r$;
  \\
  1.2 \hspace*{3ex} {\bf for each} $u\in N(v)$ {\bf do}
  \\
  1.2.1 \hspace*{5ex} {\bf if} $u\in \oo$ {\bf then} add $u^l$ to
  $N(v^l)$ and $u^r$ to $N(v^r)$;
  \\
  1.2.2 \hspace*{5ex} {\bf else if} $u$ is not marked {\bf then} mark
  $u$ and put it into $N(\oo)$;
  \\
  2 \hspace*{2ex} {\bf for each} $u\in N(\oo)$ {\bf do}
  \\
  2.1 \hspace*{3ex} {\bf for each} $v\in N(u)$ {\bf do}
  \comment{$\tail{v}\le 0\le \head{v}$.}
  \\
  2.1.1 \hspace*{5ex} {\bf if} $v\not\in \oo$ {\bf then goto} 2.1;
  \comment{henceforth $v\in X$.}
  \\
  2.1.2 \hspace*{5ex} {\bf if} $\tail{v}= 0= \head{v}$ {\bf then
    return} a subgraph in $\cF_{LI}$;
  \\
  2.1.3 \hspace*{5ex} {\bf if} $\head{v} = 0$ {\bf then} replace $v$
  by $v^l$ in $N(u)$ and add $u$ to $N(v^l)$; \comment{$u v\in \ecc$.}
  \\
  2.1.4 \hspace*{5ex} {\bf if} $\tail{v} = 0$ {\bf then} replace $v$
  by $v^r$ in $N(u)$ and add $u$ to $N(v^r)$; \comment{$u v\in \ec$.}
  \\
  \hspace*{5ex} $\setminus\!\!\setminus$ in the remaining cases
  $N[v]\cap V(H) = \{h_{-1}, h_0, h_1\}$ and $|N[u] \cap \{h_{-1},
  h_1\}| = 1$.
  \\
  2.1.5 \hspace*{5ex} {\bf if} $\head{u} = |H| - 1$ {\bf then} replace
  $v$ by $v^l$ in $N(u)$ and add $u$ to $N(v^l)$; \comment{$u v\in
    \ecc$.}
  \\
  2.1.6 \hspace*{5ex} {\bf if} $\tail{u} = 1$ {\bf then} replace $v$
  by $v^r$ in $N(u)$ and add $u$ to $N(v^r)$; \comment{$u v\in \ec$.}
  \\
  2.1.7 \hspace*{5ex} {\bf if} $u v\in \ecc$ and $v$ is not marked as
  \occ\ {\bf then} mark $v$ and put it into \occ;
  \\
  2.1.8 \hspace*{5ex} {\bf if} $u v\in \ec$ and $v$ is not marked as
  \oc\ {\bf then} mark $v$ and put it into \oc;
  \\
  2.2 \hspace*{3ex} {\bf if} $u$ is incident to edges in both \ecc\
  and \ec\ {\bf then call} Lem.~\ref{lem:one-side}(1);
  \\
  2.3 \hspace*{3ex} {\bf for each} $u'\in N(u)\setminus \oo$ {\bf do}
  \\
  2.3.1 \hspace*{5ex} {\bf if} $u$ and $u'$ are incident to edges in
  \ecc\ and \ec\, respectively, {\bf then call}
  Lem.~\ref{lem:one-side}(2);
  \\
  3 \hspace*{2ex} {\bf if} \oc\ or \occ\ does not induce a clique {\bf
    then call} Lem.~\ref{lem:O};
  \\
  4 \hspace*{2ex} add vertex $w$;
  \\
  5 \hspace*{2ex} {\bf for each} $v\in \occ$ {\bf do} put $w$ into
  $N(v^l)$ and $v^l$ into $N(w)$;
  \\
  6 \hspace*{2ex} remove $\oo$.
  \\[2ex]
  $[\star]$ \hspace*{2ex} {\bf if} a subgraph in $\cF_{LI}$ is found in any
  step {\bf then return} it.
\end{quote} \vspace*{-3mm} \strut} $$\boxit{\box4}$$
\vspace*{-7mm}
\caption{Procedure for constructing $\mho(G)$
  (Lem.~\ref{lem:construct-mho}).}
\label{fig:construct-omega-G}
\end{figure}
\begin{proof}
  The vertices of the auxiliary graph $\mho(G)$ include $\overline \oo$,
  two copies of $\oo$, and $w$, i.e., $|\mho(G)| = 2|\oo| + |\overline{\oo}|
  + 1 = |G| + |\oo| + 1 \le 2|G|$.  In $\mho(G)$, there are two edges
  derived from every edge of $G[\oo]$ and one edge from every other edge
  of $G$.  All other edges are incident to $w$, and there are \occ\ of
  them.  Therefore, $\|\mho(G)\| = ||G|| + ||G[\oo]|| + |\occ| \le
  ||G|| + ||G[\oo]|| + |\ecc| < 2||G||$.  This concludes the first
  assertion.

  For the construction of $\mho(G)$, we use the procedure described in
  Fig.~\ref{fig:construct-omega-G} (some peripheral bookkeeping
  details are omitted for the sake of clarity).  Step~1 adds vertex
  sets $L$ and $R$ (step 1.1) as well as those edges induced by them
  (step 1.2.1), and finds $N(\oo)$ (step 1.2.2).  Step~2 adds edges in
  $\ecc$ and $\ec$, and detect \occ\ and \oc.  Steps~2.2, 2.3, and 3
  verify that neither of Lems.~\ref{lem:one-side} and \ref{lem:O}
  applies; information required in these verifications can be obtained
  in step 2.1 and stored.  Steps~4 and 5 add vertex $w$ and edges
  incident to it.  Step~6 cleans $\oo$.  The main steps are 1 and 2,
  each of which checks every edge at most once, and hence the total
  time is $O(||G||)$.
\end{proof}

By steps 2.2 and 2.3 (Lem.~\ref{lem:one-side} as well as the
discussion preceding it), and step 3 (Lem.~\ref{lem:O}), a posteriori,
the following properties hold for $\mho(G)$.
\begin{proposition}\label{lem:property-mho}
  In the auxiliary graph $\mho(G)$, any path between $L$ and $R$ has
  length at least $4$, and the vertex $w$ is simplicial.
\end{proposition}

\subsection{Proof of Thm.~\ref{thm:3}: Detection of subgraphs in
  $\cF_{LI}$}
\label{sec:negative}
Each vertex $x$ of $\mho(G)$ different from $w$ is uniquely defined by
a vertex of $G$, which is denoted by $\og{x}$.  We say that $x$ is
\emph{derived from} $\og{x}$.  By the definition of $\mho(G)$, if a
pair of vertices $x$ and $y$ (different from $w$) is adjacent in
$\mho(G)$, then their original vertices in $G$ must be adjacent as
well.  However, the converse may not hold true, e.g., $h_0\sim h_1$ in
$G$ but $h^l_0\not\sim h^r_1$ in $\mho(G)$.  A pair of vertices $x,y$
of $\mho(G)$ is a \emph{broken pair} if $\og{x}\sim \og{y}$ in $G$ but
${x}\not\sim {y}$ in $\mho(G)$.

For example, $\og{v^l} = \og{v^r} = v$ for $v\in \oo$.  By abuse of
notation, we will use the same symbol for a vertex $u\in
\overline{\oo}$ of $G$ and the unique vertex of $\mho(G)$ derived from
$u$; its meaning is always clear from the context.  In other words,
$\og{u} = u$ for $u\in \overline \oo$, and in particular, $\og{h_i} =
h_i$ for $i = 2, \dots, |H|-2$.  Any pair of adjacent vertices $u,v\in
\oo$ gives two broken pairs in $\mho(G)$, namely $\{u^l, v^r\}$ and
$\{u^r, v^l\}$, and any edge $u v \in \ec$ (resp., $u v \in \ecc$),
where $v\in\oc$ (resp., $v\in\occ$), gives a broken pair $\{u, v^r\}$
(resp., $\{u, v^l\}$).  Every broken pair is in one of these cases,
and thus at least one vertex of it is in $L\cup R$.  We can mark
$\og{x}$ for each vertex during the construction of $\mho(G)$, and it
can be generalized to a set $U$ of vertices that does not contain $w$,
i.e., $\og{U} := \{\og{v}: v\in U\}$; we point out that possibly
$|\og{U}|< |U|$.
\begin{proposition}\label{lem:isomorphism}
  Let $X$ be a set of vertices of $\mho(G)$ that does not contain $w$
  or both $\{v^l,v^r\}$ for any $v\in \oo$.  Then $\mho(G)[X]$ is a
  subgraph of $G[\og{X}]$, and they are isomorphic if and only if $X$
  contains no broken pairs.
\end{proposition}
\begin{proof}
  By assumption, there is a one-to-one mapping between $X$ and \og{X}.
  If $X$ is free of broken pairs, then this mapping also gives an
  isomorphism between $\mho(G)[X]$ and $G[\og{X}]$.  On the other
  hand, if $X$ contains broken pairs, then $\mho(G)[X]$ has strictly
  less edges than $G[\og{X}]$, and thus they cannot be isomorphic.
\end{proof}

This observation enables us to prove the following lemma, which is
crucial for the identification of simplicial vertices of $G$.  Here we
use \lcc\ and \lc\ to denote the subset of vertices of $L$ derived
from \occ\ and \oc, respectively, i.e., $\lcc := \{v^l: v\in \occ\}$
and $\lc := \{v^l: v\in \oc\}$.
\begin{lemma}\label{lem:simplicial-vertices}
  A vertex $x$ different from $\{w\}\cup R$ is simplicial in $\mho(G)$
  if and only if $\og{x}$ is simplicial in $G$.
\end{lemma}
\begin{proof}
  Every vertex in $\lcc$ is adjacent to both $h^l_0$ and $w$, and thus
  cannot be simplicial in $\mho(G)$.  Likewise, a vertex in $\occ$ is
  adjacent to $h_0$ and $\overline \oo$, and thus cannot be simplicial
  in $G$.  Therefore, we may assume $x\not\in \lcc$; hence $x\not\sim
  w$ and $\og{x}\not\in \occ$.  For such a vertex $x$, any edge of $G$
  incident to \og{x} has a corresponding edge of $\mho(G)$ incident to
  $x$.  In other words, there is a one-to-one mapping between the
  neighbors of $x$ in $\mho(G)$ and the neighbors of \og{x} in $G$.
  By Prop.~\ref{lem:isomorphism}, if $N_{\mho(G)}(x)$ induces a clique
  (noting that it contains no $\{v^l,v^r\}$ for any $v\in \oo$ as they
  are nonadjacent), then $N_{G}(\og{x})$ induces a clique as well.
  This verifies the ``only if'' direction.

  Suppose, for contradiction, that the ``if'' direction is false, then
  $x$ must be adjacent to some broken pair; let it be $y,z$.  They
  cannot be both in $L$ or both in $R$; on the other hand, since they
  have distance $2$, Prop.~\ref{lem:property-mho} rules out the
  possibility that one of them in $L$ and the other in $R$.  Thus, at
  least one of $y,z$ is in $\overline{\oo}$; without loss of
  generality, let $y\in\overline{\oo}$.  Then $x$ is in
  $\overline{\oo}$ as well; otherwise, $x\in L$, and $\og{x}$ is
  adjacent to two nonadjacent vertices $y$ and $h_0$.  Now \og{z} is
  adjacent to both $x$ and $y$.  The fact $y\not\sim z$ in $\mho(G)$
  implies $\og{z} x$ and $\og{z} y$ are in $\ecc$ and $\ec$ but not
  the same.  Without loss of generality, let $\og{z} x\in \ecc$ and
  $\og{z} y\in \ec$.  Thus, $z\in\oc\cap\occ$, and $x$ and $y$ are
  adjacent to $h_{-1}$ and $h_1$, respectively.  However, \og{x} has a
  pair of neighbors $h_1$ and $y$ that is nonadjacent to each other.
  This contradiction concludes the ``if'' direction and the proof.
\end{proof}

It is worth noting that even a vertex $v\in \oo$ is not simplicial in
$G$, it is still possible that $v^r$ is simplicial in $\mho(G)$.
Consider, for example, the graph of a $6$-hole.  One may want to
verify that such a vertex has to be in \oc.  This reveals the first
purpose of introducing $w$: for every vertex $v\in \occ$, the vertex
$v^l$ is adjacent to both $h^l_0$ and $w$, thereby excluding the
possibility of deriving a simplicial vertex in $L$ from a
non-simplicial vertex in $\oo$.

Consider a broken pair $x,y$.  If $x = v^l$ for some $v\in \oo$, then
$y$ is adjacent to $v^r$, and thus any \stpath{x}{y} $P$ can be
extended to a \stpath{v^l}{v^r} of length $||P|| + 1$, i.e., $|P|$.
As a result of Prop.~\ref{lem:property-mho}, the distance between a
broken pair is at least three.  This is strengthened by the following
lemma.
\begin{lemma}\label{lem:bad-pair}
  Let $X$ be a set of vertices of $\mho(G)$ that contains a broken
  pair and induces a connected subgraph.  We can in $O(||G||)$ time
  find a subgraph of $G$ in $\cF_{LI}$ if
  \begin{inparaenum}[(\itshape 1\upshape)]
  \item $|X|\le 5$, or
  \item there exists a vertex $h\in H$ nonadjacent to $\og{X}$.
  \end{inparaenum}
\end{lemma}
\begin{proof}
  It suffices to assume that $X$ is minimal, that is, it induces a
  path whose ends $x,y$ are the only broken pair in $X$.  Then
  $w\not\in X$: it does not participate in any broken pair, hence not
  an end of the path, and it is simplicial in $\mho(G)$, hence not an
  inner vertex of any induced path.  Recall that at least one of $x$
  and $y$ is in $L\cup R$; without loss of generality, let $x = v^l\in
  L$, then $y$ is adjacent to the vertex $v^r$.  By the definition of
  broken pair, $\og{X}$ induces a cycle in $G$; on the other hand, by
  Prop.~\ref{lem:property-mho}, $|X|\ge 4$.  Therefore, $G[\og{X}]$ is
  a hole.  We return it as a short hole in case (1) or call
  Lem.~\ref{lem:c-star} with \og{X} and $h$ in case (2).
\end{proof}

Now we are ready to prove the rest of Thm.~\ref{thm:3}, which is
separated into two statements.
\begin{lemma}\label{lem:mho-is-chordal}
  If $\mho(G)$ is not chordal, then we can in $O(||G||)$ time find a
  subgraph of $G$ in $\cF_{LI}$.
\end{lemma}
\begin{proof}
  We find a hole $C$ of $\mho(G)$.  Note that $w\not\in C$ as $w$ is
  simplicial in $\mho(G)$.  Let us first take care of two trivial
  cases.  In the first case, $C$ is disjoint from both $L$ and $R$,
  and \og{C} is a hole of $G$ (Prop.~\ref{lem:isomorphism}).  This
  hole is nonadjacent to $h_0$ in $G$, which enables us to call
  Lem.~\ref{lem:c-star}.  In the other case, all vertices of $C$ are
  from $L$ or $R$, and $\og{C}$ is a hole of $G$.  This hole has a
  common neighbor $h_0$, which enables us to call
  Lem.~\ref{lem:non-consecutive}.  Since $L$ and $R$ are nonadjacent,
  the second case above must hold if $C$ is disjoint from
  $\overline{\oo}$.  Henceforth we assume that $C$ intersects
  $\overline{\oo}$ and, without loss of generality, $L$; it might
  intersect $R$ as well, but this fact is irrelevant in the following
  proof.

  By the definition of $\mho(G)$ and Lem.~\ref{lem:O}, $\lc$ is a
  clique separator of $L\setminus \lc$ and $\overline{\oo}$.
  Therefore, $C$ contains at most two vertices of $\lc$ and is
  disjoint from $L\setminus \lc$.  We define two configurations based
  on whether $C$ is adjacent to $h^l_{-1}$ or not:
  \begin{itemize}
  \item [I.]  If $C$ contains some $x\in N(h^l_{-1})$, (pick either one
    if there are two such vertices,) then $\tail{\og{x}} = -1$ and
    $\head{\og{x}} = 1$.  Starting from $x$, we traverse $C$ (in
    either direction) till the first vertex $y$ that is adjacent to
    $h_{-2}$.
  \item [II.] Otherwise, we pick $x$ to be the vertex in $V(C)\cap L$
    such that $\head{\og{x}}$ is smaller.  Note that $\tail{\og{x}} =
    0$.  Starting from $x$, we traverse $C$ (in either direction) till
    the first vertex $y$ that is adjacent to $h^r_{-1}$.
  \end{itemize}
  We need to explain what to do if the vertex $y$ is not found after
  $C$ has been exhausted, which means that $h_{-2}$ or $h^r_{-1}$ is
  nonadjacent to the hole $C$.  We check whether $C$ contains a broken
  pair or not.  If yes, then we call Lem.~\ref{lem:bad-pair}(2);
  otherwise, \og{C} induces a hole in $G$
  (Prop.~\ref{lem:isomorphism}).  This hole is nonadjacent to $h_{-2}$
  (I) or $h_{-1}$ (II), which allows us to call Lem.~\ref{lem:c-star}.
  In the following we may assume that we have found the vertex $y$.

  Let $a := \head{\og{x}}$; note that $a<|H| - 2$ (otherwise \og{x} is
  adjacent to at least $|H| - 2$ vertices in $H$).  If $y\sim h^l_a$
  (when $a = 1$) or $y\sim h_a$ (when $a > 1$), then we have an
  \stpath{h^l_0}{h^r_0} $h^l_0 h^l_1 y h_{-2} h^r_{-1} h^r_0$ (I) or
  $h^l_0 x h_a y h^r_{-1} h^r_0$ (II), which enables us to call
  Lem.~\ref{lem:bad-pair}(1).

  Starting from $x$, we traverse both directions of the hole $C$ till
  the first vertices that are adjacent to $h_{a+1}$.  Let them be
  $x_1$ and $x_2$, and let $P_1$ and $P_2$ be the resulting paths $x
  \cdots x_1$ and $x \cdots x_2$, respectively.  If
  Lem.~\ref{lem:non-consecutive} or \ref{lem:non-consecutive-2}
  applies in the traversal, then we are done.  Otherwise, for every
  inner vertex in $P_1$ and $P_2$, its neighbors in $\{h^l_{-1},
  h^l_0\cdots, h^r_0, h^r_1\}$ are subsets of $\{h^l_{-1},
  h^l_0\cdots, h_a\}$, and both $x_1$ and $x_2$ are adjacent to $h_a$.
  Therefore, $y$ appears in neither path, which implies that $x_1\ne
  x_2$ and $x_1\not\sim x_2$.  We take the hole $C' := x \cdots x_1
  h_{a+1} x_2 \cdots x$.  Every vertex in $C'$ is adjacent to $h_a$
  (or $h^l_a$), and there is no broken pair in $C'$; thus \og{C'} is a
  hole of $G$ (Lem.~\ref{lem:isomorphism}).  Noting that $h_a$ is a
  common neighbor of this hole, we can call
  Lem.~\ref{lem:non-consecutive}.

  The main step is traversing $C$, which can be done in $O(||G||)$
  time.  
\end{proof}

Now we may assume that $\mho(G)$ is chordal, and we use
Lem.~\ref{lem:find-simplicial-vertices} to find the set $S$ of
simplicial vertices of $\mho(G)$.  According to
Lem.~\ref{lem:simplicial-vertices}, $\og{S\setminus (\{w\}\cup R)}$
gives the set of simplicial vertices of $G$.  We can then build the
subgraph $\mho(G - SI(G))$; using definition one can verify that it is
precisely $\mho(G) - \{v\in V(\mho(G)): \og{v}\in SI(G)\}$, hence a
subgraph of $\mho(G)$.  As we have pointed out, it is different from
$\mho(G) - SI(\mho(G))$.

\begin{lemma}\label{lem:negative-certificate}
  If $\mho(G - SI(G))$ is not an interval graph, then we can in
  $O(||G||)$ time find a subgraph of $G$ in $\cF_{LI}$.
\end{lemma}
\begin{proof}
  By assumption, we can find a subgraph of $\mho(G - SI(G))$ in
  $\cF_I$; let $X$ be its vertex set.  If it is a hole, then we call
  Lem.~\ref{lem:mho-is-chordal}.  Hence we may assume that $\mho(G) -
  SI(G)$ is chordal, and we have a caw of $\mho(G - SI(G))$.  Since
  the largest distance between any pair of vertices in a caw is $4$,
  if $X$ contains a broken pair, then we can call
  Lem.~\ref{lem:bad-pair}(1).  Now that $X$ is free of broken pairs,
  if $w\not\in X$, then by Prop.~\ref{lem:isomorphism}, \og{X} induces
  a caw in $G - SI(G)$.  We can return \og{X} if it is small, or call
  Thm.~\ref{thm:1} otherwise.  Therefore, in the remaining cases, $X$
  contains $w$, and then $X$ must intersect \lcc.  The nonexistence of
  broken pairs then implies that $X$ is disjoint from $R$: it is
  connected, and if it intersects $R$, then it contains some vertex
  $v^r$ for $v\in \occ$; since \occ\ induces a clique, we have a
  broken pair.  The simplicial vertex $w$
  (Prop.~\ref{lem:property-mho}) has to be one terminal of $X$ and has
  either one or two neighbors in it (see Fig.~\ref{fig:at}).  We
  search for a vertex $u\in \overline{\oo}$ such that $u\og{x}\in\ecc$
  for every neighbor $x$ of $w$ in $X$.  We break the rest of the
  proof into two cases based on whether there exists such a vertex.

  Assume first that such a vertex $u$ is found.  Note that this is the
  only case when $|N(w)\cap X| = 1$.  It is easy to verify that $u$
  and $x\in N(w)\cap X$ make a broken pair, which means $u\not\in X$.
  By assumption, $u$ is adjacent to $R$, and thus nonadjacent to $L$
  (Lem.~\ref{lem:one-side}).  If $X$ contains a neighbor $u'$ of $u$,
  then we consider the shortest \stpath{x}{u'} $P$ in the caw.  Since
  $x$ is a neighbor of $w$, it cannot be a terminal.  From
  Fig.~\ref{fig:at} it can be observed that $P$ consists of at most
  $4$ vertices.  We can extend $P$ by adding edge $u' u$, and this
  results in an \stpath{x}{u} of length at most $4$, which allows us
  to call Lem.~\ref{lem:bad-pair}(1).  Otherwise, ($X$ is disjoint
  from $N_{\mho(G-SI(G))}[u]$,), the subgraph of $G-SI(G)$ induced by
  $\og{X\setminus \{w\}'}\cup \{u\}$ must be isomorphic to the
  subgraph induced by $\mho(G)[X]$, hence a caw.  We can either return
  it as a small caw, or call Thm.~\ref{thm:1}.

  Assume now that $w$ has two neighbors $x_1$ and $x_2$ in
  $X\cap\lcc$, and we have two distinct vertices $y_1, y_2\in
  \overline{\oo}$ such that $\og{x_1} y_1, \og{x_2} y_2\in \ecc$.  By
  assumption, $\og{x_1}\not\sim y_2$ and $\og{x_2}\not\sim y_1$ in $G$
  (otherwise we have already been in the previous case).  Note that
  $y_1$ and $y_2$ are nonadjacent; otherwise, $\{y_1, y_2\}$ and the
  counterparts of $\{{x_1}, x_2\}$ in $R$ induce a hole of $\mho(G)$,
  which is impossible (we have assumed that it is chordal).  Since
  $\og{x_1},\og{x_2} \in \occ$, they are nonadjacent to $h_2$
  (Lem.~\ref{lem:hole-conditions}), i.e., $\head{\og{x_1}}$ and
  $\head{\og{x_2}}$ are either $0$ or $1$.  We proceed as follows
  (Fig.~\ref{fig:lem:negative-certificate}).

\begin{figure}[h!]
\setbox4=\vbox{\hsize28pc \noindent\strut
\begin{quote}
  \vspace*{-5mm} \small

  1 \hspace*{2ex} \parbox[t]{0.95\linewidth}{ {\bf if} $y_1\sim
    h_{\head{\og{x_1}} + 1}$ {\bf then return} $4$-hole ($y_1 \og{x_1}
    h_{0} h_1 y_1$) or ($y_1 \og{x_1} h_{1} h_2 y_1$);
    \\
    {\bf if} $y_2\sim h_{\head{\og{x_2}} + 1}$ {\bf then} {\sf
      symmetric as above};}
  \\
  2 \hspace*{2ex} {\bf if} $\head{\og{x_1}} = \head{\og{x_2}}$ {\bf
    then}
  \\
  \hspace*{6ex} {\bf return} net $\{y_1, \og{x_1}, y_2, \og{x_2},
  h_{\head{\og{x_2}}}, h_{\head{\og{x_2}}+1}\}$;
  \comment{Fig.~\ref{fig:hole-2}.}
  \\
  \hspace*{2ex} $\setminus\!\!\setminus$ assume from now that
  $\head{\og{x_1}} = 1$ and $\head{\og{x_2}} = 0$.
  \\
  3 \hspace*{2ex} {\bf if} $y_{2}\not\sim h_{-1}$ {\bf then return}
  net $\{y_1, h_{-1}, y_2, \og{x_2}, h_0,
  h_1\}$; \comment{Fig.~\ref{fig:hole-4}.}
  \\
  4 \hspace*{2ex} {\bf else return} rising sun $\{y_1, h_{-1},
  \og{x_1}, y_2, \og{x_2}, h_0, h_1\}$.  \comment{Fig.~\ref{fig:hole-3}.}
  \\[-2mm]
\end{quote} \vspace*{-6mm} \strut} $$\boxit{\box4}$$
\vspace*{-9mm}
\caption{Subroutine for the proof of
  Lem.~\ref{lem:negative-certificate} ($w\in X$).}
\label{fig:lem:negative-certificate}
\end{figure}

\begin{figure}[h!]
  \centering\footnotesize
  \begin{subfigure}[b]{0.32\textwidth}
    \centering \includegraphics{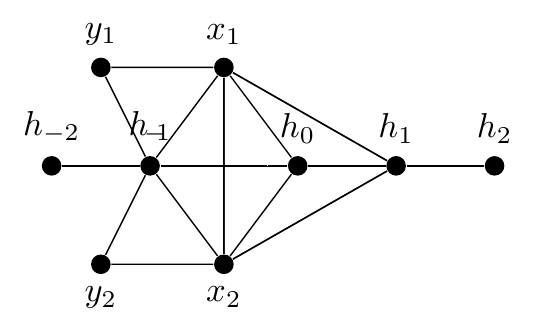}
    \caption{}
    \label{fig:hole-2}
  \end{subfigure}
  \,
  \begin{subfigure}[b]{0.32\textwidth}
    \centering \includegraphics{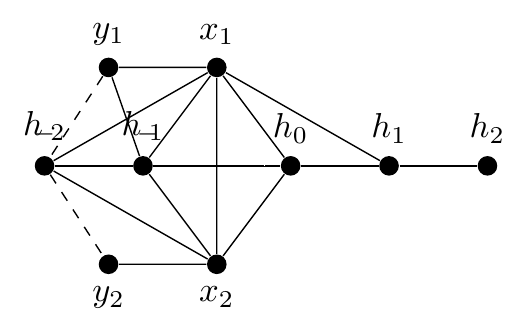}
    \caption{}
    \label{fig:hole-4}
  \end{subfigure}
  \,
  \begin{subfigure}[b]{0.32\textwidth}
    \centering \includegraphics{5-2.pdf}
    \caption{}
    \label{fig:hole-3}
  \end{subfigure}
  \caption{Structures used in the proof of
    Lem.~\ref{lem:negative-certificate}. ($\phi()$ is omitted.)}
\label{fig:negative-certificate-2}
\end{figure}

  We now verify the correctness of this subroutine.  All adjacencies
  used below are in $G$.  Step 1 considers the case where $y_1\sim
  h_{\head{\og{x_1}} + 1}$.  By construction
  (Lem.~\ref{lem:non-consecutive-2}) and noting that $\og{x_1}
  y_1\in\ecc$, it holds that $y_1\not\sim h_{\head{\og{x_1}}}$.  Thus,
  $y_1 \og{x_1} h_{1} h_2 y_1$ is a hole.  A symmetric argument
  applies when $y_2\sim h_{\head{\og{x_2}} + 1}$.  Now that the
  conditions of step~1 do not hold true, step~2 is clear from
  assumption.  Henceforth we may assume without loss of generality
  that $\head{\og{x_1}}> \head{\og{x_2}} \ge 0$.  According to
  Lem.~\ref{lem:hole-conditions}, the only possibility to make this
  true is $\head{\og{x_1}} =1$ and $\head{\og{x_2}}= 0$.
  Consequently, $\head{{y_1}} = |H| - 1$
  (Lem.~\ref{lem:non-consecutive-2}).  Steps~3 and 4 are clear from
  the assumptions above.

  The dominating steps are finding $X$ and calling appropriate
  subroutines, all of which can be done in $O(||G||)$ time.
\end{proof}

Lems.~\ref{lem:simplicial-vertices}, \ref{lem:mho-is-chordal}, and
\ref{lem:negative-certificate} together conclude Thm.~\ref{thm:3}.

\subsection{The olive-ring decomposition}
\label{sec:olive-ring-decomposition}
If we have not found a subgraph in $\cF_{LI}$, then $\mho(G - SI(G))$
must be an interval graph, which has a clique path decomposition.  We
now demonstrate how to employ it to build an olive-ring decomposition
for $G$.  This boils down to two steps, the first of which builds a
hole decomposition for $G - SI(G)$.  It is worth noting that $G -
SI(G)$ may or may not be prime, and the following proof does not use
any property of prime graphs.
\begin{lemma}\label{lem:build-hole-decomposition}
  Given a clique path decomposition for $\mho(G - SI(G))$, we can in
  $O(||G||)$ time construct a hole decomposition for $G - SI(G)$.
\end{lemma}
\begin{proof}
  For notational convenience, let $G'$ denote $G - SI(G)$.  Recall
  that $G$ is connected; since no inner vertex of a shortest path can
  be simplicial, $G'$ is connected as well.  Every vertex in $G'$ must
  be adjacent to $V(H)$, as otherwise Thm.~\ref{thm:1} should have
  been called during the construction of $\mho(G)$.

  Let $\cal P$ be the clique path decomposition for $\mho(G')$.  by
  the definition of clique decompositions, if $\mho(G')$ remains
  connected after the deletion of vertices in a bag $K$, then $K$ must
  be an end of $\cal P$.  Since $w$ is simplicial, its closed
  neighborhood $N[w]$, i.e., $\{w\}\cup \lcc$, is a maximal clique of
  $\mho(G')$.  We argue first that $\mho(G') - N[w]$ is connected by
  showing that every neighbor $x$ of $\{w\}\cup \lcc$ is connected to
  $h^l_0$ in $\mho(G') - N[w]$: note that if the deletion of a vertex
  set $X$ separates a connected graph, then there must be at least two
  neighbors of $X$ in different components of the resulting graph.
  Every vertex $x\in L\setminus \lcc$ is adjacent $h^l_0$, every
  vertex $x\in \overline\oo$ is adjacent to one of
  $h^l_1,h_2,\cdots,h^r_{-1}$, and every $x\in R$ is adjacent to
  $h^r_0$.  Since $h^l_0 h^l_1 h_2\cdots h^r_{-1} h^r_0$ remains a
  path in $\mho(G') - N[w]$, it must be connected.  Thus $N[w]$ is an
  end bag of $\cal P$; without loss of generality, let it be $K_0$.
  On the other hand, $\{v^r: v\in \occ\}$ is a minimal separator for
  $L\cup \overline \oo\cup\{w\}$ and $R\setminus \{v^r: v\in \occ\}$.
  Thus, a bag contains $h^r_0$ if and only if it appear to the right
  of this separator (we have agreed that $w$ is to the left of it).
  Let $\ell$ be the largest index such that $h^r_0\not\in K_\ell$.  We
  build the hole decomposition $\cal C$ as follows:
  \begin{inparaenum}[(\itshape 1\upshape)]
  \item take the sub-path ${\cal P}' := \{K_1\cdots K_\ell\}$ of $\cal
    P$,
  \item replace every vertex $x\in K_i$ with $1\le i\le \ell$ by
    \og{x}, and
  \item add an edge to connect bags $K_1$ and $K_\ell$.
  \end{inparaenum}

  The rest of the proof is to show that $\cal C$ is a clique
  decomposition for $G'$.  We verify first that the bags of $\cal C$
  are precisely the set of maximal cliques of $G'$, that is,
  \begin{inparaenum}[(1)]
  \item every maximal clique of $G'$ appears exactly once in
    $\cal C$, and
  \item every bag of $\cal C$ is a maximal clique of $G'$.
  \end{inparaenum}
  For any $1\le i\le \ell$, the bag $K_i$ does not contain $w$, and
  thus there exists some maximal clique $X$ of $\mho(G')$ such that
  $\og{X} = K_i$.  By Prop.~\ref{lem:isomorphism}, $K_i$ is a clique
  of $G'$.  Therefore, for condition (2), it suffices to show that
  every bag of $\cal C$ is maximal, and hence conditions (1) and (2)
  are equivalent to the following statement.
  \begin{claim}
    Let $K$ be a maximal clique of $G'$.  It holds that
  \begin{inparaenum}[(1)]
  \item $K$ appears exactly once in $\cal C$, and
  \item for any set $X$ of vertices of $\mho(G')$ with $\og{X} \subset
    K$, it is not a maximal clique of $\mho(G')$.
  \end{inparaenum}
  \end{claim}
  \begin{proof}
    We consider the intersection between $K$ and $\oo$, which has
    three cases.  In the first case, $K$ is disjoint from $\oo$.  Then
    $K$ induces a clique of $\mho(G')$ as well, and both conditions
    hold trivially.  In the second case, $K\subseteq \oo$. The
    maximality of $K$ implies $h_0\in K$.  Then $\{v^l: v\in K\}$ and
    $\{v^r: v\in K\}$ induce two disjoint cliques in $\mho(G')$, which
    are subsets of $L$ and $R$, respectively.  The bag in $L$ appears
    in $\cal C$, while the bag in $R$ contains $h^r_0$, and thus does
    not appear in $\cal C$.  Therefore, $K$ appears only once in $\cal
    C$, and (1) is satisfied.  For (2), if $X$ intersects both $L$ and
    $R$, then it does not induce a clique; otherwise it is a proper
    subset of $\{v^l: v\in K\}$ or $\{v^r: v\in K\}$, and thus induces
    a clique of $\mho(G')$ which is not maximal.

    In the remaining case we assume that $K$ intersect both $\oo$ and
    $\overline \oo$.  The set $K'$ of vertices of $\mho(G')$ derived
    from $K$ is $(K\setminus \oo) \cup \{v^l, v^r: v\in K\cap \oo\}$.  By
    construction of $\mho(G')$ (and noting Lem.~\ref{lem:one-side}),
    all edges between $K\cap \oo$ and $K\setminus \oo$ must be in either
    \ecc\ or \ec.  Consider first that these edges are in \ecc, then
    $K\cap \oo\subseteq \occ$ and $K\setminus \oo$ is nonadjacent to $L$
    in $\mho(G')$.  The only maximal clique of $\mho(G')$ contained in
    $K'$ is $X = (K\setminus \oo) \cup \{v^r: v\in K\cap \oo\}$, and
    $\og{X} = K$.  Thus, both (1) and (2) are satisfied.  A symmetric
    argument applies when all edges between $K\cap \oo$ and $K\setminus
    \oo$ are in \ec.  
    \renewcommand{\qedsymbol}{$\lrcorner$}
  \end{proof}

  It remains to verify that the hole $\cal C$ satisfies the other
  condition of the definition of clique decomposition, i.e., for every
  $v\in G'$, the bags containing $v$ induces a sub-path of $\cal C$.
  Recall that every vertex of $\mho(G')$ appears in consecutive bags
  of $\cal P$.  For $v\in \overline \oo$, all bags containing $v$
  remain in $\cal P'$, and are consecutive.  For $v\not\in
  \oo\setminus\occ$, a bag of $\cal C$ contains $v$ if and only if the
  corresponding bag in $\cal P$ contains $v^l$ and $h^l_0$, which are
  consecutive.  Assume now that $v\in \occ$.  On the one hand,
  $\occ\subset K_1$; hence, $v\in K_1$, and the bags of $\cal P'$
  obtained from bags of $\cal P$ containing $v^l$ appear consecutively
  in the left end of $\cal P'$.  On the other hand, if $\cal P'$
  contains bags of $\cal P$ containing $v^r$, then $v\in K_\ell$, and
  they appear in the right end of $\cal P'$.  After the edge is added
  between $K_1$ and $K_\ell$, these bags are connected into a sub-path
  in $\cal C$.  Observing $|H| \ge 6$, the decomposition $\cal C$ is
  clearly a hole, which concludes the proof.
\end{proof}

For any pair of adjacent vertices of $G - SI(G)$, bags containing them
are consecutive in its hole decomposition $\cal C$, and by
Lem.~\ref{lem:one-side}, they induce a proper sub-path of $\cal C$.
Note that bags containing $h_0$ and $h_1$ are not subsets of each
other.  We have thus a direction for $\cal C$: bags containing $h_1$
is \emph{clockwise} from those containing $h_0$.  We can number bags
in $\cal C$ such that $K_{i-1}$ (resp., $K_{i+1}$) is
{counterclockwise} (resp., {clockwise}) from $K_i$.  We use \lint{v}
(resp., \rint{v}) to denote the index of the smallest (resp., largest)
index of bags that contain $v$.

We now consider the shortest holes in $G$,  for which we work on the
hole decomposition $\cal C$ constructed in
Lem.~\ref{thm:build-hole-decomposition} (noting that $G$ and $G -
SI(G)$ have the same set of holes).  Since the deletion of vertices in
any bag from $G$ leaves an interval subgraph, a hole pass through all
bags of $\cal C$.  On the other hand, if there is a set of two or
three pairwise adjacent vertices that intersects every bag in $\cal
C$, then Lem.~\ref{lem:one-side} must have been applied in the
construction of $\mho(G)$.  Therefore, an inclusion-wise minimal set
of vertices intersects every bag in $\cal C$ if and only if it defines
a hole.  This observation allows us to decide the length of the
shortest holes.  We have calculated \lint{v} and \rint{v} for all
vertices $v\in V(G)\setminus SI(G)$.  For each bag $K$ in the main
cycle, we can also in linear time find the vertices $v_1$ and $v_2$ in
$K$ that achieves the minimum value for \lint{v_1} and the maximum
value for \rint{v_2}, respectively.
\begin{lemma}\label{lem:find-shortest-hole}
  We can in $O(||G||)$ time find a shortest hole of $G$.
\end{lemma}
\begin{figure}[h!]
\setbox4=\vbox{\hsize28pc \noindent\strut
\begin{quote}
  \vspace*{-5mm} \small

  0 \hspace*{2ex} build a list $U$ of $|K_0|$ slots, each of which
  contains ($u,u$) for a distinct vertex $u\in K_0$;
  \\
  1 \hspace*{2ex} order $U$ such that \lint{u} is nondecreasing; {\sf
    reached}$= 0$;
  \\
  2 \hspace*{2ex} {\bf for each} $(u,v)\in U$ {\bf do}
  \\
  2.1 \hspace*{3ex} {\bf if} $v\sim u$ and $v$ is to the left of $u$
  {\bf then return} $u$; \comment{a shortest hole found.}
  \\
  2.2 \hspace*{3ex} {\bf if} ${v}\not\in K_{\sf reached}$ or $\rint{v}
  = {\sf reached}$ {\bf then} remove ($u,v$) from $U$;
  \\
  2.3 \hspace*{3ex} {\bf else} ${\sf reached} = \rint{v}$; $v=$ the
  vertex of $K_{\sf reached}$ that reaches the rightmost \rint{v};
  \\
  3 \hspace*{2ex} {\bf goto} 2.

\end{quote} \vspace*{-6mm} \strut} $$\boxit{\box4}$$
\vspace*{-9mm}
\caption{Finding a shortest hole in an olive-ring decomposition.}
\label{fig:alg-shortest-hole}
\end{figure}
\begin{proof}
  Starting from any vertex $v$, we can find a hole as follows.  Let
  $v_0 = v$; for each $i=0,\ldots$, we take the vertex in
  $K_{\rint{v_i}}$ that reaches the rightest bag as the next vertex
  $v_{i+1}$; the process stops when the first vertex $v_j$ is adjacent
  to $v_0$ again.  Note that $v_j$ might also be adjacent to $v_1$,
  but by the selection process, never adjacent to $v_2$.  We return
  the hole $v_0\cdots v_j v_0$ or $v_1\cdots v_j v_1$.  It is easy to
  verify that any hole through $v$ cannot be shorter that the this
  hole (note that $v$ may belong to no holes at all).  The removal of
  all vertices in any bag from the graph will make it chordal, which
  means that a hole has to intersect every bag in $\cal C$.
  Therefore, it suffices to find a vertex in bag $K_0$ that is in some
  shortest hole of $G$.  A trivial implementation will take
  super-linear time, and thus the main focus will be an efficient
  implementation.  For each vertex $v\in K_0$, we do the same search,
  but we terminate it as soon as we are sure that there are shortest
  holes avoiding $v$.  The procedure is described in
  Fig.~\ref{fig:alg-shortest-hole}.

  Each iteration of step 2 finds the next vertex for the shortest hole
  starting from a vertex $u\in K_0$.  For the sake of simplicity, only
  the current vertex $v$ of this search is recorded with the starting
  vertex $u$: let ($u,v$) be a pair in $U$ at the end of the $i$th
  step, then $v$ is the $i$th vertex of the shortest hole through $u$
  found by previous subroutine.  Denote by $P[u]$ the set of vertices
  that have been associated with $u$ during the execution of this
  procedure.  Step 2.3 replaces $v_i$ by the next vertex $v_{i+1}$ for
  the starting vertex $v_0 = u$.  Note that this new vertex, if in
  $K_0$, will not be adjacent to $v_1$: it is in $U$ and after $v$, and
  thus in the second run of step 2, $v$ should have been removed from
  $U$.  Therefore, if the condition of step 2.1 is satisfied, then
  $P[u]$ will induce a shortest hole $C$ via $u$.  Clearly, for any
  $u$ that remains in $U$, any shortest hole that passes $u$ is no
  shorter than $C$.  Thus, it suffices to show that this also holds
  true for vertices $u'$ that have been removed from $U$ in step~2.2.
  Suppose there is a shortest hole $C'$ through $u'$.  Then there is
  such a hole through all vertices in $P[u']$.  Since the vertex of
  $u$ preceding $u'$ satisfies $\lint{u}\le\lint{u'}$, replacing
  $P[u']$ by the first $|P[u']|$ vertices of $P[u]$ in $C'$, we must
  obtain a cycle such that the arcs for its vertices cover the whole
  circle.  Thus we obtain a hole through $u$ that is no longer than
  $C'$, i.e., strictly shorter than $C$.  This contradiction justifies
  the correctness of the procedure.

  Using a circular linked list for storing $U$, the procedure can be
  implemented in $O(||G||)$ time: each bag is checked at most once.
\end{proof}

If the hole found by Lem.~\ref{lem:find-shortest-hole} is short, i.e.,
having length four or five, then we return it.  Otherwise we add back
$SI(G)$ to build the olive-ring decomposition for $G$ as follows.  As
promised, here we need to take into consideration the case when $G$ is
chordal.  Recall that $G$ is prime and thus connected.  We find
$SI(G)$ and check $G - SI(G)$.  If it is not an interval graph, we can
either find and return a small caw, or obtain a large caw and call
Thm.~\ref{thm:1}.  If $G - SI(G)$ is a clique, then we have a clique
tree decomposition for $G$ that is a star, which is trivially a
caterpillar.  Otherwise, the clique path decomposition for $G - SI(G)$
has at least two bags, and we can extend it by appending two bags to
its ends respectively.  This gives a path of at least $4$ bags, which
can be treated as a hole decomposition by adding an edge connecting
the end bags of the path.  This allows us to handle the chordal case
and non-chordal case, i.e., the construction of caterpillar
decomposition and olive-ring decomposition, in a unified way.

By definition, the closed neighborhood of each simplicial vertex
defines a maximal clique, which we either insert into the main
cycle/path (replacing an existing bag or as a new bag), or attach to
the main cycle/path as a pendant bag.  We would like to accommodate
the bags in the main cycle/path as long as it is possible, or
equivalently, we prefer to minimize the number of pendant bags.
However, no caw can have a clique hole decomposition, and thus at
least one vertex of it has to be left out of the main cycle; it is of
course one of the terminals.

\begin{lemma}\label{lem:olive-ring-decomposition}
  Given a hole decomposition for $G-SI(G)$, we can in $O(||G||)$ time
  construct a clique decomposition for $G$ that is an olive ring.  A
  vertex of $G$ does not appear in any bag of the main cycle if and
  only if it is a terminal of a caw of which all other vertices are in
  the main cycle.
\end{lemma}
\begin{proof}
  Let $\cal C$ be the hole decomposition for $G-SI(G)$.  We traverse
  all bags in $\cal C$, and record \lint{v} and \rint{v} for each
  $v\in V(G)\setminus SI(G)$.  We also record $|K_i\cap K_{i+1}|$ for
  every $i=0,\dots, |{\cal C}| - 1$.  This can be done as follows: for
  each $v\in V(G)\setminus SI(G)$, we add one to $|K_i\cap K_{i+1}|$
  (modulo $|\cal C|$ when $i<0$) with $i\in[\lint{v}, \rint{v}]$.  For
  every simplicial vertex $s\in SI(G)$, we calculate $p = \max_{v\in
    N(s)}\lint{v}$ and $q = \min_{v\in N(s)}\rint{v}$.  We proceed as follows.
  \begin{enumerate}[(\itshape 1\upshape)]
  \item If $p = q$ and $K_p = N(s)$, then we add $s$ into $K_p$.
  \item If there is an $\ell$ such that $\ell\in [p, q-1]$ and
    $|K_\ell\cap K_{\ell+1}| = |N(s)|$, then we insert $N[s]$ as a new
    bag between $K_\ell$ and $K_{\ell+1}$.  By the selection of $p,q$
    and the fact that $\cal C$ is a clique decomposition, it can be
    inferred that $N(s)\subseteq K_\ell\cap K_{\ell+1}$; they have the
    same cardinality if and only if they are equivalent.
  \item Otherwise we add $N[s]$ as a bag pendant to $K_p$.
  \end{enumerate}
  It is easy to verify that the obtained decomposition is an
  olive-ring decomposition.

  The insertion of a new bag into the hole (the second case) will
  change the indices of bags in the main cycle; if we always update
  $\lint{v}$ and $\rint{v}$ for every $v$ accordingly, it cannot be
  done in linear time.  Thus, for implementation, we may instead mark
  the position for each $N[s]$ and add them after all the positions
  for them have been decided.  This is justified because the insertion
  of a new bag in between does not change the $K_i\cap K_{i+1}$.  The
  first stage, computation of \lint{v} and \rint{v} for all vertices,
  can be done in one run, which takes $O(||G||)$ time.  With these
  indices, we can in $O(||G||)$ time calculate all the sizes $|K_i\cap
  K_{i+1}|$.  The detection of the position for a simplicial vertex
  $s$ takes $O(d(v))$ time, and hence $O(||G||)$ time in total.  In
  summary, the algorithm takes $O(||G||)$ time.

  The ``if'' direction of the second assertion is trivial, and we now
  verify the ``only if'' direction, i.e., for a simplicial vertex $s$
  such that we make $N[s]$ a pendant bag, there exists a caw where $s$
  is the only vertex not in the main cycle.  By assumptions, there are
  vertices $x\in K_{p-1}\setminus K_p$ and $y\in K_{q+1}\setminus
  K_q$.  Since for each $\ell\in[p, q-1]$, the set $(K_\ell\cap
  K_{\ell+1}) \setminus N(s)$ is nonempty, There is an \stpath{x}{y}
  with all inner vertices in them; this path is in the main cycle and
  avoids the neighbors of $s$.  We can thus find an \stpath{s}{x} via
  an $x'\in N(s)$ with $\rint{x'} = p$, and an \stpath{s}{y} via a
  $y'\in N(s)$ with $\lint{y'} = q$.  These three paths witness
  $\{s,x,y\}$ as an at.  The union of vertices in these paths is not
  necessarily a caw, but it must contain some caw.  Since all other
  vertices are from the main cycle, this caw necessarily contains $s$;
  moreover, since $s$ is simplicial in $G$, it is a terminal of this
  caw.  This concludes the proof.
\end{proof}

In the case when $G$ is chordal, the main cycle always has a special
edge that connects two disjoint bags (they may or may not be the
original end bags of the clique path decomposition for $G - SI(G)$).
Detaching them gives a caterpillar decomposition for $G$.
\begin{corollary}
  Given a clique path decomposition for $G-SI(G)$, we can in
  $O(||G||)$ time construct a clique decomposition $\cal K$ for $G$
  that is caterpillar.
\end{corollary}

We use $SP(G)$ to denote those vertices contained only in pendent
bags.  Clearly $SP(G)\subseteq SI(G)$, i.e., every vertex is $SP(G)$
is simplicial.

\subsection{Normal Helly circular-arc graphs}
\label{sec:nhcag}
We have now finished constructing the decomposition for prime graphs.
Before closing the first phase, let us have another look at the
decomposition in retrospect.  In the study of interval graphs, we have
used both clique path decompositions and interval models; they are
essentially equivalent and can be transformed to each other
efficiently (though strictly speaking, interval models are more space
efficiently because they need $\Theta(|G|)$ space only).  For example,
a clique path decomposition can be readily read from an interval
model, when it suffices to check the endpoints.  The other
direction is given by the following proposition.  Note that a
simplicial vertex $v$ appears in a single bag, i.e., $\lint{v} =
\rint{v}$, and we use $\pm 1/3$ to force
$\lp{v}\ne\rp{v}$.\footnote{Any positive constant strictly less than
  $1/2$ will serve this purpose.  Without using this adjustment (and
  thus allowing degenerated intervals that are points), we can obtain
  the \emph{standard} interval model
  \cite{korte-89-recognizing-interval-graphs}, which is known to have
  the minimum number of different endpoints.  The other extreme is a
  {normalized} interval mode we have already met, which has the
  maximum number, $2 |G|$, of distinct endpoints.  Clearly, the model
  given by Prop.~\ref{lem:clique-decomposition-interval-model} is not
  normalized in general.}
\begin{proposition}\label{lem:clique-decomposition-interval-model}
  Given a clique path decomposition $\cal K$ for a graph $G$, an
  interval model for $G$ can be obtained by setting $I_v := [{\lint{v}}
  - 1/3, {\rint{v}}+1/3]$.
\end{proposition} 
One may have noticed that the proofs in
Section~\ref{sec:module-preservation} exclusively use interval models,
while the constructions of this section use merely clique
decomposition.  The reason is that both representations have their
special advantages not shared by the other.  On the one hand, the
distances between immediate neighboring endpoints in a normalized
interval model give us leeway to manipulate.  On the other hand,
clique path decompositions are more convenient to accommodate the leaf
bags, i.e., compatible with caterpillar decompositions.

The hole decompositions suggest yet another way to generalize interval
graphs, i.e., to use a circle and arcs in place of the real line and
intervals respectively, and then a pair of vertices is adjacent if and
only if their corresponding arcs intersect.  Such a model is a
\emph{circular-arc model}, and a graph having a {circular-arc model}
is a \emph{circular-arc graph}.  In a circular-arc model, each vertex
$v$ corresponds to a closed arc $A_v = [\lp{v}, \rp{v}]$, where
$\lp{v}$ and $\rp{v}$, the counterclockwise and clockwise endpoints of
$A_v$ respectively, are assumed to be nonnegative and distinct.  We
point out that possibly $\lp{v}>\rp{v}$; such an arc $A_v$ necessarily
passes through the point $0$.  Similar as
Prop.~\ref{lem:clique-decomposition-interval-model}, we can retrieve a
circular-arc model from a hole decomposition as follows.  (Recall that
$\lint{v}>\rint{v}$ if and only if both $K_0$ and $K_{|{\cal C}| - 1}$
contain $v$, and ${\lint{v}} = {\rint{v}}$ if and only if $v$ is a
simplicial.)
\begin{proposition}\label{lem:clique-decomposition-arc-model}
  Given a hole decomposition $\cal K$ for a graph $G$, a circular-arc
  model with circle length $|\cal K|$ for $G$ is given by
  \[
  A_v := \begin{cases}
    [{\lint{v}} - 1/3, {\rint{v}}+1/3] & \text{if } {\lint{v}}> 0,
    \\
    [|{\cal K}| - 1/3, {\rint{v}}+1/3] & \text{if }
    {\lint{v}}= 0.
  \end{cases}
  \]
\end{proposition}

However, not every circular-arc graph has a hole decomposition, e.g.,
some circular-arc graph might have an exponential number of maximal
cliques \cite{tucker-80-recognition-cag}.  The class of circular-arc
graphs is far less understood and harder to manipulate than interval
graphs, which can be attributed to some pathologic intersecting
patterns that are only allowed in circular-arc models.  Most notably,
two arcs might intersect at both ends, and three or more arcs might
pairwise intersect but contain no common point.  A circular-arc model
that is free of these two patterns is \emph{normal and Helly}, and a
graph having such a model is a \emph{normal Helly circular-arc graph}.
Trivially, every interval model can be viewed as a circular-arc model
that is normal and Helly, which means that every interval graph is
also a normal Helly circular-arc graph.  The Helly property ensures
that every maximal clique corresponds to some point in the
circular-arc model.  As a consequence, a Helly circular-arc graph $G$
has at most $|G|$ maximal cliques and admits a clique decomposition
that is either a path or a hole \cite{gavril-74-algorithms-cag}.  It
is not hard to check that no caw is a normal Helly circular-arc graph
(see also \cite{cao-14-recognizing-nhcag}), and thus the second
assertion of Lem.~\ref{lem:olive-ring-decomposition} can be
interpreted as follows.
\begin{corollary}\label{lem:hole=no-caw}
  The subgraph $G - SP(G)$ is a normal Helly circular-arc graph, and
  it is maximal in $G$.
\end{corollary}

In a circular-arc model, the arcs for every hole is minimal in the
sense that their union covers the entire circle; in other words, every
point on the circle is contained in at least one arc of them.  One can
easily see the correlation between the neighbors of vertex $v$ in a
hole $H$ and its arc in a circular-arc model.  Indeed, conditions in
Lem.~\ref{lem:one-side} are simply the normal and Helly properties: a
circular-arc model is normal and Helly if and only if there are no
three or less arcs covering the whole circle
\cite{mckee-03-restricted-cag,lin-13-nhcag-and-subclasses}.
Therefore, the circular-arc model given by
Prop.~\ref{lem:clique-decomposition-arc-model} must be normal and
Helly.  Let $A_u$ and $A_v$ be a pair of intersecting arcs such that
neither of them is contained in the other.  The normal property
ensures that they intersect only in one end; we say that the arc $A_u$
is counterclockwise (resp., clockwise) to $A_v$ if $\rp{u}\in A_v$
(resp., $\lp{u}\in A_v$).  In particular, the arc for $h_{i+1}$ is
clockwise to that for $h_{i}$, which is consistent as the direction of
the vertices themselves.  Moreover, we also say that $h_{i+1}$ is to
the right of $h_{i}$ and $h_{i}$ is to the right of $h_{i+1}$; this
relation can be viewed by an observer placed at the center of the
circle.

In passing we would like to have a final remark on the graph classes
mentioned in this paper.  Although all classes of \lig s, chordal
graphs, at-free graphs, and normal Helly circular-arc graphs contain
interval graphs as a proper subset, they are incomparable to each
other.  This fact can be evidenced by the four graphs depicted in
Fig.~\ref{fig:superclasses}---they belong to prime \lig s, chordal
graphs, at-free graphs, and normal Helly circular-arc graphs,
respectively, but no others.  For a comprehensive treatment and for
references to the extensive literature on various graph classes, one
may refer to Golumbic~\cite{golumbic-2004-perfect-graphs},
Brandst\"{a}dt et al.~\cite{brandstadt-99-graph-classes}, and
Spinrad~\cite{spinrad-03-efficient-graph-representations}.
\begin{figure*}[h] 
  \centering
  \begin{subfigure}[b]{0.22\textwidth}
    \centering \includegraphics{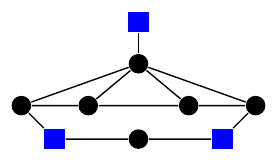} 
    \caption{a \lig}
    \label{fig:net}
  \end{subfigure}%  
  \,
  \begin{subfigure}[b]{0.20\textwidth}
    \centering \includegraphics{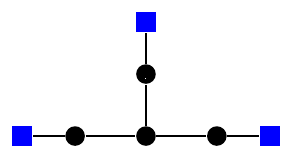} 
    \caption{a chordal graph}
    \label{fig:net}
  \end{subfigure}%  
  \,
  \begin{subfigure}[b]{0.20\textwidth}
    \centering \includegraphics{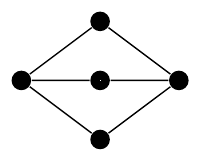} 
    \caption{an at-free graph}
    \label{fig:net}
  \end{subfigure}%  
  \,
  \begin{subfigure}[b]{0.24\textwidth}
    \centering \includegraphics{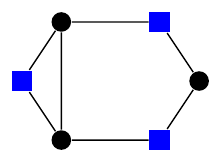} 
    \caption{a circular-arc graph}
    \label{fig:net}
  \end{subfigure}%  
  \caption{Super classes of interval graphs (squared vertices make an at).}
  \label{fig:superclasses}
\end{figure*}

 \section{Large caws}\label{sec:caw}
This section concentrates on the disposal of large caws in the
quotient graph $Q$ defined by the maximal strong modules of the input
graph $G$, conditioned on that
Thms.~\ref{thm:modules-induced-subgraph+} and
\ref{thm:modules-subgraph+} are applicable.  It is also assumed that
every \msm\ induces an interval subgraph.  Recall that the frame of a
large caw ($s: c_1,c_2:l,B,r$) is the subgraph induced by its
terminals as well as their neighbors, i.e., ($s: c_1,c_2:l,b_1;
b_d,r$), where possibly $c_1=c_2$.  When no particular base is of
interest to us, $b_1$ and $b_d$ will be referred to by $l_b$ and
$r_b$.  As explained in Section~\ref{sec:hardness-find-fis}, we cannot
assume the nonexistence of small caws or short holes in $Q$, but
whenever they surface (recall that a subgraph of $Q$ in $\cF_{LI}$
directly translates a isomorphic subgraph of $G$), the current process
is preempted, so they will not concern us.  As a consequence, we may
assume that we have already an olive-ring decomposition $\cal K$ for
$Q$, and every non-simplicial element of $V(Q)$ represents a clique
module of $G$.  Now that we need to discuss three graphs, $G$, $Q$,
and $\cal K$, at different levels, for the sake of clarity, we would
use vertices, modules, and bags to refer to elements in $V(G)$,
$V(Q)$, and $V(\cal K)$, respectively.  In particular, each bag in
$\cal K$ is a maximal clique of $Q$, hence a set of modules of $G$.
Denote by $\#v$ the size of the module $v\in V(Q)$, and $\#X =
\sum_{v\in X} \#v$ for $X\subseteq V(Q)$; note that these modules are
maximal strong modules of $G$, hence disjoint.

We are looking for solutions satisfying
Thms.~\ref{thm:modules-induced-subgraph+}-\ref{thm:modules-supergraph+}.
In such an object interval graph $G'$, every maximal strong module of
$G$ (if not completely deleted in the vertex deletion problem) remains
a module of $G'$, and they define a quotient graph $Q'$ of $G'$.  One
should be noted that these modules are not necessarily strong in $G'$,
and $Q'$ is not necessarily prime either.  For the vertex deletion
problem, only the deletion of an entire maximal strong module is
reflected in $Q'$, while the edge modifications applied between
modules can be viewed as edge modifications in $Q$.  In other words,
$Q'$ is an induced interval subgraph, spanning interval subgraph, and
interval supergraph of $Q$, but not necessarily optimum in general.
Indeed, given a normalized interval model for the object interval
graph $G'$, we can extract a normalized interval model for the
modified quotient graph $Q'$.

Our algorithms find a large caw whose frame satisfies a certain
minimality condition in $Q$, from which they identify a bounded number
of vertex/edge sets that intersect some minimum modification.  Hence
they branch on deleting/adding one of these vertex/edge sets.  To show
the existence of such a minimum modification, we will use a similar
constructive argument as Section~\ref{sec:module-preservation}, i.e.,
given a minimum modification that avoids all those vertex/edge sets,
we construct another minimum modification containing one of them and
of the same size.  To verify that the graph obtained by applying the
new modification is indeed an interval graph, we use
Prop.~\ref{lem:lig-and-modules}, and check all its conditions are
satisfied by the newly constructed graph; in particular, we will use
$Q'$ discussed above as the quotient graph.  According to
Thm.~\ref{thm:modules-induced-subgraph+},
Lem.~\ref{thm:clique-modules-subgraph+}, and
Thm.~\ref{thm:preserving-modules}, we have to keep clique modules
(with the only exception that it is completely deleted) intact.
Recall that non-terminal modules of a caw in $Q$ are necessarily
non-simplicial, which must be clique modules, and hence we are mainly
concerned with the terminals, which might or might not be clique
modules.  The conditions of Prop.~\ref{lem:lig-and-modules} hold
vacuously for the vertex deletion problem, while for the edge
modification problems, we need to make sure that we do not add edges
to connect two non-clique modules, or delete edges from the
neighborhood of a non-clique module.

Besides the operation ``project'' defined at the beginning of
Section~\ref{sec:module-preservation}, we will need the following
operation on circular-arc models.  Given a pair of distinct positive
numbers $\alpha,\beta$ and a pair of nonadjacent vertices $u,v$, we
can always reverse, rotate, and/or re-scale all arcs such that $\lp{u}
= \alpha$ and $\rp{v} = \beta$.    

\subsection{The definition and detection of minimal frames}
\label{sec:minimal-frames}
Let $\cal K$ be the olive-ring decomposition of $Q$ and let $\cal C$
be the unique cycle in $\cal K$.  Recall that $SP(Q)$ denotes the set
of simplicial modules of $Q$ that do not appear in any bag of $\cal
C$.  We may assume in this section that $SP(Q)$ is nonempty, as
otherwise according to Prop.~\ref{lem:hole=no-caw}, $Q$ contains no
caws and we have nothing to do.  Non-terminal modules of a caw $W$ are
not simplicial in $Q$ and hence appear in $\cal C$; we consider the
bags in $\cal C$ containing these non-terminal modules of $W$.  Since
they induce a connected subgraph of $Q$ with diameter one or two, they
are contained in a proper and consecutive fraction of $\cal C$, i.e.,
a sub-path.  There is an inclusion-wise minimal sub-path of $\cal C$
that contains all non-terminal modules of $W$; clearly, this is
uniquely determined by the frame of $W$, and is the same for any caw
with this frame.

Here come the formal definitions.  Recall that for $v\in V(Q)\setminus
SP(Q)$, \lint{v} and \rint{v} are indices of the leftmost and
rightmost, respectively, bags in $\cal C$ containing $v$.  The
\emph{container} of a frame $F$, denote by $[F]$, is defined to be the
set of bags with indices $\rint{l_b}, \ldots, \lint{r_b}$, which is
clearly the inclusion-wise minimal set of consecutive bags whose union
contains all non-terminal modules of $F$.  On the other hand, the set
of \emph{inner modules} of $F$ is defined to be $IN(F) :=
\bigcup_{\ell\in [\rint{l_b}+1,\lint{r_b}-1]} K_{\ell}\setminus
N_Q(s)$, which is disjoint from $V(F)$.
\begin{definition}\label{def:min-frame}
  A frame $F$ is \emph{minimal} if there exists no other frame $F'$
  such that $[F']\subset [F]$.  
\end{definition}

\begin{SCfigure}[][h]
  \centering  \includegraphics{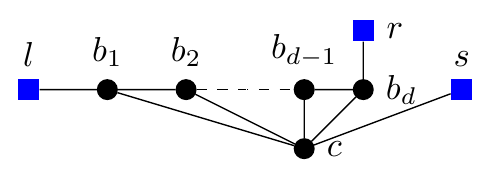} 
  \caption{A shallow terminal $s$ might appear in the main circle
    (path) of $\cal K$, here it is $\{l, b_1\}$, $\{c, b_1,b_2\}$,
    $\cdots$, $\{c, b_{d-1},b_d\}$, $\{c, s\}$.}
  \label{fig:s-in-C}
\end{SCfigure}

According to Lem.~\ref{lem:olive-ring-decomposition}, every module in
$SP(Q)$ is a terminal of some caw in $Q$.  Therefore, a natural
starting point for finding a caw in $Q$ is an $s\in SP(Q)$.  As it
will become clear in due course, we would prefer the starting module
$s$ to be either a terminal of a small caw or the shallow terminal of
a large caw.  However, a module in $SP(Q)$ might satisfy neither of
them; see, e.g., Fig.~\ref{fig:s-in-C}.  For this particular example,
we observe that bags containing $N_Q(r)$ are a proper subset of those
containing $N_Q(s)$, and we can make an alternative caterpillar
decomposition whose main path contains $V(Q)\setminus \{s\}$.  This
observation suggests us to give priority to those simplicial modules
that are adjacent to less bags in the main cycle during the
construction of $\cal K$.

We revisit the construction of the caterpillar/olive-ring
decomposition in Section~\ref{sec:olive-ring-decomposition} (the proof
of Lem.~\ref{lem:olive-ring-decomposition}) and use the following
order for handling simplicial modules.  Recall that at the onset we
have the set $SI(Q)$ of simplicial modules and a hole decomposition
for $Q - SI(Q)$.  For each simplicial module $s$, we calculate
$\max_{v\in N_Q(s)} \lp{v}$ and $\min_{v\in N_Q(s)} \rp{v}$; all these
$2|SP(Q)|$ indices can be computed in linear time.  We build an array
of size $|\cal C|$, of which the $i$th slot stores all modules $v$
with $\max_{v\in N_Q(s)} \lp{v} = i$.  Modules in each slot are sorted
in the nondecreasing order of the indices $\min_{v\in N_Q(s)} \rp{v}$.
With this array we can retrieve simplicial modules in order $s_1,
s_2,\ldots$ such that $N_Q(s_j)\not\subset N_Q(s_i)$ for every $i <
j$.  The construction used in the proof of
Lem.~\ref{lem:olive-ring-decomposition} does not assume any particular
order of $SI(Q)$, thereby still working; the running time remains
linear.  Hereafter, we may assume that $\cal K$ is constructed with
the above stipulated order.

For every module $v$ in the main cycle $\cal C$ of $\cal K$, we can
easily calculated \lint{v} and \rint{v}, which indicate the position
of $v$ in $\cal C$.  We can also assign positions to modules $s\in
SP(Q)$ by setting $\lint{s} := \max_{v\in N_Q(s)} \lint{v}$ and
$\rint{s} := \min_{v\in N_Q(s)} \rint{v}$; thus for every module $v\in
N_Q(s)$ it holds that $[\lint{s}, \rint{s}]\subseteq [\lint{v},
\rint{v}]$.  In total, we have $2|Q|$ indices for the $|Q|$ modules,
which can be calculated in linear time.  Let $s$ be the very first
module that is put into $SP(Q)$ during the construction of the clique
decomposition $\cal K$.  The order we use to handle simplicial modules
ensures us that there is no other module $s'\in SP(Q)$ satisfying
$[\lint{s'}, \rint{s'}]\subset [\lint{s}, \rint{s}]$.  We now describe
in Fig.~\ref{fig:find-minimal-caw} the procedure that, starting from
$s$, finds a minimal frame or a small caw.  For notational
convenience, we may assume that $0\not\in [\lint{s},\rint{s}]$; this
is achievable by circularly renumbering indices of bags in $\cal C$.
\begin{lemma}\label{lem:find-minimal-caw}
  The procedure in Fig.~\ref{fig:find-minimal-caw} finds in $O(||G||)$
  time a small caw or a minimal frame ($s: c_1,c_2:l,l_b; r_b,r$).  This
  frame either is in a $\dag$ or $\ddag$ whose base consists of three
  modules, or has the following properties:
  \begin{enumerate}[(\tt C1)]
  \item\label{can:1} no module in $SP(Q)\setminus \{s\}$ is adjacent
    to both $c_1$ and $c_2$;
  \item\label{can:2} $N_Q[c_1]\cap N_Q[c_2]\subseteq N_Q[v]$ for every $v\in
    N_Q(s)$;
  \item\label{can:3} $N_Q[v]\subseteq N_Q[c_1]\cap N_Q[c_2]$ for every $v\in
    IN(F)$;
  \item\label{can:4} $s$ is the only module of the frame in $SP(Q)$;
    and
  \item\label{can:5} $N_Q[l_b]\cap N_Q[c_2]\subset N_Q[c_1]$ and $N_Q[r_b]\cap
    N_Q[c_1]\subset N_Q[c_2]$.
  \end{enumerate}
\end{lemma}
\begin{figure*}[t]
\setbox4=\vbox{\hsize28pc \noindent\strut
\begin{quote}
  \vspace*{-5mm} \small

  1 \hspace*{2ex} \parbox[t]{0.85\linewidth}{ find $c_1,c_2\in N_Q(s)$
    such that $\rint{c_1} = \rint{s}$ and $\lint{c_2} = \lint{s}$;
    \\
    {\bf if} $\lint{c_1} = \lint{c_2}$ {\bf then} $c_2 = c_1$; }
  \\
  2 \hspace*{2ex} \parbox[t]{0.85\linewidth}{ find $l_b$ from
    $K_{\lint{c_2}-1}\setminus N_Q(s)$ such that \rint{l_b} attains the
    maximum value in it;
    \\
    find $r_b$ from $K_{\rint{c_1}+1}\setminus N_Q(s)$ such that
    \lint{r_b} attains the minimum value in it;}
  \\
  3 \hspace*{2ex} {\bf if} $\rint{l_b} \ge \rint{s}$ or $\lint{r_b} \le
  \lint{s}$ {\bf then}
  \\
  3.1 \hspace*{4ex} {\bf if} $\rint{l_b}> \rint{s}$ {\bf then}
  \\
  \hspace*{10ex} choose $x$ from $(K_{\lint{l_b}-1}\cap
  K_{\lint{l_b}})\setminus N_Q(s)$ and $x'$ from
  $K_{\lint{l_b}-1}\setminus K_{\lint{l_b}}$;
  \\
  \hspace*{10ex} choose $y$ from $(K_{\rint{l_b}+1}\cap
  K_{\rint{l_b}})\setminus N_Q(s)$ and $y'$ from
  $K_{\rint{l_b}+1}\setminus K_{\rint{l_b}}$;
  \\
  \hspace*{10ex} {\bf return} long claw $\{s, c_1, l_b, x, x', y, y'\}$
  or sun $\{s, c_1, c_2, l_b, x, y\}$;
  \\
  3.2 \hspace*{4ex} {\bf if} $\rint{l_b}= \rint{s}$ {\bf then}
  \\
  \hspace*{10ex} choose $x$ from $K_{\lint{s}-1}\setminus
  K_{\lint{s}-1}$;
  \\
  \hspace*{10ex} choose $y$ from $(K_{\rint{s}+1}\cap
  K_{\rint{s}})\setminus N_Q(s)$ and $y'$ from $K_{\rint{s}+1}\setminus
  K_{\rint{s}}$;
  \\
  \hspace*{10ex} {\bf return} net $\{s, c_1, x, l_b, y, y'\}$ or rising
  sun $\{s, c_1, c_2, x, l_b, y, y'\}$;
  \\
  3.3 \hspace*{4ex} {\bf else} {\sf symmetric as step 4.2};
  \comment{i.e., $\lint{r_b} = \lint{s}$.}
  \\
  4 \hspace*{2ex} \parbox[t]{0.85\linewidth}{ find $l$ to be any
    module with $\rint{l} = \lint{l_b}$;
    \\
    find $r$ to be any module with $\lint{r} = \rint{r_b}$;}
  \\
  5 \hspace*{2ex} \parbox[t]{0.85\linewidth}{ {\bf if} $\rint{l} <
    \lint{c_1}$, then set $c_2 = c_1$;
    \\
    {\bf if} $\lint{r} > \rint{c_2}$, then set $c_1 = c_2$;}
  \\
  6 \hspace*{2ex} {\bf if} $\rint{l_b} \ge \lint{r_b}$ {\bf then return}
  $\{s,c_1,c_2,l_b,r_b,l,r\}$ as a net or rising sun; \comment{$c_1=c_2$
    or not.}
  \\
  7 \hspace*{2ex} {\bf else return} ($s:c_1,c_2:l, l_b; r_b, r$) as a
  minimal frame.

\end{quote} \vspace*{-6mm} \strut} $$\boxit{\box4}$$
\vspace*{-9mm}
\caption{Finding a minimal frame (Lem.~\ref{lem:find-minimal-caw}).}
\label{fig:find-minimal-caw}
\end{figure*}
\begin{proof}
  Step 1 finds center(s) $c_1,c_2$, whose existence is ensured by the
  definition of $\lint{s}$ and $\rint{s}$.  It is worth noting that
  $c_1$ and $c_2$ might or might not refer to the same module; neither
  case will be assumed by the rest of the procedure.  Step 2 finds
  $l_b, r_b$, whose existence is ensured by
  Cor.~\ref{lem:hole=no-caw}.  If $(K_{\lint{c_2}-1}\cap
  K_{\lint{c_2}}) \setminus N_Q(s) = \emptyset$, then $N_Q[s]$ can be
  put between bags $K_{\lint{c_2}-1}$ and $K_{\lint{c_2}}$,
  contradicting Cor.~\ref{lem:hole=no-caw}; a symmetric argument
  applies to $r_b$.  By their selection and noting that there is no
  short hole, we can conclude $\rint{l_b} \ge {\lint{c_2}}$ and then
  $l_b\sim c_2$; likewise, $c_1\sim r_b$.  In step 3, the existence of
  $x$ and $y$ can also be argued by Cor.~\ref{lem:hole=no-caw}, and
  the existence of $x',y'$ is by the definition of clique
  decomposition.  The condition $\lint{l_b}< \lint{s}$ or
  $\rint{s}<\rint{r_b}$ then implies a path of length $4$ that
  connects $K_{\lint{s}-1}$ and $K_{\rint{s}+1}$ and avoids $N_Q(s)$,
  which, together with $\{s,c_1,c_2\}$, contains a small caw.  Step 4
  finds base terminals $l,r$, whose existence is ensured by the fact
  that $K_{\lint{l_b}}$ and $K_{\rint{r_b}}$ are maximal cliques, and
  hence $K_{\lint{l_b}}\setminus K_{\lint{l_b}+1} \ne\emptyset$ and
  $K_{\rint{r_b}}\setminus K_{\lint{r_b} - 1} \ne\emptyset$.  

  Step 5 is clear.  To verify steps 6 and 7, we need to check first
  the pairwise adjacency relation.  By construction, both $l l_b c_2
  r_b r$ and $l l_b c_2 r_b r$ are paths of length $4$.  Since the
  length of a shortest hole is at least $6$
  (Lem.~\ref{lem:find-shortest-hole}), $\rint{l}\le \lint{c_2} - 1$
  implies that $l\not\sim c_2,r_b,r$.  Likewise, we conclude
  $r\not\sim c_1,l_b$.  Therefore, if $l_b\sim r_b$, then
  $\{s,c_1,c_2,l,l_b, r_b,r\}$ is a net (when $c_1 = c_2$) or a rising
  sun (when $c_1 \ne c_2$).  This justifies step~6.  Otherwise,
  $l_b\not\sim r_b$, then ($s:c_1,c_2:l, l_b; r_b, r$) is a frame,
  witnessed by any \stpath{l_b}{r_b} through $\bigcup_{\rint{l_b}<
    \ell <\lint{r_b}}K_\ell \setminus N_Q(s)$ (noting that
  $(K_{\ell}\cap K_{\ell+1}) \setminus N_Q(s)$ is nonempty for every
  $\ell$ in this range).

  If we have found a small caw in the previous steps, then we are
  done, and hence in the rest of the proof we assume that the output
  is a frame $F$.  We argue first that $F$ is minimal.  Suppose for
  contradiction, that there exists a frame $F' = $ ($s': c'_1, c'_2:
  l', l'_b; r'_b, r'$) such that $[F']\subset [F]$.  Clearly,
  $[\lint{s'},\rint{s'}] \subseteq [\lint{l'_b},\rint{r'_b}] \subset
  [\lint{l_b},\rint{r_b}] = [\lint{s},\rint{s}]$, and thus from our
  selection of $s$, we can infer that $s'\not\in SP(Q)$.  But then at
  least one of $l'$ and $r'$ must be simplicial and belong to $SP(Q)$,
  which should be processed before $s$, contradicting our selection of
  $s$.  We have thus concluded the minimality of $F$, and then proceed
  to the second assertion, for which we check the conditions one by
  one.

  \myref{1}.  We first check whether there exists some module $s'\in
  SP(Q)\setminus \{s\}$ that is adjacency to both $c_1$ and $c_2$.  If
  not, \myref{1} is satisfied.  Otherwise, by definition,
  $\lint{s'}\ge \lint{c_2} = \lint{s}$ and $\rint{s'}\le
  \rint{c_1}=\rint{s}$.  Neither inequality can be strict, as
  otherwise $s'$ should have been put into $SP(Q)$ before $s$, which
  is impossible.  Noting that $N_Q(s')\ne N_Q(s)$ as $Q$ is prime,
  there must be some module $x$ that is adjacent to precisely one of
  $s$ and $s'$.  If $x\not\sim s$, then we can return ($s:
  c_1,c_2,l,l_b,x,r_b,r$) as the caw: here note that $x$ is adjacent to
  neither $l$ nor $r$, as otherwise we have entered step 3.  The
  situation is symmetric when $x\sim s'$.

  \myref{2}.  Let $x$ be any module in $N_Q[c_1]\cap N_Q[c_2]$
  different from $s$.  As a result of \myref{1}, $x\not\in SP(Q)$, and
  is contained in $\bigcup_{\lint{c_2}\le\ell\le\rint{c_1}}K_{\ell}$.
  By definition, $\lint{v}\le\lint{c_2}$ and $\rint{v}\ge\rint{c_1}$,
  implying $v\sim x$.

  \myref{3} and \myref{4} are immediate from the selection of
  the center(s) $c_1,c_2$ and the terminal $s$, respectively.

  \myref{5}.  We check whether $N_Q[l_b]\cap N_Q[c_2]\setminus
  N_Q[c_1]$ is empty, and if not, find a module $x$ from it.  From
  $x\not\sim c_1$ we can infer $s\not\sim x$.  Since $Q$ does not
  contain $4$-hole, we must have $x\sim l$, and then we return sun
  $\{s, l, l_b, c_1, c_2, x\}$.  We then do a symmetric check to
  $N_Q[l_b]\cap N_Q[c_2]\setminus N_Q[c_1]$, and fail to detect a
  subgraph in $\cF_{LI}$ only when \myref{5} is satisfied.

  The runtime is clearly $O(||G||)$.  This concludes the proof.
\end{proof}

Recall that we have renumbered indices of bags in $\cal C$ to exclude
$K_0$ from the detected minimal frame; as a result,
$\lint{l}\le\lint{l_b}\le\rint{l}<\lint{c_2}\le\rint{l_b} <\lint{r_b}\le
\rint{c_1}< \lint{r} \le \rint{r_b}\le \rint{r}$.
We point out that the frame of a caw with the minimum number of
modules is unnecessarily minimal; see, e.g.,
Fig.~\ref{fig:minimal-frame}.  In the following subsections, $F$ will
be the minimal frame found by Lem.~\ref{lem:find-minimal-caw}.
\begin{SCfigure}[][h]
  \centering  \includegraphics{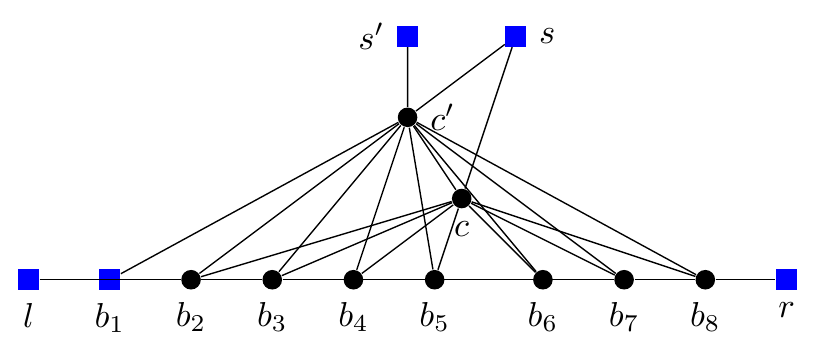} 
  \caption{The minimal frame is ($s: c,c: b_1, b_2; b_8, r$),
    witnessed by caw ($s: c,c: b_1, b_2 \cdots b_8, r$), but caw
    ($s':c',c':l, b_1 b_2 c b_8, r$) has less vertices.}
  \label{fig:minimal-frame}
\end{SCfigure}

\subsection{Breaking large caws by vertex deletions}
Recall that if a module $v$ is not simplicial in $Q$, then it must be
a clique module; it either is completely deleted, or remains intact in
a \miib.

\begin{lemma}\label{lem:interval-deletion-main}
  Let $G[U^*]$ be a \miib\ of $G$ such that $U^*$ intersects every
  module in $F$.  Define $S_\imath := K_{\imath}\cap
  K_{\imath+1}\setminus N_Q(s)$ for $\rint{l_b}\le\imath<\lint{r_b}$,
  and let $\ell$ be an index such that $\#S_{\ell}$ is minimum among
  all of them.  There exists a \miib\ $ G[U]$ disjoint from all
  modules in $S_{\ell}$.
\end{lemma}
\begin{proof}
  Let $V^*_- := V(G)\setminus U^*$, and by abuse of notation, we use
  $Q - V^*_-$ to denote the subgraph of $Q$ induced by those maximal
  strong modules that are not fully contained in $V^*_-$.  Clearly, $Q
  - V^*_-$ is a quotient graph of $G - V^*_-$, hence an interval
  graph.  A module might be partially deleted; by
  Thm.~\ref{thm:modules-induced-subgraph+}, such a module cannot be a
  clique.  Replacing all non-simplicial modules of $IN(F)$ in $V^*_-$
  by $S_{\ell}$ gives a new set $V_-$; we show that $G - V_-$ is the
  claimed \miib, i.e. $U = V(G)\setminus V_-$.  In every
  \stpath{l_b}{r_b} through $IN(F)$ in $Q$, at least one module is
  completely contained in $V^*_-$; otherwise, this path, together with
  $F$, makes a caw in $Q - V^*_-$, which is impossible.  Therefore,
  $V^*_-$ has to contain all modules in a minimal \stsep{l_b}{r_b} in
  $IN(F)$, i.e., $S_\imath$ for some $\rint{l_b}\le\imath<\lint{r_b}$.
  Noting that all different modules in $V^*_-$ and $V_-$ are cliques
  and fully contained in them, by the selection of $S_\ell$, we can
  conclude $|V_-|\le |V^*_-|$.

  To show that $G[U]$ is indeed an interval graph, it suffices to
  verify that $Q - V_-$ is an interval graph: given a normalized
  interval model for $Q - V_-$, we can build a (not necessarily
  normalized) interval model for $G[U]$ as follows.  Note that every
  module in $Q - V_-$ induces an interval subgraph of $G[U]$; it is a
  clique if and only the module is a clique in $G$ and then it must be
  disjoint from $V_-$.  For a clique module in $Q - V_-$, we use the
  same interval for all vertices in this module.  For a non-clique
  module $v\in Q - V_-$, let $M_v$ be the set of vertices in the
  module not in $V_-$.  We find some point $\rho$ that is contained in
  all intervals for $N_{Q - V_-}(v)$ but not an endpoint itself, and
  project an interval model for $G[M_v]$ to $[\rho - \epsilon, \rho +
  \epsilon]$.  For the existence of $\rho$, note that $N_{Q - V_-}(v)$
  is, if nonempty, a clique of $Q$ and every module in it is a clique
  module of $G$.  It is easy to verify that this gives an interval
  model for $G - V_-$.  Thus, the rest of the proof is devoted to
  constructing an interval model for $Q - V_-$.

  Let ${\cal I}^* = \{I^*_v: v\in V(Q - V^*_-)\}$, where $I^*_v =
  [\lpp{^*}{v}, \rpp{^*}{v}]$, be a normalized interval model for the
  interval graph $Q - V^*_-$.  Let ${\cal A} = \{A_v: v\in
  V(Q)\setminus SP(Q)\}$, where $A_v = [\lpp{^c}{v},\rpp{^c}{v}]$, be
  the circular-arc model for $Q - SP(Q)$ obtained using
  Prop.~\ref{lem:clique-decomposition-arc-model}.  We make the
  following assumptions on ${\cal I}^*$ and ${\cal A}$.
  \begin{enumerate}[(1)]
  \item The interval for ${l_b}$ is to the left of that for ${r_b}$ in
    ${\cal I}^*$, i.e., $\rpp{^*}{l_b}<\lpp{^*}{r_b}$.  Since
    $l_b\not\sim r_b$, this assumption is clear.  As a consequence of
    this assumption, when $c_1\ne c_2$, all endpoints of intervals for
    $c_1$ and $c_2$ are ordered as follows (see
    Fig.~\ref{fig:fixed-frame}):
    $$\lpp{^*}{c_1}< \rpp{^*}{l}< \lpp{^*}{c_2} <\rpp{^*}{c_1}
    <\lpp{^*}{r}< \rpp{^*}{c_2}.$$ 
  \item The models ${\cal I}^*$ and $\cal A$ are aligned such that
    $[\lpp{^*}{c_2}= \lpp{^c}{c_2}]$ and $[\rpp{^*}{c_1}=
    \rpp{^c}{c_1}]$.  We will use the same values for \lp{c_2} and
    \rp{c_1} respectively in $\cal I$.  As a result, we can refer to
    \lp{c_2} and \rp{c_1} without specifying which model.
  \item For every $v\in N_{Q-V^*_-}(s)$, the interval $I^*_v$ fully
    contains $[\lp{c_2}, \rp{c_1}]$.  By (\myref{2}), $N_Q[c_1]\cap
    N_Q[c_2]\subseteq N_Q[v]$, and thus we can always extend $I^*_v$
    to cover $[\lp{c_2}, \rp{c_1}]$ without breaking the model.
  \item For every $v\in IN(F)$ that remains in $Q - V^*_-$, the
    interval $I^*_v$ is fully contained in $[\lp{c_2}, \rp{c_1}]$.  By
    (\myref{3}), $N_Q[v]\subseteq N_Q[c_1]\cap N_Q[c_2]$, and thus we
    can always shrink $I^*_v$ to satisfy this assumption.
  \end{enumerate}
  Note that assumptions (3) and (4) do not conflict with each other.

  Let $\rho$ be a point in the circular-arc model such that $\rho \in
  A_v$ if and only if $v\in K_{\ell}\cap K_{\ell+1}$.  The existence
  of $\rho$ is clear from the construction of $\cal A$, e.g., $\ell
  +1/2$ in the original model given by
  Prop.~\ref{lem:clique-decomposition-arc-model}.  The center(s) $c_1,
  c_2$ belong to both $K_{\ell}$ and $K_{\ell+1}$, and thus $\rho\in
  [\lp{c_2}, \rp{c_1}]$.  We are now ready to present the new interval
  model ${\cal I}$
  as follows.  If a module $v$ remains in $Q - V^*_-$ and $I^*_v$ is
  disjoint from or fully contains $[\lp{c_2}, \rp{c_1}]$, then $I_v =
  I^*_v$; both $c_1$ and $c_2$ are in this case.  The interval $I_s =
  [\rho-\epsilon, \rho+\epsilon]$.  For every remaining module $v$ in
  $IN(F)$, we use $I_v = A_v$.  If a module $v$ belongs to none of
  above, i.e., $I^*_v$ intersects only part of $[\lp{c_2}, \rp{c_1}]$,
  then we set
  \[
  \lp{v} =
  \begin{cases}
    \lpp{^*}{v} & \text{if } \lpp{^*}{v} < \lp{c_2},
    \\
    \lpp{^c}{v} & \text{otherwise,}
  \end{cases}
  \quad \text{ and }
  \rp{v} =
  \begin{cases}
    \rpp{^*}{v} & \text{if } \rpp{^*}{v} > \rp{c_1},
    \\
    \rpp{^c}{v} & \text{otherwise.}
  \end{cases}
  \]

  Since both the given interval model $\cal I^*$ and circular-arc
  model $\cal A$ represent only part of $V(Q)$, we need to verify that
  $I_v$ is well-defined for every module $v$ in $Q- V_-$.  That is,
  the construction does not use any nonexistent interval $I^*_v$
  (i.e., a non-simplicial module in $IN(F)$) or nonexistent arc $A_v$
  (i.e., in $SP(Q)$).  For a module $v\in IN(F)$, we used $I_v = A_v$.
  By construction, a module $v\in SP(Q)$ remains in $Q - V_-$ if and
  only if it is in $Q- V^*_-$; for such a module $v$ different from
  $s$, we use $I_v = I^*_v$ as $I^*_v$ is disjoint from $[\lp{c_2},
  \rp{c_1}]$ (\myref{1}).

  We now verify that the interval model ${\cal I}$ represents $Q -
  V_-$.  Note that by construction, if an interval $I_v$ is not a
  subset of $[\lp{c_2}, \rp{c_1}]$, then $I_v\setminus [\lp{c_2},
  \rp{c_1}] = I^*_v\setminus [\lp{c_2}, \rp{c_1}]$.  First, given a
  pair of intersecting intervals $I_u$ and $I_v$ in $\cal I$, we show
  $u\sim v$; note that all subgraphs mentioned here are induced
  subgraphs of $G$.  Let $\rho$ be a point in $I_u\cap I_v$.  If
  $\rho\not\in [\lp{c_2}, \rp{c_1}]$, then $\rho$ is in both $I^*_u$
  and $I^*_v$ as well.  If one of $u$ and $v$ is $s$, then the other
  is in $N_Q(s)$.  Otherwise, $\rho$ is in both $A_u$ and $A_v$.  In
  each case, we can conclude $u\sim v$.  Second, given a pair of
  adjacent modules $u,v$ in $Q - V_-$, we show $I_u\cap
  I_v\ne\emptyset$.  If neither of them is in $IN(F)$, then intervals
  $I^*_u$ and $I^*_v$ are defined and intersect.  If $I^*_u\cap I^*_v$
  contains a point $\rho\not\in [\lp{c_2}, \rp{c_1}]$, then $\rho$ is
  also contained in both $I_u$ and $I_v$.  Otherwise, both $I^*_u$ and
  $I^*_v$ intersects $[\lp{c_2}, \rp{c_1}]$; by (\myref{1}),
  $u,v\not\in SP(Q)$.  Thus, $A_u$ and $A_v$ are defined and intersect
  at some point $\rho\in [\lp{c_2}, \rp{c_1}]$, which is contained in
  both $I_u$ and $I_v$.  Assume now that, without loss of generality,
  $u\in IN(F)$.  Then $I_u = A_u\subseteq [\lp{c_2}, \rp{c_1}]$.  By
  (\myref{1}) and (\myref{3}), $v\not\in SP(Q)$, and thus $I_u\cap I_v
  = A_u\cap A_v$.  In summary, a pair of modules $u,v$ in $Q - V_-$ is
  adjacent if and only if $I_u\cap I_v\ne\emptyset$; thus, $\cal I$ is
  an interval model for $Q - V_-$.  This implies that $G - V_-$ is an
  interval graph and concludes the proof.
\end{proof}
Therefore, there always exists a minimum solution that deletes either
one module from the frame, or the $S_\ell$, which can found in linear
time.  It is worth noting that Lem.~\ref{lem:interval-deletion-main}
is stronger than a similar result presented in our previous work
\cite{cao-14-interval-deletion}, which is argued using the
characterization of forbidden induced subgraphs.  The improvement is
not only on the constant, decreased from 10 to 8, but more
importantly, the new proof does not require the graph $G$ to be
chordal.  This relaxation permits us to attend to large caws in a
non-chordal graph.

\subsection{Breaking large caws by edge deletions}
We consider all caws with the given minimal frame $F$ in $Q$.  They
differ only in the bases, and one has the smallest size if and only if
its base $B$ is a shortest \stpath{l_b}{r_b} through $IN(F)$, which
can be found as follows.  Let $b_1 = l_b$.  For each $i=1,\dots$, we
find the next vertex $b_{i+1}$ from $K_{\rint{b_i}}\setminus N_Q(s)$
such that $\rint{b_{i+1}}$ reaches the rightest position.  This
process stops at the first $i$ satisfying $r_b\sim b_i$, and then we
set $d = i+1$ and $b_d = r_b$; note that $b_{i}\not\sim r$ as
otherwise $F$ is not minimal.

The three (unnecessarily induced) paths $s c_1 b_1 l$, $s c_2 b_d r$,
and $l Br$ witness the at $\{s,l,r\}$; therefore, to make a spanning
interval subgraph, at least one edge in these paths has to be deleted.
If the base $B$ found above has a bounded length, then the total
number of edges in the caw is bounded as well, and hence we can try
deleting every of them.  Therefore, we are mainly concerned with the
other case, for which we focus on $IN(F)$.  Consider a bag $K_\imath$
with $\imath \in [\rint{l_b},\lint{r_b}]$ and a nontrivial partition
($X,Y$) of of $K_\imath\setminus N_Q(s)$.  We use $E_{\imath, X,Y}$ to
denote the set of edges to be deleted if we cut in between, i.e.,
\[
(X\times Y) \cup \{v u: v\in X,\,\lint{u} \in [\imath, \rint{v}] \}
\cup \{v u: v\in Y,\,\rint{u}\in [\lint{v}, \imath]\}.
\]
It is a subset of $E(Q)$, and we measure its size by the number of
inter-module edges of $G$ corresponding to the edges in $Q$, i.e.,
$\sum_{u v\in E_{\imath, X,Y}} \#u \cdot \#v$.  We point out that the
bound $8k + 16$ in the following lemma is not tight, and the current
value is chosen for the simplicity of the presentation instead of
efficiency.  The constants can be improved with more careful analysis;
however, how to achieve a sub-linear bound does not seem clear to us
at this moment.
\begin{lemma}\label{lem:interval-edge-deletion-main}
  Let $B$ be a shortest \stpath{l_b}{r_b} in $G[IN(F)\cup\{l_b,r_b\}]$ and
  its length be more than $8k+16$.  Let $\ell \in
  [\rint{l_b},\lint{r_b}]$ and ($X,Y$) be a partition of
  $K_\ell\setminus N_Q(s)$ such that $E_{\ell, X,Y}$ has the minimum
  size among all bags and partitions.  If there is a \misb\
  $\underline G^* := G - E^*_-$ of $G$ such that
  \begin{enumerate}
  \item $|E^*_-| \le k$,
  \item all maximal strong modules of $G$ remains modules of
    $\underline G^*$, which define the quotient graph $\underline
    Q^*$, and
  \item $\underline Q^*$ contains $\{s c_1, s c_2, l_b c_1, r_b c_2\}$
    as well as every $b_i b_{i+1}$ with $i\in \{0, \dots, 4k+7\}\cup
    \{d-4k-7,\dots, d\}$,
  \end{enumerate}
  then there is a \misb\ $\underline G : = G - E_-$ such that
  \begin{enumerate}
  \item all maximal strong modules of $G$ remains modules of
    $\underline G$, which define the quotient graph $\underline
    Q$, and
  \item none of $E_{\ell, X,Y}$ is present in $\underline Q$.
  \end{enumerate}
\end{lemma}
\begin{proof}
  We use $P_L$ and $P_R$ to denote the paths $b_0\cdots b_{4k+8}$ and
  $b_{d-4k-7}\cdots b_{d+1}$ in $Q$.  Let ${\cal I}^* = \{I^*_v: v\in
  V(Q)\}$, where $I^*_v = [\lpp{^*}{v}, \rpp{^*}{v}]$, be a normalized
  interval model for $\underline Q^*$.  By assumption, apart from the
  frame $F$, both paths $P_L$ and $P_R$ also remain intact in
  $\underline Q^*$; they will together set the basic shape of
  intervals in ${\cal I}^*$.  Let ${\cal A} = \{A_v: v\in
  V(Q)\setminus SP(Q)\}$, where $A_v = [\lpp{^c}{v}, \rpp{^c}{v}]$, be
  the circular-arc model for $Q - SP(Q)$ obtained using
  Prop.~\ref{lem:clique-decomposition-arc-model}.  We make following
  assumptions on $\cal I^*$ and $\cal A$.
  \begin{enumerate}[(1)]
  \item The interval for ${l_b}$ is to the left of that for ${r_b}$ in
    ${\cal I}^*$, i.e., $\rpp{^*}{l_b}<\lpp{^*}{r_b}$.  Since
    $l_b\not\sim r_b$ in $\underline Q^*$, this assumption is clear.
    Moreover, as $l\in N_{\underline Q^*}(l_b)\setminus N_{\underline
      Q^*}(c_2)$, it follows that $\lpp{^*}{l_b} < \rpp{^*}{l} <
    \lpp{^*}{c_2}$.
  \item If $c_1\ne c_2$ then $\lpp{^*}{c_1} < \lpp{^*}{c_2} <
    \rpp{^*}{c_1} < \rpp{^*}{c_2}$.  To satisfy the first inequality,
    we can always extend $I^*_{c_1}$ to the left by setting
    $\lpp{^*}{c_1} = \lpp{^*}{c_2} -\epsilon$ to satisfy this
    assumption.  The safeness of this extension can be argued using
    (\myref{5}): since $[\lpp{^*}{c_2}, \lpp{^*}{c_1}]$ is a subset of
    $I^*_{l_b}$ and by the selection of $\epsilon$, any interval
    $I^*_v$ that intersects $[\lpp{^*}{c_2}-\epsilon, \lpp{^*}{c_1}]$
    only if $v\sim c_1$ in ${\underline Q^*}$.  A symmetric extension
    works for the last inequality.  In other words, $\lpp{^*}{c_1} \le
    \lpp{^*}{c_2} < \rpp{^*}{c_1} \le \rpp{^*}{c_2}$ always holds
    true.
  \item For every $v\in N_{\underline Q^*}(s)$, the interval $I^*_v$
    either fully contains $[\lpp{^*}{c_2}, \rpp{^*}{c_1}]$, or is
    disjoint from it.  By (\myref{2}), $N_{\underline Q^*}[c_1]\cap
    N_{\underline Q^*}[c_2]\subseteq N_{\underline Q^*}[v]$, and thus
    we can always extend $I^*_v$ to one or both directions to cover
    $[\lpp{^*}{c_2}, \rpp{^*}{c_1}]$ without breaking the model.
  \item For every $v\in IN(F)$, the interval $I^*_v$ is either fully
    contained in, or disjoint from $[\lpp{^*}{c_2}, \rpp{^*}{c_1}]$.
    By (\myref{3}), $N_{\underline Q^*}[v]\subseteq N_{\underline
      Q^*}[c_1]\cap N_{\underline Q^*}[c_2]$, and thus if $I^*_v$
    intersects $[\lpp{^*}{c_2}, \rpp{^*}{c_1}]$ but is not fully
    contained in it, then we can always shrink $I^*_v$ to satisfy this
    assumption.
  \end{enumerate}

  From these assumptions two properties of $\cal I^*$ can be derived.
  \begin{claim}\label{lem:break-segment-2}
    The interval for any module in $ N_Q(s)$ fully contains
    $[\lpp{^*}{c_2}, \rpp{^*}{c_1}]$.
  \end{claim}
  \begin{proof}
    Suppose, for contradiction, that there is a module $v\in N_Q(s)$
    such that $I^*_v$ does not fully contains $[\lpp{^*}{c_2},
    \rpp{^*}{c_1}]$.  Then by assumption (3), $I^*_v$ is disjoint from
    it.  For every module $x\in IN(F)$, if $I^*_x$ intersects $I^*_v$,
    then by assumption (4), $I^*_x$ is disjoint from $[\lpp{^*}{c_2},
    \rpp{^*}{c_1}]$ as well.  In other words, $x$ is nonadjacent to at
    least one of $c_1$ and $c_2$ in $\underline Q^*$.  Since $x$ is
    adjacent to both $c_1,c_2$, and $v$ in $Q$, we need to delete
    either $\{x c_1, x c_2\}$ or $x v$.  In other words, each module
    in $IN(F)$ is incident to at least one deleted edge, which means
    that the total number of deleted edges is at least $|IN(F)| > k$.
    This contradiction proves this claim.
    \renewcommand{\qedsymbol}{$\lrcorner$}
  \end{proof}
  
  Recall that $I^*_s \subset [\lpp{^*}{c_2}, \rpp{^*}{c_1}]$; as a
  result, $N_Q(s) = N_{\underline Q^*}(s)$.
  \begin{claim}\label{lem:break-segment-2}
    There exist $p$ and $q$ with $1\le p\le 4k + 4$ and $d - 4k - 3\le
    q\le d$ such that $\underline Q^*$ retains all edges incident to
    $N_Q(b_p)\setminus N_Q(s)$ and $N_Q(b_q)\setminus N_Q(s)$.
  \end{claim}
  \begin{proof}
    By symmetry, it suffices to verify the existence of $p$.  For each
    $i$ with $1\le i\le 4k+4$, the sets $N_Q(b_i)\setminus N_Q(s)$ and
    $N_Q(b_{i+4})\setminus N_Q(s)$ are nonadjacent; otherwise, there
    is a shorter base than $B$, contradicting our selection of $B$.
    In other words, an edge cannot be incident to both
    $N_Q(b_i)\setminus N_Q(s)$ and $N_Q(b_{i+4})\setminus N_Q(s)$.
    Suppose that for each $i$ with $1\le i\le 4k+4$, some edge
    incident to $N_Q(b_i)\setminus N_Q(s)$ is delete, then the total
    number of deleted edges is at least $(4k + 4)/4> k$, a
    contradiction.
    \renewcommand{\qedsymbol}{$\lrcorner$}
  \end{proof}

  As a result of Claim~\ref{lem:break-segment-2} and assumption (4),
  intervals for both ${b_p}$ and ${b_q}$ are fully contained in
  $[\lpp{^*}{c_2}, \rpp{^*}{c_1}]$.  By assumption (1), and noting
  that paths $P_L$ and $P_R$ remain in $\underline Q^*$, the interval
  for $b_p$ must be to the left of that for $b_q$, i.e.,
  $\rpp{^*}{b_p}< \lpp{^*}{b_q}$.  Let $U_I := \bigcup_{\lint{b_p}\le
    \imath\le \rint{b_q}}K_{\imath}\setminus N_Q(s)$.  The following
  claim characterizes other modules $v$ satisfying $N_{\underline
    Q^*}(v) = N_Q(s)$.
  \begin{claim}\label{lem:break-segment-3}
    Let $Z$ be the subset of modules in $V(Q)\setminus (U_I\cup
    N_Q(s))$ whose intervals intersect $[\rpp{^*}{b_p},
    \lpp{^*}{b_q}]$.  For each $v\in Z$, the interval $I^*_v$ is fully
    contained in $[\rpp{^*}{b_p}, \lpp{^*}{b_q}]$, and $N_{\underline
      Q^*}(Z) = N_Q(s)$.
  \end{claim}
  \begin{proof}
    By \myref{1}, $v\not\in SP(Q)$, and then a neighbor of $b_p$ or
    $b_q$ must belong to $U_I\cup N_Q(s)$; thus, $v$ is adjacent to
    neither of them, and the first assertion follows.  The direction
    of $N_Q(s) \subseteq N_{\underline Q^*}(Z)$ of the second
    assertion follows from Claim~\ref{lem:break-segment-2}.  Suppose,
    for the contradiction of the other direction, that there is a pair
    of modules $v\in Z$ and $u\in U_I$ such that $v\sim u$ in
    $\underline Q^*$.  Such a module $u$ must be adjacent to $b_p$ or
    $b_q$, and remains so in $\underline Q^*$.  However, since $P_L$
    and $P_R$ are retained, $I^*_v$ must lie between the intervals of
    $b_{4k+8}$ and $b_{d-4k-7}$.  This is impossible, and hence the
    second assertion must hole true.
    \renewcommand{\qedsymbol}{$\lrcorner$}
  \end{proof}
  In particular, $s$ belong to the set specified in
  Claim~\ref{lem:break-segment-3}.  Let $\xi$ be the middle point
  between $\rpp{^*}{b_p}$ and $\lpp{^*}{b_q}$, i.e., $\xi =
  (\rpp{^*}{b_p}+ \lpp{^*}{b_q})/2$.  The new interval model ${\cal I}
  = \{I_v: v\in V(Q)\}$ is constructed as follows.  Consider first a
  module $v\not\in U_I$, which is either in $N_Q[s]$ or nonadjacent to
  $b_p$ or $b_q$.  If $I^*_v$ intersects but not fully contains
  $[\rpp{^*}{b_p}, \lpp{^*}{b_q}]$, then $I^*_v$ must be fully
  contained in $[\rpp{^*}{b_p}, \lpp{^*}{b_q}]$, and we set $I_v$ by
  projecting $I^*_v$ from $[\rpp{^*}{b_p}, \lpp{^*}{b_q}]$ to
  $[\xi-\epsilon, \xi+\epsilon]$.  Otherwise, $I^*_v$ either fully
  contains or is disjoint from $[\rpp{^*}{b_p}, \lpp{^*}{b_q}]$, then
  we set $I_v = I^*_v$.  Note that $s$ and $N_Q(s)$ are in the first
  and second categories respectively.  Consider now module $v\in U_I$.
  Note that since $U_I\subset IN(F)\subset V(Q)\setminus SP(Q)$, the
  arc $A_v$ is defined.  If $\lpp{^*}{v}$ (resp., $\rpp{^*}{v}$) is
  not in $[\rpp{^*}{b_p}, \lpp{^*}{b_q}]$, then $\lp{v} = \lpp{^*}{v}$
  (resp., $\rp{v} = \rpp{^*}{v}$).  Otherwise, $\lp{v}$ (resp.,
  $\rp{v}$) is obtained by projecting $\lpp{^c}{v}$ (resp.,
  $\rpp{^c}{v}$) from $[\rpp{^c}{b_p}, \lpp{^c}{b_q}]$ to
  \[
  \begin{cases}
    [\rpp{^*}{b_p}, \xi-\epsilon] & \text{if } \rint{v}<\ell \text{ or }
    v\in X,
    \\
    [\xi+\epsilon, \lpp{^*}{b_q}]) & \text{if } \lint{v}>\ell \text{ or
    } v\in Y.
  \end{cases}
  \]
  See Fig.~\ref{fig:interval-edge-deletion-main} for an illustration
  of this construction.  The interval graph represented by the
  interval model ${\cal I}$ will be the claimed graph $\underline Q$.
  Note that $\rp{b_p} = \rpp{^*}{b_p}$ and $\lp{b_q} = \lpp{^*}{b_q}$.

  \begin{figure}[h]
  \centering  \includegraphics{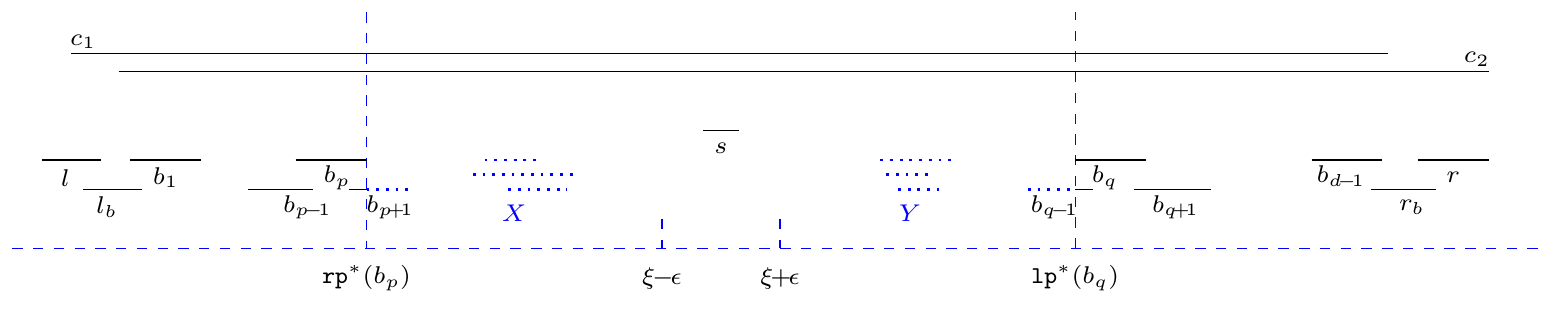} 
  \caption{The interval model $\cal I$ (solid and dotted intervals are
    from $\cal I^*$ and $\cal A$ respectively).}
  \label{fig:interval-edge-deletion-main}
  \end{figure}

  We consider first edges incident to $v\in V(Q)\setminus U_I$.  If
  $I^*_v$ intersects $[\rpp{^*}{b_p}, \lpp{^*}{b_q}]$, then by
  Claim~\ref{lem:break-segment-3} and the definition of projecting
  operation, $v$ has the same neighbors in $\underline Q^*$ and
  $\underline Q$.  Otherwise, $I^*_v$ either fully contains or is
  disjoint from $[\rpp{^*}{b_p}, \lpp{^*}{b_q}]$, and in both cases
  $v$ has the same neighbors in $\underline Q^*$ and $\underline Q$.
  On the adjacency of $U_I\times U_I$, it is easy to verify that
  $\underline Q$ is disjoint from $E_{\ell, X,Y}$.  For another pair
  of vertices, the definition of projecting operation, they are
  adjacent in $\underline Q$ if and only if they are adjacent in $G$.
  Therefore, $\underline Q = \underline Q^* + E(Q[U_I]) - E_{\ell,
    X,Y}$, and it is a subgraph of $Q$.  It follows that $\underline
  G$ is a subgraph of $G$.

  To argue that $\underline G$ is an interval graph as well, we verify
  that every simplicial module of $Q$ remains simplicial in
  $\underline Q$.  It suffices to consider those modules that possibly
  have different neighborhoods in $\underline Q^*$ and $\underline Q$,
  i.e., modules in $U_I$.  Let $v\in U_I$.  If $I_v$ is to the left of
  $\xi-\epsilon$, then $N_{\underline Q}(v)$ comprises $N_Q(s)$ as
  well as those vertices in $N_{\underline Q}(v)\cap U_I$ whose
  intervals are also to the left of $\xi-\epsilon$; the are clearly a
  clique.  Therefore, $\underline G$ is an interval graph.  It remains
  to show that $\underline G$ is maximum, for which it suffices to
  consider $\underline Q^*$ and $\underline Q$.  Recall that $I^*_s$
  must lie between $[\rpp{^*}{b_p}, \lpp{^*}{b_q}]$.  Therefore,
  $b_{4k+4}$ and $b_{d-4k-4}$ are disconnected in the subgraph
  $\underline Q^*- N_{\underline Q^*}(s)$.  In other words, there must
  be some $E_{\imath, X',Y'}$ that are not in $E(\underline G)$, whose
  size is no less than $E_{\ell, X,Y}$ by assumption.  On the other
  hand, all different edges of $\underline Q$ and $\underline Q^*$ are
  incident to $U_I$.  In summary, we have $||\underline G||\ge
  ||\underline G^*||$, and this concludes the proof.
\end{proof}

\subsection{Breaking large caws by edge additions}
Since holes are easy to fill, we may always assume that the graph $G$
is chordal, and the decomposition $\cal K$ is a caterpillar.  Clearly,
$l_b$ and $r_b$ must be in the main path of the caterpillar.  They
together decide a left-right relation with respect to $\cal K$, and
more specifically, separate $\cal K$ into three parts, the left
(modules $v$ with $\rint{v}\le \rint{l_b}$), the right (modules $v$
with $\lint{v}\ge \lint{r_b}$), and the middle (other modules).  As
long as $F$ is unchanged in a \misp\ $\widehat G$, in an interval
model $\cal I$ for $\widehat G$, intervals for $F$ are arranged as
Fig.~\ref{fig:fixed-frame}.  One should note that in $\cal I$,
however, intervals for the right-hand side vertices might intersect
$I_{c_1}\cap I_{c_2}$ or even lie to the left of it.  Or informally,
we might have to bend the right-hand side of $F$ to merge it to the
left (see, e.g., Fig.~\ref{fig:completion-example}), and a symmetric
operation might be applied to the left-hand side as well.  This is the
main source of complication of the disposal of a minimal frame.
\begin{figure*}[h]
  \centering  \includegraphics{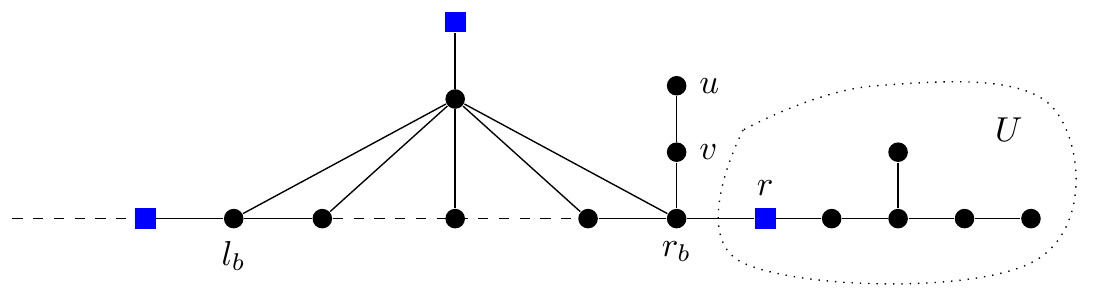} 
  \caption{If $\#u$ and $\#v$ are sufficiently large, then we have to
    add edges to connect $U$ to the ``left'' of $r_b$.}
  \label{fig:completion-example}
\end{figure*}

\begin{lemma}\label{lem:interval-completion-main}
  Let $S_\imath := K_{\imath}\cap K_{\imath+1}\setminus N_Q(s)$ for
  $\rint{l_b}\le\imath<\lint{r_b}$, and let $\ell$ be an index such
  that $\#S_\ell$ is minimum among all of them.  If there is a \misp\
  $\widehat G^*$ of $G$  such that
  \begin{enumerate}
  \item all maximal strong modules of $G$ remains modules of $\widehat
    G^*$, which define the quotient graph $\widehat Q^*$, and
  \item $\widehat Q^*$ does not contain any missing edge of $F$,
  \end{enumerate}
  then there is a \misp\ $\widehat G$ of $G$ that contains all edges
  between $s$ and $S_\ell$.
\end{lemma}
\begin{proof}
  Let $\widehat Q = \widehat Q^* - (\{s\} \times IN(F)) + (\{s\}\times
  S_\ell)$; we show that replacing $\widehat Q^*$ by $\widehat Q$
  gives the claimed \misp\ of $G$.  It is clear there that must be
  some $\ell'$ such that $\{s\}\times S_{\ell'}\subseteq E(\widehat
  Q^*)$.  The selection of $\ell$ then implies $||\widehat Q||\le
  ||\widehat Q^*||$ and $||\widehat G||\le ||\widehat G^*||$.  Also
  clear is that $G \subset \widehat G$.  Therefore, it suffices to
  show that $\widehat G$ is an interval graph.

  Let ${\cal I}^* = \{I^*_v: v\in V(Q)\}$, where $I^*_v =
  [\lpp{^*}{v}, \rpp{^*}{v}]$, be a normalized interval model for
  $\widehat Q^*$ extracted from a normalized interval model for
  $\widehat G^*$.  We now modify it into an interval model for
  $\widehat Q$.  We will use the same values of $\rpp{^*}{l_b}$ and
  $\lpp{^*}{r_b}$ for $\rp{l_b}$ and $\lp{r_b}$.  Let $U_L := \{v:
  \rint{v}\le \rint{l_b}\}$ and $U_R := \{v: \lint{v}\ge \lint{r_b}\}$.
  Observe that $\{U_L, IN(F), U_R, N_Q[s]\}$ partitions $V(Q)$, and each
  of them induces a connected subgraph.  We start from characterizing
  the endpoints of intervals for $U_L$ and $U_R$ that lie between
  $[\rp{l_b}, \lp{r_b}]$.  For this purpose, we make following
  assumptions:
  \begin{enumerate}[(1)]
  \item The interval for ${l_b}$ is to the left of that for ${r_b}$ in
    ${\cal I}^*$, i.e., $\rp{l_b}<\lp{r_b}$.  Since $l_b\not\sim r_b$
    in $\widehat Q^*$, this assumption is clear.
  \item An endpoint $\rho$ of an interval $I^*_v$ is integral if and
    only if one of the following conditions is satisfied:
    \begin{inparaenum}[(i)] 
    \item $v\in IN(F)$; 
    \item $v\in U_L$ and $\rho \le \rp{l_b}$; and
    \item $v\in U_R$ and $\rho \ge \lp{r_b}$.
    \end{inparaenum} 
    These integral endpoints are assumed to be consecutive.
  \end{enumerate}

  By definition, $l_b\in U_L$ and $r_b\in U_R$.  Hence $\lp{r_b}$ and
  $\rp{l_b}$ are both integral, and they separate the real line into
  three parts.  The \emph{thickness} of a nonintegral point $\rho$ is
  defined by
  \[
  \theta_\rho :=
  \begin{cases}
    \#\{v\not\in U_R : \rho \in I^*_v\} & \text{if } \rho < \rp{l_b},
    \\
    \#\{v\not\in U_L : \rho \in I^*_v\} & \text{if } \rho > \lp{r_b},
    \\
    \#\{v\in IN(F) : \rho \in I^*_v\} & \text{otherwise.}
  \end{cases}
  \]
  By assumption (2), all the modules counted in the thickness have
  integral endpoints.  As a result, for every integer $\alpha$, all
  points in $(\alpha, \alpha+1)$ have the same thickness.  The
  following claim characterizes intervals for $U_R$ that reaches
  $\lp{r_b}$ or even $\rp{l_b}$.  If such intervals exist, then
  $\lpp{^*}{U_R} < \lp{r_b}$.
  \begin{claim}
    Let $\alpha := \lpp{^*}{U_R}$.
    \begin{enumerate}
    \item If $\alpha < \rp{l_b}$, then $\theta_\alpha< \theta_\rho$
      holds for any nonintegral point $\rho$ with $\lceil \alpha\rceil
      < \rho <\rp{l_b}$.
    \item If $\rp{l_b}< \alpha < \lp{r_b}$, then $\theta_\alpha<
      \theta_\rho + \#\{v\in U_L: \rho\in I^*_v\}$ holds for any
      nonintegral point $\rho$ with $\max\{\lceil \alpha\rceil,
      \rp{l_b}\} < \rho <\lp{r_b}$.
    \end{enumerate}
  \end{claim}
  \begin{proof}
    By definition, $\alpha$ can be integral only when $\alpha =
    \lp{r_b}$, and hence $\alpha$ is nonintegral in both cases.  The
    right-hand terms of both inequalities are the number of intervals
    for $U_L\cup IN(F)$ that contain $\rho$.  Here we prove the first
    assertion, and the second follows analogously. Suppose, for
    contradiction, that there exists some point $\rho$ with
    $\theta_\alpha\ge \theta_\rho$.  Without loss of generality, we
    may assume that $\rho$ is one of the leftest points having this
    property (there must be a pair of endpoints $\beta_1$ and
    $\beta_2$ such that every point in between has this property).  We
    project all endpoints of intervals for $U_R$ in $[\alpha,\rho]$
    from $[\alpha,\rho]$ to $[\rho-\epsilon, \rho]$, and argue that
    the modified intervals defines an interval supergraph of $Q'$ with
    strictly less edges than $\widehat Q^*$.  Necessarily $\rho$
    avoids any simplicial module of $\widehat Q^*$, and thus every
    module whose interval containing $[\rho-\epsilon, \rho]$ (if and
    only if it contains $\rho$) is a clique.  Therefore, $Q'$ implies
    an interval supergraph $G'$ (having $Q'$ as a quotient graph) of
    $G$.

    We verify first that the interval graph defined by the new model
    is a supergraph of $Q$.  Only intervals for $U_R$ have been
    modified, while adjacency between any pair of modules in $U_R$ is
    retained.  Thus, it suffices to consider a pair of modules $v\in
    V(Q)\setminus U_R$ and $u\in U_R$.  Such a module $v$ is adjacent
    to $U_R$ in $Q$ only if $v\sim r_b$, and thus in the new model, the
    intervals for $u$ and $v$ remains intersecting.  As
    $\theta_\alpha\ge \theta_\rho$, the interval supergraph defined by
    the new model contains strictly less edges than $\widehat Q^*$.
    This contradiction justifies the claim.
    \renewcommand{\qedsymbol}{$\lrcorner$}
  \end{proof}

  Intervals for $U_L$ that reaches $\rp{l_b}$ have a symmetrical
  characterization.  Let $\theta$ be the minimum thickness obtained in
  $[\rp{l_b}, \lp{r_b}]$.  Let $\xi_1$ and $\xi_2$ be the smallest and
  largest integer in $[\rp{l_b}, \lp{r_b}]$ such that
  $\theta_{\xi_1+0.5} = \theta = \theta_{\xi_2+0.5}$.  It is possible
  that $\xi_1 = \xi_2$; in this case, we increase every endpoint that
  is at least $\xi_1+1$ by $1$, and then reset $\xi_2$ to $\xi_1+1$.
  Henceforth we can assume $\xi_1 < \xi_2$.  Note that then the
  integral point $\xi_2$ may not be an endpoint of any interval; this
  fact is immaterial in the argument to follow.  The following claim
  further characterizes the endpoints of these intervals that lie
  between $[\rp{l_b},\lp{r_b}]$.

\begin{figure}[t]
  \centering  \includegraphics{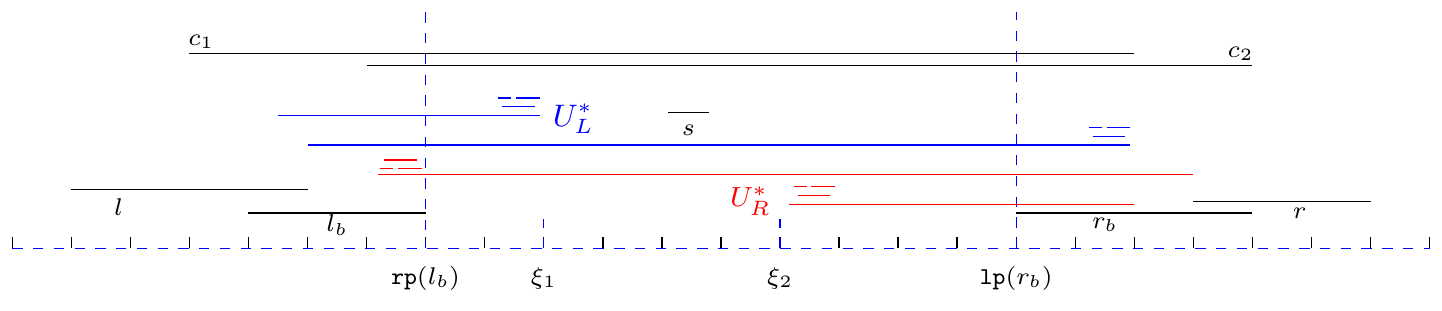} 
  \caption{The interval model used in the proof of
    Lem.~\ref{lem:interval-completion-main}.}
  \label{fig:claim-2}
\end{figure}
  
  \begin{claim}\label{claim:main}
    Let $\alpha := \lpp{^*}{U_R}$ and $\beta := \rpp{^*}{U_L}$.  No interval for
    $v\in U_R$ has an endpoint in $[\lceil \alpha\rceil, \xi_2]$, and no
    interval for $v\in U_L$ has an endpoint in $[\xi_1, \lfloor
    \beta\rfloor]$.
  \end{claim}
  \begin{proof}
    We show by contradiction, that is, we construct a strictly smaller
    supergraph of $G$ than $\widehat G^*$ supposing such an endpoint
    exists.  Let $U^*_L := \{v\in U_L: \rpp{^*}{v} > \rp{l_b}\}$ and
    $U^*_R := \{v\in U_R: \lpp{^*}{v} < \lp{r_b}\}$.  If
    $\alpha\ge\xi_2$ , then $I_v = I^*_v$ for every $v\in U_R$.  In
    the remaining case, $\alpha$ is non-integral.  If
    $\rp{l_b}<\alpha<\xi_2$ then we project all endpoints for $U_R$ in
    $[\alpha, \xi_2+1]$ to $[\xi_2+\epsilon, \xi_2+1]$.  Otherwise,
    $\alpha < \rp{l_b}$ (noting that $\alpha\ne \rp{l_b}$), we take a
    point $\rho$ in $[\alpha, \xi_2]$ such that $\#\{v\in U_R :
    \rho\in I_v\}$ is minimum.  We project all endpoints for $U_R$ in
    $[\alpha,\rho]$ to $[\alpha,\lceil \alpha\rceil]$, and all
    endpoints for $U_R$ in $[\rho, \xi_2+1]$ to $[\xi_2, \xi_2+1]$.
    Note that the operation is based on the thickness, which is
    irrelevant to $U^*_L$ and $U^*_R$.  Therefore, the operations are
    well-defined.  We have given above the details on modifying
    intervals for $U_R$ and $s$.  We apply symmetrical modifications
    to $U_L$.  Meanwhile, we keep intervals for $N_Q(s)$ and $IN(F)$
    unchanged.

    It is easy to verify that this new model defines an interval
    supergraph of $G$; let it be $\widehat G'$.  We now verify its
    size is no more than $||\widehat G^*||$.  Clearly, the adjacency
    between any pair of modules in $U_R$ is not changed.  We calculate
    the three parts, namely, between $U^*_L$ and $V(G)\setminus
    (U^*_L\cup U^*_R)$, between $U^*_R$ and $V(G)\setminus (U^*_L\cup
    U^*_R)$, and between $U^*_L$ and $U^*_R$.

    Consider first edges of $\widehat G'$ between $U^*_R$ and
    $V(G)\setminus (U^*_L\cup U^*_R)$.  It suffices to consider edges
    between $U^*_R$ and $U_0 := \{v: I^*_v \cap (\lfloor
    \alpha\rfloor, \xi_2 + 1)\ne\emptyset\}$.  Let
    \begin{inparaenum}[\itshape i\upshape)]
    \item $U_1$ be the set of modules with intervals in
      $[\alpha,\rho]$ (i.e., its new interval is in $[\alpha, \lceil
      \alpha\rceil]$);
    \item $U_2$ be the set of modules with intervals containing
      $\rho$ (i.e., its new interval contains in $[\lceil \alpha\rceil],
      \xi_2$); and
    \item $U_3$ be the set of modules with intervals in $[\rho,
      \xi_2+1]$ (i.e., its new interval is in $[\xi_2, \xi_2+1]$).
    \end{inparaenum}
    In the new graph, the number of edges between $U^*_R$ and
    $V(G)\setminus (U^*_L\cup U^*_R)$ is:
    \begin{equation}
      \label{eq:8-1}
      \#U_1\cdot \theta_{\alpha} + \#U_2\cdot \#U_0 + \#U_3\cdot \theta      
    \end{equation}
    Consider next edges between $U^*_R$ and $V(G)\setminus (U^*_L\cup
    U^*_R)$ in $\widehat G^*$; to count them, we assign each edge to
    precisely one of its ends:
    \begin{itemize}
    \item If the left endpoint of some $v\in U_0$ is an integer $i$
      with $\alpha<i<\xi_2$, then we assign to $v$ all edges between
      it and $\{u\in U^*_L: i\in I^*_u\}$.
    \item If the right endpoint of some $v\in U_0$ is an integer $i+1$
      with $\alpha<i\le\xi_2$, then we assign to $v$ all edges between
      it and $\{u\in U^*_L: i+1\in I^*_u\}$.
    \item For every $v\in U^*_L$ with $\rpp{^*}{v}\in (i,i+1)$ and
      $\rpp{^*}{v}< \rho$, we assign to it all edges between $v$ and
      $\{u\in U_0: [i,i+1]\subseteq I^*_u\}$.
    \item For every $v\in U^*_L$ with $\lpp{^*}{v}\in (i,i+1)$ and
      $\lpp{^*}{v}> \rho$, we assign to it all edges between $v$ and
      $\{u\in U_0: [i,i+1]\subseteq I^*_u\}$.
    \end{itemize}
    The total number of edges we have assigned is at least
    \eqref{eq:8-1}.  Clearly, edges assigned to different modules are
    disjoint, and thus the number of edges between $U^*_R$ and
    $V(G)\setminus (U^*_L\cup U^*_R)$ in $\widehat G^*$ is no less
    than $\widehat G'$.  The equality holds only when no interval for
    $v\in U_R$ has an endpoint in $[\lceil \alpha\rceil, \xi_2]$.

    A symmetric argument applies to edges between $U^*_L$ and
    $V(G)\setminus (U^*_L\cup U^*_R)$.  It remains to count edges
    between $U^*_L$ and $U^*_R$.  In $\widehat G'$, the number is
    precisely $\sum_{v\in U^*_L}\#v \#(N_{\widehat G^*}(v)\cap
    U^*_R)$.  By the selection of $\rho$ and noting that $\xi_1 <
    \xi_1+1\le \xi_2$, the number in $\widehat G$ is no less than
    this.  Summing the three parts up, we conclude that $||\widehat
    G'||< ||\widehat G^*||$ when either of the claimed conditions is
    not satisfied.  This contradiction concludes the claim.
    \renewcommand{\qedsymbol}{$\lrcorner$}
  \end{proof}

  Claim~\ref{claim:main} implies that the point $\rho\in
  [\rp{l_b},\lp{r_b}]$ where $\#\{v: \rho\in I^*_v\}$ is minimized
  always lies between $\xi_1$ and $\xi_2+1$, and it is completely
  decided by intervals for $IN(F)$.  Adding edges among $IN(F)$ will
  not decrease the thickness of any point.  Therefore, the fact that
  $\widehat G^*$ is minimum implies $\widehat G^*[IN(F)] = G[IN(F)]$.
  By the selection of $S_\ell$, all nonintegral points in $[\alpha,
  \beta] = \bigcap_{v\in S_\ell} I^*_v$ have the minimum thickness;
  they are contained in $(\xi_1, \xi_2+1)$.  Note that all intervals
  containing $[\alpha, \beta]$ are non-simplicial, hence cliques.  We
  can pick any $\rho$ from $(\alpha, \beta)$ and set $I_s$ to $[\rho -
  \epsilon, \rho+\epsilon]$.  Clearly, the model we have obtained
  represents $\widehat Q$, and then $\widehat G$ is an interval graph.
  Therefore, there is a \misp\ satisfying the claimed condition.
\end{proof}

 \section{The algorithms}\label{sec:algs}
A very sketchy outline (Fig.~\ref{fig:alg-interval-deletion}) for our
algorithms has been presented at the end of Section~\ref{sec:outline}.
Now with all the details developed in preceding sections, we are
finally ready to flesh it out.

\subsection{The algorithm for {\sc interval vertex deletion}}
In any circular-arc model, the arcs corresponding to vertices of a
hole have to cover the entire circle.  The converse holds true as well
if the model is normal and Helly: a minimal set of arcs that cover the
circle has size at least four.  Therefore, we have the following
observation on non-chordal graphs.
\begin{proposition}[Lin et al.~\cite{lin-13-nhcag-and-subclasses}]
  \label{lem:normal-and-helly}
  If a normal Helly circular-arc graph $G$ is not chordal, then every
  circular-arc model of $G$ is normal and Helly.
\end{proposition}

We say that a set $V_-$ of vertices is a \emph{hole cover} of $G$ if
$G - V_-$ is chordal.  Recall that we are only concerned with hole
covers of normal Helly circular-arc graphs, where the deletion of a
hole cover actually leaves an interval graph.  For any particular
point in a normal and Helly circular-arc model, the set of vertices
whose arcs contain this point makes a hole cover; more importantly,
any hole cover contains such a set of vertices.  According to
Prop.~\ref{lem:normal-and-helly}, the model given by
Prop.~\ref{lem:clique-decomposition-interval-model} is normal and
Helly, and thus we can derive the following proposition, which bears
striking resemblance with the minimal vertex separators of interval
graphs.
\begin{proposition}\label{lem:hole-cover+}
  Let $\cal K$ be the hole decomposition of a normal Helly
  circular-arc graph $G$.  For any inclusion-wise minimal hole cover
  $V_-$ of $G$, there is a bag $K_\ell$ of $\cal K$ such that $V_- =
  K_\ell\cap K_{\ell + 1}$.
\end{proposition}
In the special case when the graph is already chordal, the only
minimal hole cover is empty, which can be chosen as the intersection
of the two end bags.  But in general, the bag $K_\ell$ in
Prop.~\ref{lem:hole-cover+} is not necessarily unique.  Since there
are at most $|G|$ bags in $\cal K$, a hole cover of the minimum size
can be obtained in linear time.  We would like to point out that
Prop.~\ref{lem:hole-cover+} remains true in a graph with an olive-ring
decomposition; see \cite{cao-14-interval-deletion}.

The algorithm presented in Fig.~\ref{fig:alg-interval-vertex-deletion}
looks for a \miib\ for a ``YES'' instance, i.e., when the input graph
$G$ has an induced interval subgraph on $|G| - k$ vertices.
Otherwise, it terminates by returning ``NO.''
\begin{figure*}[t]
\setbox4=\vbox{\hsize28pc \noindent\strut
\begin{quote}
  \vspace*{-5mm} \small

  {\bf Algorithm interval-vertex-deletion($G,k$)}
  \\
  {\sc input}: a graph $G$ and an integer $k$.
  \\
  {\sc output}: a set $V_-$ of at most $k$ vertices s.t. $ G -
  V_-$ is a \miib\ of $G$; or ``NO.''

  0 \hspace*{2ex} {\bf if} $k<0$ {\bf then return} ``NO''; {\bf if}
  $G$ is an interval graph {\bf then return} $\emptyset$; $V_- =
  \emptyset$;
  \\
  1 \hspace*{2ex} {\bf if} the quotient graph $Q$ defined by maximal
  strong modules of $G$ is a clique {\bf then}
  \\
  1.1 \hspace*{3ex} {\bf if} two maximal strong modules are not
  cliques {\bf then}
  \\
  \hspace*{10ex} find a $4$-hole and {\bf branch} on deleting one
  vertex from it;
  \\
  \hspace*{6ex} $\setminus\!\!\setminus$ {\em Since $G$ is not an
    interval graph, at least one module is nontrivial.}
  \\
  1.2 \hspace*{3ex} {\bf return interval-vertex-deletion}($G[M], k$),
  where $M$ is the only non-clique module;
  \\
  2 \hspace*{2ex} {\bf if} $Q$ is edgeless {\bf then}
  \\
  2.1 \hspace*{3ex} {\bf for each} component $M$ of $G$ {\bf do}
  \\
  \hspace*{10ex} $V_M = $ {\bf interval-vertex-deletion}($G[M], k$);
  \\
  \hspace*{10ex} {\bf if} $V_M = $ ``NO'' {\bf then return} ``NO'';
  \\
  \hspace*{10ex} $V_- = V_-\cup V_M$; $\quad$ $k = k - |V_M|$;
  \\
  2.2 \hspace*{3ex} {\bf return} $V_-$;
  \\
  3 \hspace*{2ex} {\bf call decompose}($Q$);
  \\
  4 \hspace*{2ex} {\bf if} a short hole or small caw is found {\bf
    then branch} on deleting one vertex from it;
  \\
  $\setminus\!\!\setminus$ {\em We have an olive-ring decomposition
    $\cal K$.}
  \\
  5 \hspace*{2ex} {\bf if } a module $M$ represented by $v\in
  V(Q)\setminus SI(Q)$ is not a clique {\bf then}
  \\
  \hspace*{7ex} find a $4$-hole and {\bf branch} on deleting one
  vertex of it;
  \\
  6 \hspace*{2ex} {\bf for each} module $M$ represented by $v\in
  SI(Q)$ {\bf do}
  \\
  \hspace*{7ex} $V_M = $ {\bf interval-vertex-deletion}($G[M], k$);
  \\
  \hspace*{7ex} {\bf if} $V_M = $ ``NO'' {\bf then return} ``NO'';
  \\
  \hspace*{7ex} $V_- = V_-\cup V_M$; $\quad$ $k = k - |V_M|$;
  \\
  7 \hspace*{2ex} {\bf if } $\cal K$ is not a hole {\bf then}
  \comment{here vertices are in $Q$.}
  \\
  7.1 \hspace*{3ex} {\bf call} Lem.~\ref{lem:find-minimal-caw} to find
  a small caw or a minimal frame $F$;
  \\
  7.2 \hspace*{3ex} {\bf if} a small caw is found {\bf then branch} on
  deleting one vertex from it;
  \\
  7.3 \hspace*{3ex} {\bf call} {\bf branch} on deleting either one
  module from $F$ or $S_\ell$ specified in
  Lem.~\ref{lem:interval-deletion-main};
  \\
  8 \hspace*{2ex} {\bf call} Prop.~\ref{lem:hole-cover+}.

\end{quote} \vspace*{-6mm} \strut} $$\boxit{\box4}$$
\vspace*{-9mm}
\caption{Algorithm for \textsc{interval vertex deletion}.}
\label{fig:alg-interval-vertex-deletion}
\end{figure*}

\begin{theorem}\label{thm:alg-interval-deletion}
  Algorithm {\bf interval-vertex-deletion} solves \textsc{interval
    vertex deletion} in $O(8^k\cdot ||G||)$ time.
\end{theorem}
\begin{proof}
  Let us verify first its correctness.  Step 0 gives two trivial exit
  conditions.  Steps 1 and 2 take care of the case where the quotient
  graph $Q$ is not prime.  The algorithm enters step 1 if $Q$ is a
  clique.  Step 1.1 is straightforward, and if it is not the case,
  then there exists a unique \msm\ $M$ that is not trivial (because
  $G$ is not an interval graph).  In other words, all vertices in
  $V(G)\setminus M$ are universal, and hence it suffices to solve the
  problem on the subgraph $G[M]$ with the same parameter $k$ (step
  1.2).  The algorithm enters step 2 if $Q$ is edgeless; its
  correctness is clear.  Now that $Q$ is prime, step~3 applies
  algorithm decompose to it.  Steps 4 and 5 are straightforward.  Step
  6 finds a \miib\ of $G[M]$ for each simplicial module in $Q$; its
  correctness follows from Thm.~\ref{thm:modules-induced-subgraph+}.
  Note that this step does not change the quotient graph $Q$, and in
  the rest of the algorithm we only need to consider $Q$, i.e., a
  maximal strong module $M$ either is deleted or remains intact.  The
  algorithm enters step 7 when the decomposition $\cal K$ is not a
  hole.  Step 7.1 first calls Lem.~\ref{lem:find-minimal-caw} to find
  a small caw or minimal frame $F$, step 7.2 is trivial, and step 7.3
  is justified by Lem.~\ref{lem:interval-deletion-main}.  Step 8
  follows from Prop.~\ref{lem:hole-cover+}.

  Each branching step has at most $8$ sub-instances, in which the
  parameter is decreased by at least one; the total number of
  sub-instances is thus upper bounded by $8^{k}$.  Each sub-instance
  can be produced in $O(||G||)$ time.  Finally, step~8 takes
  $O(||G||)$ time for each sub-instance.  Therefore, the algorithm
  runs in $O(8^k\cdot ||G||)$ time.
\end{proof}

We now show how to adapt algorithm interval-vertex-deletion into an
approximation algorithm for the minimum interval vertex deletion
problem as follows.  The main changes are made in the branching steps,
where we need to make a choice among at most eight sets of vertices.
Whenever a small caw or a short hole is found, instead of branching
into different directions, we simply delete the whole subgraph.  For a
large caw in the quotient graph $Q$, step 7.3 of algorithm
interval-vertex-deletion finds seven or eight vertex sets and branches
on deleting one of them.  Our approximation algorithm will instead
picks from each set an arbitrary vertex to delete.  By
Lem.~\ref{lem:interval-deletion-main}, at least one of these deleted
vertices is in some minimum modification.  Therefore, the number of
minimum modifications of the remaining graph decreases by at least one
(here we are using the hereditary property).  This verifies that the
performance ratio is 8.  Observing that we can delete at most $|G|$
vertices, and each deletion operation can be done in linear time, the
approximation algorithm can be implemented in $O(|G|\cdot ||G||)$
time.
\begin{theorem}
  There is an $O(|G|\cdot ||G||)$-time approximation algorithm for the
  minimum interval vertex deletion problem with performance ratio 8.
\end{theorem}

\subsection{The algorithm for {\sc interval edge deletion}}
Again, we start from the disposal of holes in a normal Helly
circular-arc graph $G$.  We say that a set $E_-$ of edges is an
\emph{edge hole cover} of $G$ if $G - E_-$ is chordal.  Let $\cal K$
be a clique hole decomposition of $G$.  For each pair of adjacent
vertices $v\in K_\ell$ and $u\not\in K_\ell$, we can define a
left-right relation with respect to the bag $K_\ell$.  We write
$v\rightarrow u$ if $\lint{u} \in [\ell+1,\rint{v}]$, and $v\leftarrow
u$ otherwise (i.e., $\rint{u} \in [\lint{v}, \ell - 1,]$).  It should
be noted here that $[\lint{u},\rint{u}]$ might be a subset of
$[\lint{v}, \rint{v}]$, and thus they do not have a left-right
relation in general.  Any (possibly trivial) partition ($X,Y$) of
$K_\ell$ defines the following edge set
\begin{equation}
  \label{eq:1}
  E_- =  (X\times Y) \cup \{v u: v\in X,\, u\not\in K_\ell, v\rightarrow u\} \cup
  \{v u: v\in Y,\, u\not\in K_\ell, v\leftarrow u\}.
\end{equation}
\begin{proposition}\label{lem:edge-hole-cover-1}
  For any bag $K_\ell$ of $\cal K$ and any partition ($X,Y$) of
  $K_\ell$, the set $E_-$ given by \eqref{eq:1} is an edge hole cover
  of $G$.
\end{proposition}
\begin{proof}
  We build an interval model for $G - E_-$.
  Without loss of generality, we may circularly renumber the bags such
  that the bag is $K_0$, and then we take the circular-arc model $\cal
  A$ given by Prop.~\ref{lem:clique-decomposition-arc-model}.  Note
  that an arc $A_v$ contains the point $0$ if and only if $v\in K_0$.
  Let $c$ denote the length of the circle in the model.  We use the
  following intervals:
  \[
  \label{eq:arcs}
  I_v := 
  \begin{cases}
    [\lp{v}, c] & \text{if } v\in X,
    \\
    [0, \rp{v} ] & \text{if } v\in Y,
    \\
    [\lp{v}, \rp{v} ] & \text{otherwise } (v\not\in K_0).
  \end{cases}
  \]
  It is easy to verify this set of intervals represents $G - E_-$.
\end{proof}
More important is the other direction, which, however, is a
little bit of more technical.  Recall that we can compute the length
of the shortest hole in linear time
(Lem.~\ref{lem:find-shortest-hole}).
\begin{lemma}\label{lem:edge-hole-cover+}
  Let $G$ be a normal Helly circular-arc graph that does not contain a
  hole of length less than $3 k + 6$.  Any minimal edge hole cover
  $E_-$ of size at most $k$ is defined by \eqref{eq:1} with some bag
  $K_\ell$ and a partition ($X,Y$) of $K_\ell$.
\end{lemma}
\begin{proof}
  Let $\cal A$ be the circular-arc model for $G$ given by
  Prop.~\ref{lem:clique-decomposition-arc-model}, which must be normal
  and Helly (Prop.~\ref{lem:normal-and-helly}).  It suffices to show
  that any minimal hole cover $E_-$ of size at most $k$ contains some
  set of edges given by \eqref{eq:1}, and then the equality follows
  from the minimality of $E_-$ and Prop.~\ref{lem:edge-hole-cover-1}.
  The case $k=1$ can be easily checked; hence in the following we may
  assume that $k\ge 2$.

  Suppose, for contradiction, that $E_-$ does not contain for any
  partition ($X,Y$) of any bag the edge set defined by \eqref{eq:1}.
  Then we find a hole of $G' = G - E_-$ as follows.  Let $H$ be a
  shortest hole of $G$; denote its vertices by $v_1,\cdots, v_{|H|}$,
  where the indices should be understood as modulo $|H|$.  Note that
  for any $1\le i\le |H|$, the neighborhoods $N[v_i]$ and $N[v_{i+3}]$
  are disjoint, as otherwise there exists a shorter hole than $H$.
  Since $|E_-|\le k$, there are at least two vertices $v_i$ and $v_j$
  with $i+3\le j$ such that $E_-$ contains no edges induced by
  $N[v_i]$ or $N[v_j]$.  We may assume that the arcs for $v_i$ and
  $v_j$ are not contained in other arcs; otherwise, we can simply
  replace $v_i$ or $v_j$ by it.  Since $H$ is the shortest, after the
  replacement the new cycle is still a hole.  Note that $G'$ may or
  may not be a circular-arc graph, and every arc we mentioned below is
  referred to the arc in $\cal A$, the circular-arc model for $G$.

  Starting from $v_i$ as $x$, the next vertex is chosen from
  $N_{G'}(x)$ such that its arc is the rightmost; we continue till the
  first new arc reaches \lp{v_j}.  This gives an induced
  \stpath{v_i}{v_j} $P_{ij}$ in $G'$ with all arcs in $[\lp{v_i},
  \rp{v_j}]$.  Likewise, we can find another induced \stpath{v_j}{v_i}
  $P_{j i}$ in $G'$ with all arcs in $[\lp{v_j}, \lp{v_i}]$.
  Combining these two paths yields a cycle, which, however, might not
  be induced.  In particular, the two neighbors of $v_i$ (or $v_j$) in
  the two paths might be adjacent in $G'$.  In this case, we omit
  $v_i$ (or $v_j$).  Consequentially, we obtain an induced cycle of
  $G'$ whose arcs cover the entire circle in the model $\cal A$, which
  means that it contains at least $|H|$ vertices and is a hole of
  $G'$.  The existence of a hole in $G'$ contradicts that $E_-$ is a
  edge hole cover and concludes the proof.
\end{proof}

Prop.~\ref{lem:hole-cover+} and Lem.~\ref{lem:edge-hole-cover+} are
the technical versions of Lem.~\ref{lem:hole-cover}.  Similar as
Prop.~\ref{lem:hole-cover+}, Lem.~\ref{lem:edge-hole-cover+} can also
be adapted for a graph with an olive-ring decomposition, but the edge
variation is more complicated, and slight modifications are needed to
make it work.

The algorithm presented in Fig.~\ref{fig:alg-interval-edge-deletion}
looks for a \misb\ for a ``YES'' instance, i.e., when the input graph
$G$ has a spanning interval subgraph on $||G|| - k$ edges.  Its first
six steps are similar as algorithm interval-vertex-deletion
(Fig.~\ref{fig:alg-interval-vertex-deletion}), and the main differences
appear in steps 7 and 8.  If the algorithm enters step 7.4, then it
branches into $O(k)$ sub-instances.  Otherwise, to apply
Lem.~\ref{lem:interval-edge-deletion-main}, we need to find the bag
$K_\ell$ and a nontrivial partition ($X,Y$) of it.  Note that if a bag
contains more than $k+1$ vertices, then the edge set $E_{\ell,X,Y}$
decided by any nontrivial partition will be larger than $k$; hence we
only need to check those bags containing at most $k+1$ vertices.  Such
a bag has at most $2^{k+1}$ partitions, and thus we can find a minimum
one in $O(2^{k+1}\cdot |G|)$ time, and then we again branch into
$O(k)$ sub-instances.  Likewise, Lem.~\ref{lem:edge-hole-cover+} (step
9) only needs to check bags of size $k+1$, and hence can be done in
$O(2^{k+1}\cdot |G|)$ time.  The dominating steps are thus 7.4, 7.5,
8, and 9, and the runtime is $k^{O(k)} \cdot ||G||$.  This concludes
the following theorem.
\begin{figure*}[t]
\setbox4=\vbox{\hsize28pc \noindent\strut
\begin{quote}
  \vspace*{-5mm} \small

  {\bf Algorithm interval-edge-deletion($G,k$)}
  \\
  {\sc input}: a graph $G$ and an integer $k$.
  \\
  {\sc output}: a set $E_-$ of at most $k$ edges s.t. $ G - E_-$
  is a \misb\ of $G$; or ``NO.''

  0 \hspace*{2ex} {\bf if} $k<0$ {\bf then return} ``NO''; {\bf if}
  $G$ is an interval graph {\bf then return} $\emptyset$; $E_- =
  \emptyset$;
  \\
  1 \hspace*{2ex} {\bf if} the quotient graph $Q$ defined by maximal
  strong modules of $G$ is a clique {\bf then}
  \\
  1.1 \hspace*{3ex} {\bf if} two maximal strong modules are not
  cliques {\bf then}
  \\
  \hspace*{10ex} find a $4$-hole and {\bf branch} on deleting one edge
  from it;
  \\
  \hspace*{6ex} $\setminus\!\!\setminus$ {\em Since $G$ is not an
    interval graph, at least one module is nontrivial.}
  \\
  1.2 \hspace*{3ex} {\bf return interval-edge-deletion}($G[M], k$),
  where $M$ is the only non-clique module;
  \\
  2 \hspace*{2ex} {\bf if} $Q$ is edgeless {\bf then}
  \\
  2.1 \hspace*{3ex} {\bf for each} component $M$ of $Q$ {\bf do}
  \\
  \hspace*{10ex} $E_M = $ {\bf interval-edge-deletion}($G[M], k$);
  \\
  \hspace*{10ex} {\bf if} $E_M = $ ``NO'' {\bf then return} ``NO'';
  \\
  \hspace*{10ex} $E_- = E_-\cup E_M$; $\quad$ $k = k - |E_M|$;
  \\
  2.3 \hspace*{3ex} {\bf return} $E_-$;
  \\
  3 \hspace*{2ex} {\bf call decompose}($Q$);
  \\
  4 \hspace*{2ex} {\bf if } a short hole or small caw is found {\bf
    then branch} on deleting one edge from it;
  \\
  $\setminus\!\!\setminus$ Hereafter we have an olive-ring
  decomposition $\cal K$.
  \\
  5 \hspace*{2ex} {\bf if } a module $M$ represented by $v\in
  V(Q)\setminus SI(Q)$ is not a clique {\bf then}
  \\
  \hspace*{7ex} find a $4$-hole and {\bf branch} on deleting one edge
  from it;
  \\
  6 \hspace*{2ex} {\bf for each} module $M$ represented by $v\in
  SI(Q)$ {\bf do}
  \\
  \hspace*{7ex} $E_M = $ {\bf interval-edge-deletion}($G[M], k$);
  \\
  \hspace*{7ex} {\bf if} $E_M = $ ``NO'' {\bf then return} ``NO'';
  \\
  \hspace*{7ex} $E_- = E_-\cup E_M$; $\quad$ $k = k - |E_M|$;
  \\
  7 \hspace*{2ex} {\bf if } $\cal K$ is not a hole {\bf then}
  \\
  7.1 \hspace*{3ex} use Lem.~\ref{lem:find-minimal-caw} to find a
  small caw or a minimal frame $F$;
  \\
  7.2 \hspace*{3ex} {\bf if} a small caw is found {\bf then branch} on
  deleting one edge from it;
  \\
  7.3 \hspace*{3ex} find a shortest \stpath{l_b}{r_b} path $B$ through
  $IN(F)$;
  \\
  7.4 \hspace*{3ex} {\bf if} $|B|\le 4k+16$ {\bf then branch} on
  deleting one edge from $F$ or $B$;
  \\
  7.5 \hspace*{3ex} {\bf else branch} using
  Lem.~\ref{lem:interval-edge-deletion-main};
  \\
  8 \hspace*{2ex} {\bf if} the shortest hole of $Q$ has length $\le 3k
  + 6$ {\bf then} branch on deleting one edge from it;
  \\
  9 \hspace*{2ex} {\bf else call} Lem.~\ref{lem:edge-hole-cover+}.

\end{quote} \vspace*{-6mm} \strut} $$\boxit{\box4}$$
\vspace*{-9mm}
\caption{Algorithm for \textsc{interval edge deletion}.}
\label{fig:alg-interval-edge-deletion}
\end{figure*}
\begin{theorem}\label{thm:alg-interval-edge-deletion}
  Algorithm {\bf interval-edge-deletion} solves {\sc interval edge
    deletion} in $k^{O(k)} \cdot ||G||$ time.
\end{theorem}

\subsection{The algorithm for {\sc interval completion}}
It is now well known that holes can be easily filled in.  On the one
hand, the existence of a hole of more than $k+3$ vertices will
immediately imply ``NO'' for the completion problem.  On the other
hand, a hole of a bounded length has only a bounded number of minimal
ways to fill
\cite{kaplan-99-chordal-completion,cai-96-hereditary-graph-modification},
of which an interval supergraph must contain one.
\begin{lemma}[\cite{cai-96-hereditary-graph-modification}]
\label{lem:fill-holes}
  A minimal set of edges that fills a hole $H$ has size $|H| - 3$, and
  the number of such sets is upper bounded by $4^{|H|-3}$.  Moreover,
  they can be enumerated in $O(4^{|H|-3}\cdot ||G||)$ time.
\end{lemma}

In a small caw, the number of edges is $6, 9, 12, 6$, or $10$, and the
number of missing edges (i.e., edges in the complement graph) is $9,
6, 9, 15$, or $11$, respectively.  To fix a small caw by adding edges,
A simpleminded implementation will branch into 15 directions.  We
observe that we do not need to try all of them, and $6$ directions
will suffice.
\begin{proposition}\label{lem:6-enough-completion}
  For each small caw in a graph $G$, there is a set of at most $6$
  edges such that any interval supergraph $\widehat G$ of $G$ contains
  at least one of them.
\end{proposition}
\begin{figure*}[h] 
  \centering
  \begin{subfigure}[b]{0.18\textwidth}
    \centering    \includegraphics{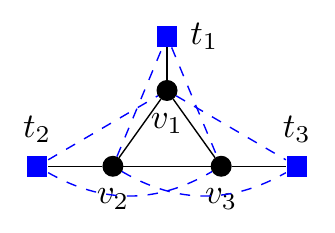} 
  \end{subfigure}%  
  \,
  \begin{subfigure}[b]{0.18\textwidth}
    \centering    \includegraphics{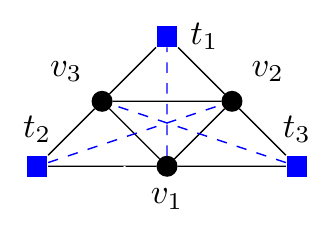} 
  \end{subfigure}%  
  \,
  \begin{subfigure}[b]{0.18\textwidth}
    \centering    \includegraphics{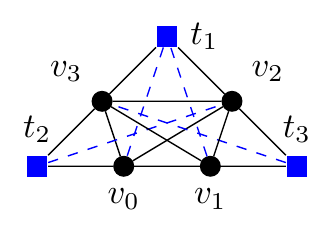} 
  \end{subfigure}
  \,
  \begin{subfigure}[b]{0.18\textwidth}
    \centering    \includegraphics{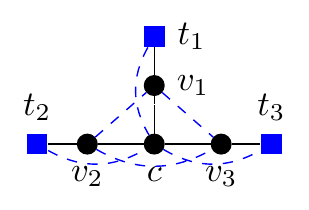} 
  \end{subfigure}%  
  \,
  \begin{subfigure}[b]{0.18\textwidth}
    \centering    \includegraphics{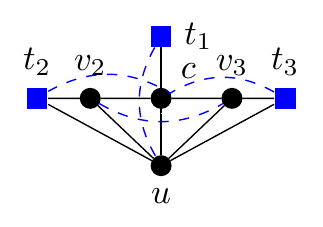} 
  \end{subfigure}%  
  \caption{An interval supergraph must contain some dashed edge.}
  \label{fig:fill-edges-labeled}
\end{figure*}
\begin{proof}
  These edges are depicted as dashed in
  Fig.~\ref{fig:fill-edges-labeled}.  We argue the correctness by
  considering the intervals for these vertices in the object interval
  supergraph $\widehat G$.  Note that for a pair of vertices $v,u$
  such that $I_v\subset I_u$, all neighbors of $v$ must be adjacent to
  $u$ in $\widehat G$; in other words, we need to add edges to connect
  every vertex in $N_G[v]\setminus N_G[u]$ to $u$.

  Net and sun.  The three non-terminal vertices $v_1,v_2$, and $v_3$
  make a triangle.  If there exist $i,j\in \{1,2,3\}$ such that
  $I_{v_i}\subseteq I_{v_j}$, then we must add dashed edges to connect
  the terminal neighbor(s) of $v_i$ to $v_j$.  We assume then
  otherwise, and let $I_{v_i}$ be the middle one among $I_{v_1}$,
  $I_{v_2}$, and $I_{v_3}$.  In the net, we must add a dashed edge to
  connect $t_i$ (the terminal neighbor of $v_i$) to at least one of
  $\{v_1,v_2,v_3\}\setminus \{v_i\}$.  In the sun, we need to add the
  dashed edge $t_i v_i$.

  Rising sun.  If any interval of a non-terminal vertex in
  $\{v_0,v_1,v_2, v_3\}$ contains $I_{v_2}$ or $I_{v_3}$ (besides
  itself), then we need to add one of the dashed edges.  This is also
  the case when $I_{v_0}\subseteq I_{v_2}$ or $I_{v_1}\subseteq
  I_{v_3}$.  Otherwise, $I_{v_0}\cup I_{v_1}$ must properly contain
  $I_{v_2}\cap I_{v_3}$, and then we need to add a dashed edge to
  connect $t_1$ to at least one of $v_0$ and $v_1$.

  Long claw.  The non-terminal vertices make a claw.  If no edge is
  added between them, then there must be $i\in \{1,2,3\}$ such that
  $I_{v_i}\subset I_c$, so we need to add the dashed edge $t_i c$.

  Whipping top.  If we do not add any of $\{v_2 v_3 t_2 c, c t_3\}$,
  then $I_c$ must be properly contained in $I_u$, so we need to add
  the dashed edge $t_1 u$.
\end{proof}

With a similar observation as Prop.~\ref{lem:6-enough-completion}, we
can decrease the number of cases in disposing of a large caw as well.
The frame of a $\dag$ (resp., $\ddag$) has $6$ vertices and $5$ edges
(resp., $7$ vertices and $11$ edges).  Therefore, a frame $F$ always
has $10$ missing edges.  It is worth mentioning that with our general
notation, these $10$ missing edges have exactly the same labels for
\dag s and \ddag s.  As a result, a direct implementation of
Lem.~\ref{lem:interval-completion-main} will branch into 11
sub-instances.  For a frame $F = $ ($s:c_1,c_2:l, l_b; r_b, r$), let
us denote by $E_+(F)$ the set of edges $\{l c_2, c_1 r, l_b r_b, s
l_b, s r_b\}$.

\begin{proposition}\label{lem:fill-long-aw}
  Let $F$ be a frame ($s:c_1,c_2:l, l_b; r_b, r$) of a large caw in a
  graph $G$.  If a \misp\ $\widehat G$ of $G$ contains none of
  $E_+(F)$, then $F$ remains unchanged in $\widehat G$.
\end{proposition}
\begin{proof}
  We need to show that $\widehat G$ cannot contain any of the other
  missing edges of the frame $F$ not in $E_+(F)$, i.e., $\{s l, s r, l
  r, l r_b, l_b r\}$.  Since $\widehat G$ contains neither $s l_b$ nor
  $l c_2$, it cannot contain edge $s l$; otherwise $s c l_b l s$ is a
  $4$-hole.  A symmetric argument excludes $s r$.  Likewise, the
  nonexistence of $l_b r_b$ and $l c_2$ excludes the edge $l r_b$; and
  a symmetric argument excludes $l_b r$.  If $l r$ is an edge of
  $\widehat G$, then there is a $5$-hole $c l_b l r r_b c$ or $4$-hole
  $l c_1 c_2 r l$ depending on whether $c_1 = c_2$ or not.  These
  contradictions conclude the proposition.
\end{proof}

\begin{figure*}[t]
\setbox4=\vbox{\hsize28pc \noindent\strut
\begin{quote}
  \vspace*{-5mm} \small

  {\bf Algorithm interval-completion($G,k$)}
  \\
  {\sc input}: a graph $G$ and an integer $k$.
  \\
  {\sc output}: a set $E_+$ of at most $k$ edges s.t. $G + E_+$
  is a \misp\ of $G$; or ``NO.''

  0 \hspace*{2ex} {\bf if} $k<0$ {\bf then return} ``NO''; {\bf if}
  $G$ is an interval graph {\bf then return} $\emptyset$; $E_+ =
  \emptyset$;
  \\
  1 \hspace*{2ex} {\bf if} $G$ contains a hole $H$ {\bf then branch}
  using Lem.~\ref{lem:fill-holes};
  \\
  2 \hspace*{2ex} {\bf if} the quotient graph $Q$ defined by maximal
  strong modules of $G$ is a clique {\bf then}
  \\
  \hspace*{7ex} {\bf return interval-edge-deletion}($G[M], k$), where
  $M$ is the only non-clique module;
  \\
  3 \hspace*{2ex} {\bf for each} module $M$ represented by $v\in
  SI(Q)$ {\bf do}
  \\
  \hspace*{7ex} $E_M=$ {\bf interval-completion}($G[M], k$);
  \\
  \hspace*{7ex} {\bf if} $E_M$ is ``NO'' {\bf then return} ``NO'';
  \\
  \hspace*{7ex} $E_+ = E_+\cup E_M$; $k = k - |E_M|$;
  \\
  4 \hspace*{2ex} {\bf if} $Q$ is an interval graph {\bf then return}
  $E_+$;
  \\
  5 \hspace*{2ex} {\bf call decompose}($Q$);
  \\
  6 \hspace*{2ex} {\bf if } a small caw is found {\bf then branch}
  using Lem.~\ref{lem:6-enough-completion};
  \\
  7 \hspace*{2ex} use Lem.~\ref{lem:find-minimal-caw} to find a small
  caw or a minimal frame $F$;
  \\
  8 \hspace*{2ex} {\bf if } a small caw is found {\bf then branch}
  using Lem.~\ref{lem:6-enough-completion};
  \\
  9 \hspace*{2ex} {\bf branch} on adding one edge in $E_+(F)$ or the
  edge set specified in Lem.~\ref{lem:interval-completion-main}.

\end{quote} \vspace*{-6mm} \strut} $$\boxit{\box4}$$
\vspace*{-9mm}
\caption{Algorithm for \textsc{interval completion}.}
\label{fig:alg-interval-completion}
\end{figure*}

Therefore, for a large caw, we only need to consider six cases.  We
are now ready to present the parameterized algorithm for the
\textsc{interval completion} problem, which is given in
Fig.~\ref{fig:alg-interval-completion}.
\begin{theorem}\label{thm:alg-interval-completion}
  Algorithm {\bf interval-completion} solves {\sc interval completion}
  in   $O(6^k\cdot ||G||)$ time.
\end{theorem}
\begin{proof}
  Let us verify first its correctness.  Step 0 gives two trivial exit
  conditions.  Step~1 is ensured by Lem.~\ref{lem:fill-holes}, after
  which $G$ is chordal.  Since it is chordal but not an interval
  graph, all but one maximal strong module is trivial, and hence it
  suffices to solve the only non-trivial module $M$; this justifies
  step 2.  The correctness of step 3 is ensured by
  Thm.~\ref{thm:modules-supergraph+}.  In the case where $Q$ is an
  interval graph, the problem has been already solved before step 4.
  Now that $Q$ is a prime non-interval graph, step~5 applies algorithm
  decompose to it.  Step 6 is straightforward; since the graph is
  chordal, we only need to take care of caws.  Now that $Q$ is chordal
  but not an interval graph, the decomposition $\cal K$ is a
  caterpillar.  Step 7 first calls Lem.~\ref{lem:find-minimal-caw} to
  find a small caw or minimal frame $F$, step 8 follows from
  Lem.~\ref{lem:6-enough-completion}, and step 9 is justified by
  Lems.~\ref{lem:interval-completion-main} and \ref{lem:fill-long-aw}.

  According to Lem.~\ref{lem:fill-holes}, step 1 branches into at most
  $4^{|H|-3}$ sub-instances, in each of which the parameter is
  decreased by ${|H|-3}$.  By Lems.~\ref{lem:6-enough-completion},
  \ref{lem:interval-completion-main}, and \ref{lem:fill-long-aw}, each
  other branching step has at most $6$ sub-instances, in which the
  parameter is decreased by at least one.  Therefore, the total number
  of sub-instances is upper bounded by $6^{k}$.  Noting that each
  sub-instance can be produced in $O(||G||)$ time, the algorithm runs
  in $O(6^k\cdot ||G||)$ time.
\end{proof}

\section{Concluding remarks}\label{sec:remark}
We have presented improved algorithms for \textsc{interval vertex
  deletion} and \textsc{interval completion}, running in time $O(8^k
\cdot ||G||)$ and $O(6^k \cdot ||G||)$ respectively, as well as the
first FPT algorithm for \textsc{interval edge deletion}, which,
nevertheless, takes time $k^{O(k)} \cdot ||G||$.  The dominating steps
are the disposal of large caws and long holes, both of which involve
edge separators in an \emph{interval-graph-like} structure.  We
believe that further algorithmic study of them will improve the
runtime to $O(c^k\cdot ||G||)$ for some constant $c$.  In addition,
the trivial 12-way branching scheme for small caws can be improved
with a similar observation as Prop.~\ref{lem:6-enough-completion}: a
net or a long claw has only six edges, and for each of the others,
there are at most seven groups of edges of which one has to be deleted
(as shown in Fig.~\ref{fig:deleted-edges}).  Therefore, for these
small caws, a 7-way branching suffices.
\begin{figure*}[h]
  \centering
  \begin{subfigure}[b]{0.2\textwidth}
    \centering    \includegraphics{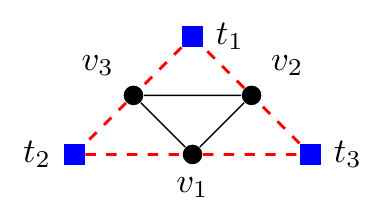} 
  \end{subfigure}
  \quad
  \begin{subfigure}[b]{0.2\textwidth}
    \centering    \includegraphics{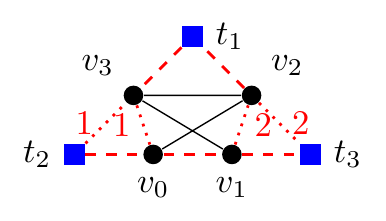}
  \end{subfigure}%  
  \quad
  \begin{subfigure}[b]{0.2\textwidth}
    \centering    \includegraphics{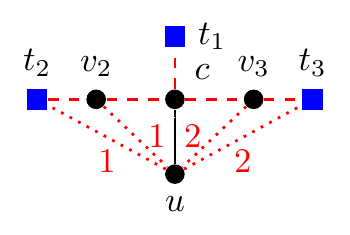}
  \end{subfigure}%  
  \caption{At least one dashed edge or a group of dotted edges is deleted.}
  \label{fig:deleted-edges}
\end{figure*}

We have conducted a comprehensive study of modules in \miib s, \misb
s, and \misp s.  The resulting
Thms.~\ref{thm:modules-induced-subgraph}-\ref{thm:modules-supergraph}
not only are crucial in our algorithms, but find other applications as
well.  Thm.~\ref{thm:modules-supergraph} has been used by Bliznets et
al.~\cite{bliznets-14-interval-completion} to develop the first
subexponential-time parameterized algorithm for {\sc interval
  completion}.  Without a precise analysis of the polynomial factor in
its running time, the authors of
\cite{bliznets-14-interval-completion} pointed out that it must be
high.  An immediate question is then on the possibility of combining
our techniques with their idea to deliver a linear (or low-degree
polynomial) subexponential time algorithm.  Further, since
Thms.~\ref{thm:modules-induced-subgraph}-\ref{thm:modules-supergraph}
hold regardless of $k$, they may also be useful for other purposes,
particularly exact algorithms and kernelization algorithms.  In fact,
the kernelization of these problems has received a lot of attention
and been posed as important open problems by many authors; whether our
observation on modules sheds some light to them is worth further
study.

As we have alluded to, the development of modules and related concepts
is in symbiosis with graph classes.  They were originally proposed by
Gallai \cite{gallai-67-transitive-orientation} in his study of
comparability graphs.  Since it is well known that (\itshape
1\upshape) the complement of an interval graph is a comparability
graph; and (\itshape 2\upshape) a module of a graph is also a module
of its complement graph, it is natural to study modules of interval
graphs.  The relationship between modules and interval graphs has been
built ever since \cite{mohring-85-decomposition}.
Hsu~\cite{hsu-95-recognition-cag} characterized prime interval graphs,
and used this characterization to develop a linear-time recognition
algorithm for interval graphs
\cite{hsu-99-recognizing-interval-graphs}, which is arguably the
simplest among all known recognition algorithms for interval graphs.
A very special kind of modules, clique modules (also known as twin
classes), have been previously used to solve modification problems to
interval graphs \cite{villanger-13-pivd,cao-14-interval-deletion}; in
contrast, the usage of modules in the present work is far more
extensive.

Modules may be used in solving modification problems to other graph
classes as well.  For example, Bessy et
al.~\cite{bessy-10-3-leaf-power-modification} have shown that the
optimum solutions of some graph modification problems preserve maximal
clique modules (also known as critical cliques in literature).  One
can derive from Fomin and Villanger
\cite{fomin-12-subexponential-fill-in} the there exist optimum
solutions to the fill-in problem that preserve all modules.  Another
question of interest is to investigate which other graph modification
problems have this property, and more importantly, how to apply the
similar observations on modules to solve them.  Further, in the
implementation of the meta-approach mentioned in the end of
Section~\ref{sec:major-results}, modules are also useful in detecting
those small forbidden induced subgraphs that are not prime.  The
elimination of 4-holes in the present paper, for example, are
completely based on modules.  In a follow-up work
\cite{cao-14-edge-deletion}, we use modules to facilitate the
detection of other small forbidden induced subgraphs, yielding
improved algorithms for some edge deletion problems.  Other
algorithmic usage of modules in graph problems can be found in surveys
\cite{mohring-85-decomposition,habib-10-survey-md}.

We have correlated interval graphs and normal Helly circular-arc
graphs and used this nontrivial relationship in characterizing \lig
s.\footnote{In an earlier version of this manuscript, we also used the
  resulting Thm.~\ref{thm:decompose-lig} to develop linear-time
  parameterized algorithms for modification problems to unit interval
  graphs (i.e., interval graphs representable by intervals of the same
  length).  A simpler approach was discovered afterward, which applies
  to editing problem as well \cite{cao-14-unit-interval-editing},
  whereupon we drop these results from the current version.}
Interestingly, it can also be used in the other way.  In a follow-up
paper \cite{cao-14-recognizing-nhcag}, we used a similar definition of
the auxiliary graph and pertinent observations to recognize normal
Helly circular-arc graphs.  Our algorithm runs in linear time and is
able to detect a forbidden induced subgraph if the answer is ``NO.''
As pointed out in Section~\ref{sec:nhcag}, normal Helly circular-arc
graphs and \lig s are incomparable to each other.  For example, the
most used subgraphs in $\cF_{LI}$ in
Sections~\ref{sec:forbidden-subgraphs} and \ref{sec:olive-ring} are
short holes, which, however, are not forbidden in normal Helly
circular-arc graphs.  On the opposite, every hole in a normal Helly
circular-arc graph is dominating, but a \lig\ might not have this
property.  Therefore, the structural analysis of
\cite{cao-14-recognizing-nhcag} is significantly different from the
current paper.  More combinatorial and algorithmic applications of
this observation would be interesting topics for future work.  Also
worthwhile is further study of \lig s.  For example, from
Prop.~\ref{lem:lig-and-modules} and Thm.~\ref{thm:decompose-lig} we
can derive the following property of a connected \lig\ $G$ (which is
not required to be prime):
\begin{itemize}
\item if $G$ is not chordal, then there is a shortest hole $H$ such
  that $V(H)$ is a dominating set;
\item otherwise, there is a pair of vertices $u,v$ and a shortest
  \stpath{u}{v} $P$ such that $V(P)$ is a dominating set.
\end{itemize}
Moreover, such a dominating hole or path can be found in linear time.
In fact, a chordal \lig\ is a diametral path graph
\cite{deogun-95-diametral-path-graphs}, a graph class whose definition
was inspired by dominating pairs in at-free graphs
(\cite{corneil-97-at-free}), and the non-chordal ones should be
compared with non-chordal circular-arc graphs, where every hole is
trivially dominating.

The present paper has been focused on modification problems that allow
only a single type of operations.  This constraint is not inherent.
Marx~\cite{marx-10-chordal-deletion} developed an FPT algorithm for
the chordal deletion problem that allows both vertex and edge
deletions (parameterized by the sum $k_1+k_2$ of modifications,
including $k_1$ vertex deletions and $k_2$ edge deletions).  Our
techniques can be revised to work on a similar problem to interval
graphs, but the analysis would become way too complicated to be
discussed here.  More widely studied is the edge editing problem
\cite{natanzon-01-edge-modification,sharan-02-thesis}.  Burzyn et
al.~\cite{burzyn-06-NPC-edge-modification} proved that it is NP-hard
to obtain an interval graph by the minimum number of edge
modifications.  The fixed-parameter tractability of this problem
(parameterized by the total number of modifications), to the best of
our knowledge, is still open.  One could even allow simultaneously all
three types of operations.  In fact,
Cai~\cite{cai-96-hereditary-graph-modification} has formulated the
following general modification problem on a hereditary graph class
\cG: given a graph $G$ and nonnegative integers $k_1$, $k_2$, and
$k_3$, the task is to transform $G$ into a graph in \cG\ by at most
$k_1$ vertex deletions, $k_2$ edge deletions, and $k_3$ edge
additions.  On this formulation two remarks are in order.  First, it
does not make sense to impose a combined quota on the total number of
modifications, as it is then trivially degenerated to the vertex
deletion problem (see the remarks in the end of
Section~\ref{sec:gmp-fpt}).  Second, this formulation generalizes all
three problems studied in this paper as well the two problems
mentioned above.  Another natural extension to the edge deletion
problem and completion problem is the sandwich problem, which, given
two graphs $G_1$ and $G_2$ such that $G_1\subseteq G_2$, asks for a
graph $G\in\cG$ such that $G_1\subseteq G\subseteq G_2$
\cite{golumbic-93-temporal-reasoning,golumbic-94-interval-sandwich,golumbic-95-graph-sandwich}.
On chordal graphs, both the general modification problem and the
sandwich problem (parameterized by $||G|| - ||G_1||$) are known to be
FPT \cite{cao-14-chordal-editing,fomin-12-subexponential-fill-in}.  It
is natural to investigate the their fixed-parameter tractability on
interval graphs; some partial results on the interval sandwich problem
were reported in \cite{kaplan-99-bounded-degree-interval-sandwich,
  bodlaender-11-exact-intervalizing-colored-graphs}.  However, it is
not immediately clear whether there always exists a module-preserving
optimum solution, and how to dispose of large caws in a local way.

It is known that finding minimal non-interval subgraphs is closely
related to finding Tucker sub-matrices, i.e., minimal sub-matrices
that invalidate the consecutive-ones property (C1P).  Again, a Tucker
sub-matrex can be found in linear time
\cite{lindzey-13-find-forbidden-subgraphs}~\cite[Section
3.3]{dom-08-dissertation}.  But as finding Tucker sub-matrices is as
least as hard as finding non-interval subgraphs, there is little hope
to find a minimum one in the same time.  Can our techniques be applied
to solve modification problems related to C1P
(\cite{dom-08-dissertation,narayanaswamy-13-d-cos-r})?

\paragraph{Acknowledgment.}  I am grateful to Paul Seymour and Michel
Habib for helpful discussion on the complexity of finding shortest
holes.  I am indebted to Sylvain Guillemot and Yang Liu for careful
reading of an early version of this manuscript and their penetrating
criticism.  
% I am deeply indebted to
% \href{http://www.youtube.com/watch?v=KOv5aEHbjeU&list=PLHYQROpSSJpGfREcMWKIWpO8xwftHlYrF}{Concert
%   YY - Showcase Of Wyman Wong} for its company during the writing of
% the whole paper.
 
%http://v.youku.com/v_show/id_XNTA0ODkwNjk2.html

{
\small %\printbibliography
\bibliographystyle{plainurl}
\bibliography{../journal,../main}
}
\end{document}